\documentclass[11pt]{article}
\usepackage[margin=1in]{geometry}
\usepackage{amsmath,amssymb}
\usepackage{amsthm}
\usepackage{xspace}
\usepackage{graphicx}
\usepackage{paralist}
\usepackage{wrapfig}
\usepackage{subfigure}
\usepackage{lineno}

\graphicspath{{fig/}}

\newcommand{\nocontentsline}[3]{}
\newcommand{\tocless}[2]{\bgroup\let\addcontentsline=\nocontentsline#1{#2}\egroup}

\renewcommand{\H}{\mathcal{H}}
\newcommand{\G}{\mathcal{G}}
\newcommand{\rskel}{\mathfrak{R}} %name of an R-skeleton
\newcommand{\pskel}{\mathfrak{P}} %name of a P-skeleton
\newcommand{\att}{\text{att}} %set of attachments
\newcommand{\con}[1]{K_{#1}} %the conflict graph of cycle #1
\newcommand{\conr}[1]{K^-_{#1}} %the reduced conflict graph of cycle #1
\newcommand{\ach}[1]{\text{Ach}_{#1}}%alternating chain with #1 bridges
\newcommand{\rskelemb}{\rskel^+}%a chosen embedding of an R-skeleton
\newcommand{\rskelnot}{\rskel^-}%a not-chosen embedding of an R-skeleton
\newcommand{\T}{\ensuremath{\mathcal{T}}}
\newcommand{\sskel}{\mathfrak{S}}% name of a generic skeleton

\newtheorem{theorem}{Theorem}
\newtheorem{lemma}[theorem]{Lemma}
\newtheorem{definition}[theorem]{Definition}
\newtheorem{proposition}[theorem]{Proposition}
\newtheorem{claim}{Claim}
\newtheorem{observation}{Observation}

\title{A Kuratowski-Type Theorem for \\Planarity of Partially Embedded
  Graphs\thanks{Supported by the GRADR -- EUROGIGA project no.
GIG/11/E023}}

\newcommand{\peplong}{\textsc{PartiallyEmbeddedPlanarity}}
\newcommand{\pep}{\textsc{Pep}\xspace}
\newcommand{\peg}{\textsc{Peg}\xspace}
\newcommand{\pegs}{\textsc{Peg}s\xspace}

\begin{document}

%\numberofauthors{3}

\author{V\'\i t Jel\'\i nek\thanks{Computer Science Institute, Charles
    University, Prague, Czech Republic, email:
    \texttt{jelinek@iuuk.mff.cuni.cz}}, %
  Jan Kratochv\'\i l\thanks{Department of Applied Mathematics and
    Institute for Theoretical Computer Science, Charles University,
    Prague, Czech Republic, email: \texttt{honza@kam.mff.cuni.cz}}, %
  Ignaz Rutter\thanks{Institute of Theoretical Informatics, Karlsruhe
    Institute of Technology (KIT), Karlsruhe, Germany, \mbox{email:
    \texttt{rutter@kit.edu}}}}

\date{}

\maketitle

\begin{abstract}
A partially embedded graph (or \peg) is a triple $(G,H,\mathcal{H})$, where $G$
is a graph, $H$ is a subgraph of $G$, and $\mathcal{H}$ is a planar embedding
of~$H$. We say that a \peg $(G,H,\mathcal{H})$ is planar if
the graph $G$ has a planar embedding that extends the embedding~$\mathcal{H}$.

We introduce a containment relation of \pegs analogous to graph
minor containment, and characterize the minimal non-planar \pegs
with respect to this relation. We show that all the minimal
non-planar \pegs except for finitely many belong to a single easily
recognizable and %JK
explicitly described infinite family.  We also describe a more complicated
containment relation which only has a finite number of minimal non-planar \pegs.

Furthermore,
by extending an existing planarity test for \pegs, we obtain a
polynomial-time algorithm which, for a given \peg, either produces a
planar embedding or identifies an obstruction.

\end{abstract}

\maketitle

\section{Introduction}
\label{sec:intro}

A \emph{partially embedded graph} (\peg) is a triple $(G,H,\cal H)$,
where $G$ is a graph, $H$ is a subgraph of~$G$, and $\H$ is a planar
embedding of~$H$.  The problem \peplong (\pep) asks whether a \peg
$(G,H,\H)$ admits a planar (non-crossing) embedding of $G$ whose
restriction to $H$ is $\H$.  In this case we say that the \peg
$(G,H,\H)$ is \emph{planar}.  Despite of this being a very natural
generalization of planarity, this approach has been considered only
recently~\cite{abfjk-tppeg-10}.  It should be mentioned that all
previous planarity testing algorithms have been of little use for
\pep, as they all allow flipping of already drawn parts of the graph,
and thus are not suitable for preserving an embedding of a given
subgraph.

%JK More changes in the coming paragraph
It is shown in~\cite{abfjk-tppeg-10} that planarity of \pegs can be
tested in linear time.  In this paper we complement the algorithm
in~\cite{abfjk-tppeg-10} by a study of the combinatorial aspects of
this question.  In particular, we provide a complete characterization
of planar \pegs via a small set of forbidden substructures, similarly
to the celebrated Kuratowski theorem~\cite{Kuratowski}, which
characterizes planarity via the forbidden subdivisions of $K_5$ and
$K_{3,3}$, and the closely related theorem of Wagner~\cite{Wagner},
which characterizes planarity via forbidden $K_5$ and $K_{3,3}$
minors.  Our characterization can then be used to modify the existing
planarity test for partially embedded graphs into a certifying
algorithm that either finds a solution or finds a certificate, i.e., a
forbidden substructure, that shows that the instance is not planar.

Understanding the forbidden substructures may be particularly
beneficial in studying the problem \emph{simultaneous embedding with
  fixed edges}, or SEFE for short, which asks whether two graphs $G_1$
and $G_2$ on the same vertex set $V$ admit two drawings $\Gamma_1$ and
$\Gamma_2$ of $G_1$ and $G_2$, respectively, such that (i) all
vertices are mapped to the same point in $\Gamma_1$ and $\Gamma_2$,
(ii) each drawing $\Gamma_i$ is a planar drawing of $G_i$ for $i=1,2$,
and (iii) edges common to $G_1$ and $G_2$ are represented by the same
Jordan curve in $\Gamma_1$ and $\Gamma_2$.  J\"unger and
Schulz~\cite{Junger} show that two graphs admit a SEFE if and only if
they admit planar embeddings that coincide on the intersection graph.
In this sense, our obstructions give an understanding of which
configurations should be avoided when looking for an embedding of the
intersection graph.

For the purposes of our characterization, we introduce a set of operations,
called \emph{\peg-minor operations}, that preserve the planarity of
\pegs. Note that it is not possible to use the usual minor operations, as
sometimes, when contracting an edge of $G$ not belonging to $H$, it is not clear
how to modify the embedding of~$H$.  Our minor-like operations are defined
in Section~\ref{sec:preliminaries}.

Our goal is to identify all minimal non-planar \pegs in the minor-like
order determined by our operations; such \pegs are referred to as {\em
  obstructions}.  Our main theorem says that all obstructions are
depicted in Fig.~\ref{fig:obstructions} or belong to a well
described infinite class of so called \emph{alternating chains} (the
somewhat technical definition is postponed to
Section~\ref{sec:preliminaries}).  It can be verified that each of
them is indeed an obstruction, i.e., it is not planar, but
applying any of the \peg-minor operations results in a planar~\peg.

\newcommand{\obstrscale}{.9}

\begin{figure}[htbp]
  \centering
\iffalse % ENABLE for 5 rows of 5 figures, each
  \begin{tabular}{|c|c|c|c|c|}
    \hline
    \includegraphics[scale=\obstrscale]{k5} &
    \includegraphics[scale=\obstrscale]{k33} &
    \includegraphics[scale=\obstrscale]{1} &
    \includegraphics[scale=\obstrscale]{2}&
    \includegraphics[scale=\obstrscale]{3} \\ \hline
    \includegraphics[scale=\obstrscale]{4}&
    \includegraphics[scale=\obstrscale]{5}&
    \includegraphics[scale=\obstrscale]{6}&
    \includegraphics[scale=\obstrscale]{7}&
    \includegraphics[scale=\obstrscale]{8} \\ \hline
    \includegraphics[scale=\obstrscale]{9}&
    \includegraphics[scale=\obstrscale]{10}&
    \includegraphics[scale=\obstrscale]{11}&
    \includegraphics[scale=\obstrscale]{12}&
    \includegraphics[scale=\obstrscale]{13} \\ \hline
    \includegraphics[scale=\obstrscale]{14} &
    \includegraphics[scale=\obstrscale]{15} &
    \includegraphics[scale=\obstrscale]{16} &
    \includegraphics[scale=\obstrscale]{17}&
    \includegraphics[scale=\obstrscale]{18} \\ \hline
    \includegraphics[scale=\obstrscale]{19}&
    \includegraphics[scale=\obstrscale]{20}&
    \includegraphics[scale=\obstrscale]{21}&
    \includegraphics[scale=\obstrscale]{22}&\\\hline
\end{tabular}
\fi
%\iffalse % ENABLE for 4 rows with 6 figures, each
  \begin{tabular}{|c|c|c|}
    \hline
    \includegraphics[scale=\obstrscale]{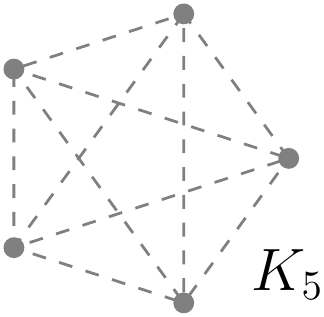} &
    \includegraphics[scale=\obstrscale]{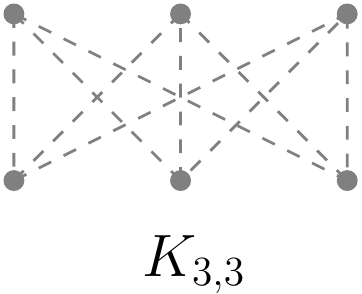} &
    \includegraphics[scale=\obstrscale]{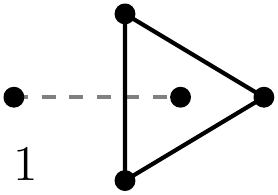} \\\hline
    \includegraphics[scale=\obstrscale]{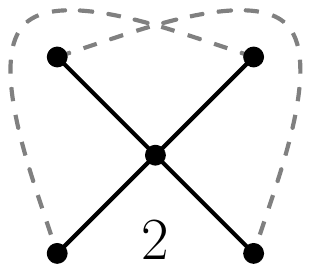}&
    \includegraphics[scale=\obstrscale]{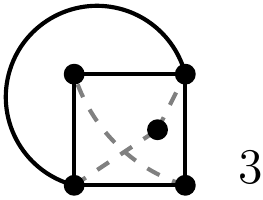}&
    \includegraphics[scale=\obstrscale]{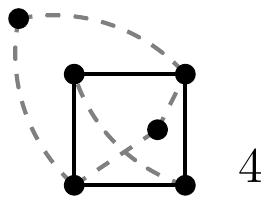} \\ \hline
    \includegraphics[scale=\obstrscale]{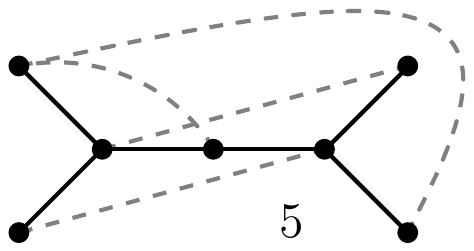}&
    \includegraphics[scale=\obstrscale]{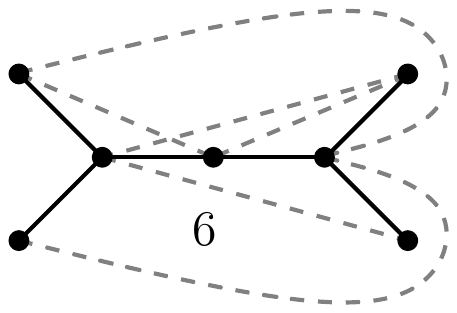}&
    \includegraphics[scale=\obstrscale]{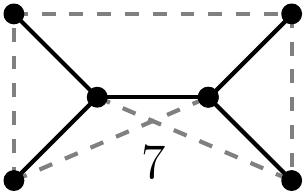}\\\hline
    \includegraphics[scale=\obstrscale]{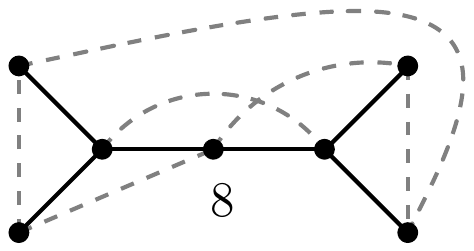}&
    \includegraphics[scale=\obstrscale]{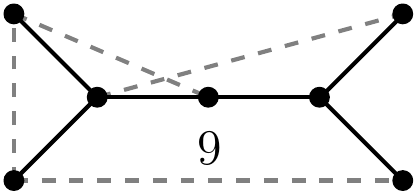}&
    \includegraphics[scale=\obstrscale]{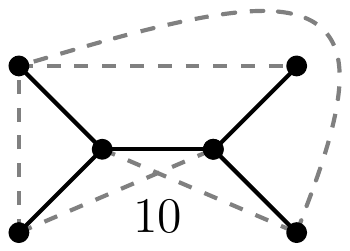} \\ \hline
    \includegraphics[scale=\obstrscale]{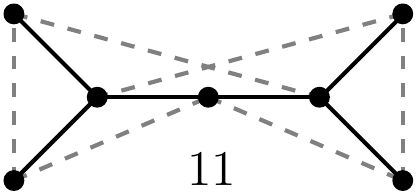}&
    \includegraphics[scale=\obstrscale]{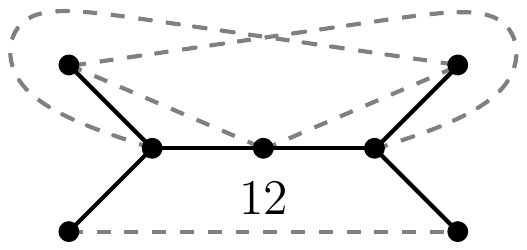}&
    \includegraphics[scale=\obstrscale]{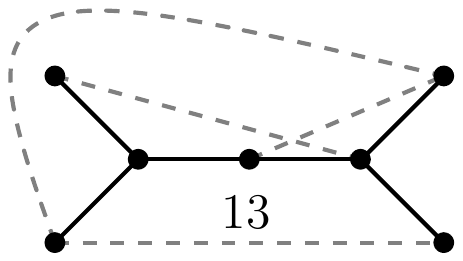}\\\hline
    \includegraphics[scale=\obstrscale]{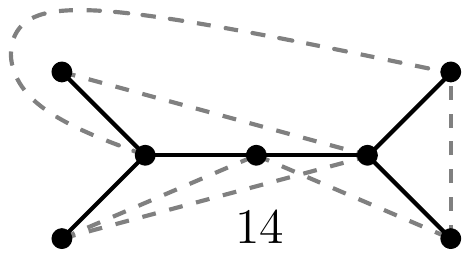} &
    \includegraphics[scale=\obstrscale]{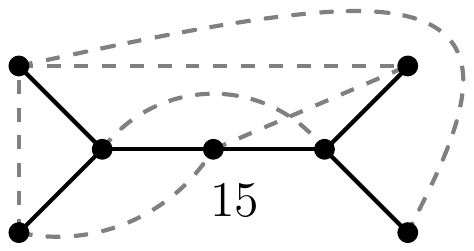} &
    \includegraphics[scale=\obstrscale]{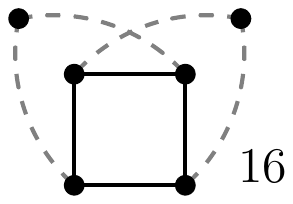} \\ \hline
    \includegraphics[scale=\obstrscale]{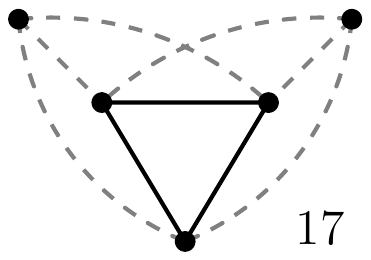}&
    \includegraphics[scale=\obstrscale]{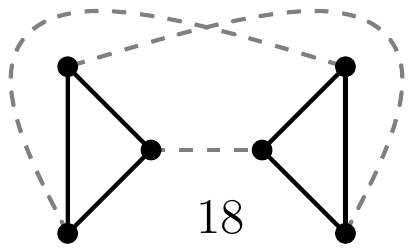}&
    \includegraphics[scale=\obstrscale]{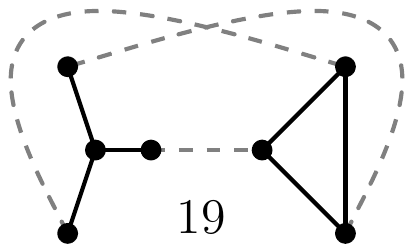}\\\hline
    \includegraphics[scale=\obstrscale]{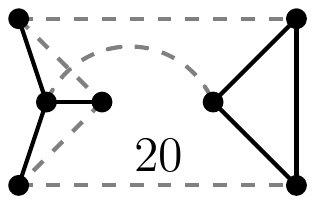}&
    \includegraphics[scale=\obstrscale]{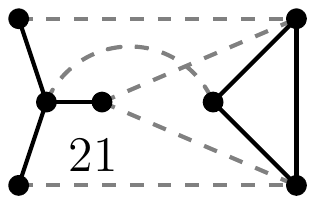}&
    \includegraphics[scale=\obstrscale]{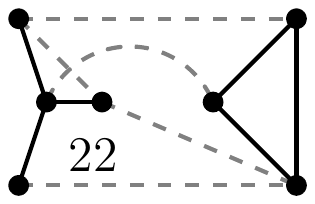}\\\hline
\end{tabular}
%\fi
\caption{The obstructions not equal to the $k$-fold
  alternating chains for $k \ge 4$. The black solid edges belong to
  $H$, the light dashed edges to $G$, but not to~$H$. All the vertices
  belong to both $G$ and $H$, except for $K_5$ and $K_{3,3}$, where
  $H$ is empty.}
  \label{fig:obstructions}
\end{figure}

We say that a \peg \emph{avoids} a \peg\ $X$ if it does not contain
$X$ as a \peg-minor.  Furthermore, we say that a \peg is
\emph{obstruction-free} if it avoids all \pegs of
Fig.~\ref{fig:obstructions} and all alternating chains of lengths
$k\ge 4$.  Then our main theorem can be expressed as follows.

\begin{theorem}\label{thm:main}
  A \peg is planar if and only if it is obstruc\-tion-free.
\end{theorem}

Since our \peg-minor operations preserve planarity, and since all the
listed obstructions are non-planar, any planar \peg is
obstruction-free.  The main task is to prove that an obstruction-free
\peg is planar.

%In this extended abstract, we prove the theorem for
%\pegs $(G,H,\H)$ where $G$ is 2-connected.

Having identified the obstructions, a natural question is if the
\peg-planarity testing algorithm of~\cite{abfjk-tppeg-10} can be
extended so that it provides an obstruction if the input is
non-planar.  It is indeed so.

\begin{theorem}\label{thm:alg}
  There is a polynomial-time algorithm that for an input \peg\
  $(G,H,\H)$ either constructs a planar embedding of $G$ extending
  $\H$, or provides a certificate of non-planarity, i.e., identifies
  an obstruction present in $(G,H,\H)$ as a \peg-minor.
\end{theorem}

The paper is organized as follows.  In
Section~\ref{sec:preliminaries}, we first recall some basic
definitions and results on \pegs and their planarity, and then define
the \peg-minor order and the alternating chain obstructions.  In
Section~\ref{sec:biconnected}, we show that the main theorem holds for
instances where $G$ is biconnected.  We extend the main theorem to
general (not necessarily biconnected) \pegs in
Section~\ref{sec:disconnected}.  In Section~\ref{sec:other}, we
present possible strengthening of our \peg-minor relations, and show
that when more complicated reduction rules are allowed, the modified
\peg-minor order has only finitely many non-planar \pegs.  In
Section~\ref{sec:conclusion} we briefly provide an argument for
Theorem~\ref{thm:alg} and then conclude with some open problems.

\section{Preliminaries and Notation}
\label{sec:preliminaries}

\paragraph{Embeddings}
A \emph{drawing} of a graph is a mapping of each vertex to a distinct
point in the plane and of each edge to a simple Jordan curve that
connects its endpoints.  A drawing is \emph{planar} if the curves
representing the edges intersect only in common endpoints.  A graph is
\emph{planar} if it admits a planar drawing.  Such a planar drawing
determines a subdivision of the plane into connected regions, called
\emph{faces}, and a circular ordering of the edges incident to each
vertex, called \emph{rotation scheme}.  Traversing the border of a
face $F$ in such a way that the face is to the left yields a set of
circular lists of vertices, the \emph{boundary} of $F$.  Note that the
boundary of a face is not necessarily connected if the graph is not
connected and that vertices can be visited several times if the graph
is not biconnected.  The boundary of a face $F$ can uniquely be
decomposed into a set of simple edge-disjoint cycles, bridges (i.e.,
edges that are not part of a cycle) and isolated vertices.  We orient
these cycles so that $F$ is to their left to obtain the \emph{facial
  cycles} of~$F$.

Two drawings are \emph{topologically equivalent} if they have the same
rotation scheme and, for each facial cycle, the vertices to its left
are the same in both drawings.  A \emph{planar embedding} is an
equivalence class of planar drawings.  Let $\G$ be a planar embedding
of $G$ and let $H$ be a subgraph of $G$.  The \emph{restriction} of
$\G$ to $H$ is the embedding of $H$ that is obtained from $\G$ by
considering only the vertices and the edges of~$H$.  We say that $\G$
is an \emph{extension} of a planar embedding $\H$ of $H$ if the
restriction of $\G$ to $H$ is~$\H$.

\paragraph{Connectivity and SPQR-trees}
A graph is \emph{connected} if any pair of its vertices is connected
by a path.  A maximal connected subgraph of a graph $G$ is a
\emph{connected component} of~$G$.  A \emph{cut-vertex} is a vertex
$x\in V(G)$ such that $G-x$ has more components than $G$.  A connected
graph with at least three vertices is \emph{2-connected} (or
\emph{biconnected}) if it has no cut-vertex.  In a biconnected graph
$G$, a \emph{separating pair} is a pair of vertices $\{x,y\}$ such
that $G-x-y$ has more components than~$G$.  A biconnected graph with
at least four vertices is \emph{3-connected} if it has no separating
pair.  We say that a \peg $(G,H,\H)$ is \emph{connected},
\emph{biconnected} and \emph{3-connected} if $G$ is connected,
biconnected and 3-connected, respectively.  An edge of a graph $G$ is
sometimes referred to as a \emph{$G$-edge}, and a path in $G$ is a
\emph{$G$-path}.

A connected graph can be decomposed into its maximal biconnected
subgraphs, called \emph{blocks}.  Each edge of a graph belongs to
exactly one block, only cut-vertices are shared between different
blocks.  This gives rise to the block-cutvertex tree of a connected
graph~$G$, whose nodes are the blocks and cut-vertices of~$G$, and
whose edges connect cut-vertices to blocks they belong to.

The planar embeddings of a 2-connected graph can be represented by the
SPQR-tree, which is a data structure introduced by Di Battista and
Tamassia~\cite{dt-omtc-96,dt-opt-96}.  A more detailed description of
the SPQR-tree can be found in the
literature~\cite{dt-omtc-96,dt-opt-96,gm-lti-00,Mutzel}.  Here we just
give a sketch and some notation.

The SPQR-tree $\T$ of a 2-connected graph $G$ is an unrooted tree that
has four different types of nodes, namely S-,P-,Q- and R-nodes.  The
Q-nodes are the leaves of $\T$, and they correspond to edges of~$G$.
Each internal node $\mu$ of $\T$ has an associated biconnected
multigraph $\sskel$, its \emph{skeleton}, which can be seen as a
simplified version of the graph~$G$. The subtrees of~$\mu$ in~$\T$
when~$\T$ is rooted at~$\mu$ determine a decomposition of $G$ into
edge-disjoint subgraphs $G_1,\dotsc,G_k$, each of which is connected
and shares exactly two vertices $u_i, v_i$ with the rest of the graph
$G$.  Each $G_i$ is represented in the skeleton of $\mu$ by an edge
$e_i$ connecting $u_i$ and $v_i$.  We say that $G_i$ is the
\emph{pertinent graph} of the edge $e_i$.  We also say that the
skeleton edge $e_i$ \emph{contains} a vertex $v$ or an edge $e$ of
$G$, or that $v$ and $e$ \emph{project} into $e_i$, if $v$ or $e$
belong to the pertinent graph $G_i$ of~$e_i$. For a subgraph $G'$ of
$G$, we say that $G'$ \emph{intersects} a skeleton edge $e_i$, if at
least one edge of $G'$ belongs to~$G_i$.

The skeleton of an S-node is a cycle of length $k\ge 3$, the skeleton
of a P-node consists of $k\ge 3 $ parallel edges, and the skeleton of
an R-node is a 3-connected planar graph.  The SPQR-tree of a planar
2-connected graph $G$ represents all planar embeddings of $G$ in the
sense that choosing planar embeddings for all skeletons of $\T$
corresponds to choosing a planar embedding of $G$ and vice versa.

Suppose that $e=uv$ is an edge of the skeleton of a node $\mu$ of an
SPQR-tree of a biconnected graph $G$, and let $G_e$ be the pertinent
graph of~$e$. The graph $G_e$ satisfies some additional restrictions
depending on the type of $\mu$: if $\mu$ is an S-node, then $G_e$ is
biconnected, and if $\mu$ is a P-node, then either $G_e$ is a single
edge $uv$ or $G_e-\{u,v\}$ is connected. Regardless of the type of
$\mu$, every cut-vertex in $G_e$ separates $u$ from $v$, otherwise $G$
would not be biconnected.

\paragraph{\peg-minor operations}

We first introduce a set of operations that preserve planarity when
applied to a \peg~$I=(G,H,\H)$.  The set of operations is chosen so
that the resulting instance~$I'=(G',H',\H')$ is again a \peg (in
particular, $H'$ is a subgraph of~$G'$ and $\H'$ is a planar embedding
of~$H'$).  It is not possible to use the usual minor operations, as
sometimes, when contracting an edge of $G-H$, the embedding of the
modified graph~$H$ is not unique and some of the possible embeddings
lead to planar \pegs, while some do not. This happens, e.g., when a 
contraction of a $G$-edge creates a new cycle of $H$-edges, in which 
case it is not clear on which side of this cycle the remaining components 
of $H$ should be embedded (see Fig.~\ref{fig:contract}).

\begin{figure}
\centering
\includegraphics{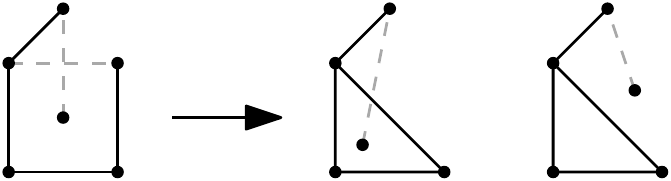}
\caption{An example of a planar \peg (left) in which a contraction of a $G$-edge may result in two distinct \pegs, one of which is non-planar}\label{fig:contract}
\end{figure}

We will consider seven minor-like operations, of which the first
five are straightforward.
\begin{enumerate}
\itemsep=-.5ex
\item Vertex removal: Remove from $G$ and $H$ a vertex $v \in V(G)$
  with all its incident edges.
\item Edge removal: Remove from $G$ and $H$ an edge $e \in E(G)$.
\item Vertex relaxation: For a vertex $v \in H$ remove $v$ and all its
  incident edges from $H$, but keep them in $G$.  In other words,
  vertex $v$ no longer has a prescribed embedding.
\item Edge relaxation: Remove an edge $e \in E(H)$ from $H$, but keep
  it in $G$.
\item $H$-edge contraction: Contract an edge $e \in E(H)$ in both $G$
  and $H$, update $\H$ accordingly.
\newcounter{opcounter}
\setcounter{opcounter}{\value{enumi}}
\end{enumerate}

The contraction of $G$-edges is tricky, as we have to care about two
things.  First, we have to take care that the modified subgraph~$H'$
remains planar and second, even if it remains planar, we do not want to
create a new cycle~$C$ in~$H$ as in this case the relative positions
of the connected components of~$H$ with respect to this cycle may
not be uniquely determined.  We therefore have special requirements
for the $G$-edges that may be contracted and we distinguish two
types, one of which trivially ensures the above two conditions and
one that explicitly ensures them.

\begin{enumerate}
\itemsep=-.5ex
\setcounter{enumi}{\value{opcounter}}
\item Simple $G$-edge contraction: Assume that $e=uv$ is an edge
  of $G$, such that at least one of the two vertices $u$ and $v$ does
  not belong to~$H$.  Contract~$e$ in~$G$, and leave $H$ and $\H$
  unchanged.

\item Complicated $G$-edge contraction: Assume that $e=uv$ is an edge
  of~$G$, such that $u$ and $v$ belong to distinct components of~$H$,
  but share a common face of~$\H$.  Assume further that both $u$ and
  $v$ have degree at most~1 in~$H$.  This implies that we may uniquely
  extend $\H$ to an embedding $\H^+$ of the graph $H^+$ that is
  obtained from $H$ by adding the edge $uv$.  Afterwards we perform an
  $H$-edge contraction of the edge~$uv$ to obtain the new \peg.
\end{enumerate}

If a contraction produces multiple edges, we only preserve a single
edge from each such set of multiple edges, so that $G$ and $H$ remain
simple.  Note that the resulting embedding $\H$ may depend on which
edge we decide to preserve.

Let $(G,H,\H)$ be a \peg and let $(G',H',\H')$ be the result of one of
the above operations on $(G,H,\H)$.
The extra conditions on $G$-edge contractions ensure that the
embedding $\H'$ is uniquely determined from
the embedding of~$H$ after such contraction.  The conditions on vertex degrees in $H$ ensure that the rotation scheme of the $H'$-edges around the vertex created by the contraction is unique.  In the complicated $G$-edge contraction, the
requirement that the endpoints need to lie in distinct connected
components of~$H$ that share a face ensures that the contraction does
not create a new cycle in $H'$ and that~$H'$ has a
unique planar embedding induced by~$\H$.

It is not hard to see that an
embedding $\G$ of $G$ that extends $\H$ can be transformed into an
embedding~$\G'$ of $G'$ that extends~$\H'$.  Therefore, all the above
operations preserve planarity of~\pegs.  If a \peg $A$ can be obtained
from a \peg $B$ by applying a sequence of the above operations, we say
that $A$ is a \emph{\peg-minor} of $B$ or that $B$ \emph{contains $A$
  as a \peg-minor}.

\paragraph{Alternating chains}

Apart from the obstructions in Fig.~\ref{fig:obstructions}, there is
an infinite family of obstructions, which we call \emph{the
  alternating chains}.  To describe them, we need some terminology.
Let $C$ be a cycle of length at least four, and let $u$, $v$, $x$ and
$y$ be four distinct vertices of $C$. We say that the pair of vertices
$\{u,v\}$ \emph{alternates} with the pair $\{x,y\}$ on $C$, if $u$ and
$v$ belong to distinct components of $C-x-y$.

Intuitively, an alternating chain consists of a cycle $C$ of $H$ and a
sequence of internally disjoint paths $P_1,\dots,P_k$ of which only
the endpoints belong to~$C$, such that for each $i=1,\dots,k-1$, the
endpoints of $P_i$ alternate with the endpoints of $P_{i+1}$ on $C$,
and no other pair of paths has alternating endpoints.  Now assume that
$P_1$ contains a vertex that is prescribed inside~$C$.  Due to the
fact that the endpoints of consecutive paths alternate this implies
that all $P_i$ with $i$ odd must be embedded inside $C$, while all
$P_i$ with $i$ even must be embedded outside.  A $k$-fold alternating
chain is such that the last path $P_k$ is prescribed in a way that
contradicts this, i.e., it is prescribed inside $C$ if $k$ is even and
outside, if $k$ is odd.  Generally it is sufficient to have paths of
length~1 for $P_2,\dots,P_{k-1}$ and to have a single vertex (for the
prescription) in each of $P_1$ and $P_k$.  We now give a precise
definition.

\begin{figure}[tb]
\hfil \includegraphics{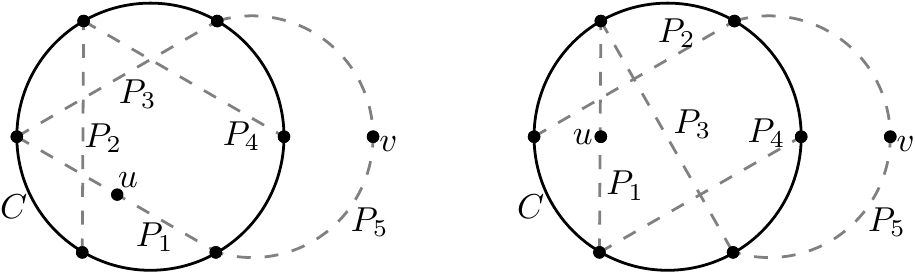}
\caption{Two non-isomorphic 5-fold alternating chains.}\label{fig:ach}
\end{figure}

Let $k\ge 3$ be an integer. A \emph{$k$-fold alternating chain} is a
\peg $(G,H,\H)$ of the following form:
\begin{itemize}
\itemsep=-.5ex
\item The graph $H$ consists of a cycle $C$ of length $k+1$ and two
  isolated vertices $u$ and~$v$.  If $k$ is odd, then~$u$ and~$v$ are
  embedded on opposite sides of $C$ in $\H$, otherwise they are
  embedded on the same side.
\item The graph $G$ has the same vertex set as $H$, and the edges of
  $G$ that do not belong to $H$ form $k$ edge-disjoint paths $P_1,
  \dotsc, P_k$, whose endpoints belong to~$C$.  The path $P_1$ has two
  edges and contains $u$ as its middle vertex, the path $P_k$ has two
  edges and contains $v$ as its middle vertex, and all the other paths
  have only one edge.
\item The endpoints of the path $P_i$ alternate with the endpoints of
  the path $P_j$ on $C$ if and only if $j=i+1$ or $i=j+1$.
\item All the vertices of $C$ have degree~4 in $G$ (i.e., each of them
  is a common endpoint of two of the paths $P_i$), with the exception
  of two vertices of $C$ that have degree three. One of these two
  vertices is an endpoint of $P_2$, and the other is an endpoint of
  $P_{k-1}$.
\end{itemize}

Let $\ach{k}$ denote the set of $k$-fold alternating chains.  It can
be checked that for each $k\ge 4$, the elements of $\ach{k}$ are
obstructions; see Lemma~\ref{lem-ach}.  Obstruction~4 from
Fig.~\ref{fig:obstructions} is actually the unique member of
$\ach{3}$, and is an obstruction as well.  However, we prefer
to present it separately as an `exceptional' obstruction, because we
often need to refer to it explicitly.  Note that for $k\ge 5$ we may
have more than one non-isomorphic $k$-fold chain; see
Fig.~\ref{fig:ach}.

\section{Biconnected Pegs}
\label{sec:biconnected}

In this section we prove Theorem~\ref{thm:main} for biconnected \pegs.
We first recall a characterization of biconnected planar \pegs via
SPQR-trees.

\begin{definition}
  Let $(G,H,\H)$ be a biconnected \peg.

  A planar embedding of the skeleton of a node of the SPQR-tree of~$G$
  is \emph{edge-compatible with $\H$} if, for every vertex $x$ of the
  skeleton and for every three edges of $H$ incident to $x$ that
  project to different edges of the skeleton, their order
  determined by the embedding of the skeleton is the same as their
  order around $x$ in $\H$.

  A planar embedding of the skeleton $\sskel$ of a node~$\mu$ of the
  SPQR-tree of~$G$ is \emph{cycle-compatible with $\H$} if, for every
  facial cycle $\vec{C}$ of $\H$ whose edges project to a simple cycle
  $\vec{C'}$ in $\sskel$, all the vertices of $\sskel$ that lie to the
  left of $\vec{C}$ and all the skeleton edges not belonging
  to~$\vec{C}$ that contain vertices that lie to the left of $\vec{C}$
  in $\H$ are embedded to the left of $\vec{C'}$; and analogously for
  the vertices to the right of $\vec{C}$.

  A planar embedding of a skeleton of a node of the SPQR-tree of~$G$
  is \emph{compatible} if it is both edge- and cycle-compatible.
\end{definition}

Angelini et al. showed that a biconnected \peg is planar if and only
if the skeleton of each node admits a compatible
embedding~\cite[Theorem 3.1]{abfjk-tppeg-10}.  We use this
characterization and show that any skeleton of a biconnected \peg that
avoids all obstructions admits a compatible embedding.  Since
skeletons of S-nodes have only one embedding, and their embedding is
always compatible, we consider P- and R-nodes only.  The two types of
nodes are handled separately in Subsections~\ref{ssec-pnode}
and~\ref{ssec-rnode}, respectively.

The following lemma will be useful in several parts of the proof.

\begin{lemma}\label{lem:interval}
  Let $(G,H,\H)$ be a \peg, let $u$ be a vertex of a skeleton
  $\mathfrak{S}$ of a node~$\mu$ of the SPQR-tree of~$G$, and let $e$ be
  an edge of $\mathfrak{S}$ with endpoints $u$ and $v$.  Let $F
  \subseteq E(H)$ be the set of edges of $H$ that are incident to $u$
  and project into~$e$.  If the edges of $F$ do not form an interval
  in the rotation scheme of~$u$ in $\H$ then $(G,H,\H)$ contains
  obstruction~2.
\end{lemma}

\begin{proof}
  If $F$ is not an interval in the rotation scheme, then there exist
  edges $f,f' \in F$ and $g,g' \in E(H) \setminus F$, all incident to
  $u$, and appearing in the cyclic order $f,g,f',g'$ around $u$ in
  $\H$.  Let $x$ and $x'$ be the endpoints of $f$ and $f'$ different
  from $u$ and let $y$ and $y'$ be the endpoints of $g$ and $g'$
  different from $u$.  For any skeleton edge~$f$, we denote with $G_f$ the
  pertinent graph of~$f$.

  If $\mu$ is an S-node, then $g$ and $g'$ project to the same
  skeleton edge $uw$ with $v \ne w$.  Note that $G_{uv}$ and $G_{uw}$
  share only the vertex $u$ and moreover, they are both connected even
  after removing $u$.  Therefore, there exist disjoint paths $P$ in
  $G_{uv}$ and $Q$ in $G_{uw}$ connecting $x$ to $x'$ and $y$ to $y'$,
  respectively.  We may relax all internal vertices and all edges of
  $P$ and $Q$, and then perform simple edge contractions to replace
  each of the two paths with a single edge.  This yields
  obstruction~2.

  If $\mu$ is an R-node, then $G_{uv} - u$ is connected, and hence it
  contains a path $P$ from $x$ to $x'$.  Moreover, since $G - G_{uv}$
  is connected, it has a path $Q$ from $y$ to $y'$.  As in the
  previous case, contraction of $P$ and $Q$ yields obstruction~2.

  If~$\mu$ is a P-node, then $G_e - \{u,v\}$ is connected, and
  therefore there is a path $P$ connecting $x$ to $x'$ in
  $G_e-\{u,v\}$.  Analogous to the previous cases, a path $Q$ from $y$
  to $y'$ exists that avoids $u$ and $P$.  Again their contraction
  yields obstruction~2.
\end{proof}

In the following, we assume that the $H$-edges around each
vertex of a skeleton that project to the same skeleton edge form an
interval in the rotation scheme of this vertex.

\subsection{P-Nodes}\label{ssec-pnode}

Throughout this section, we assume that $(G,H,\H)$ is a biconnected
obstruc\-tion-free \peg.  We fix a P-node $\mu$ of the SPQR-tree of $G$,
and we let $\pskel$ be its skeleton.  Let $u$ and $v$ be the two
vertices of $\pskel$, and let $e_1,\dots,e_k$ be its edges.  Let $G_i$
be the pertinent graph of~$e_i$.  Recall that $G_i$ is
either a single edge connecting $u$ and $v$, or it does not contain
the edge $uv$ and $G_i - \{u,v\}$ is connected.

The goal of this section is to prove that $\pskel$ admits a compatible
embedding.  We first deal with edge-compatibility.

\begin{lemma}\label{lem:p-edge}
  Let $(G,H,\H)$ be a biconnected obstruc\-tion-free \peg.  Then every
  P-skeleton $\pskel$ has an edge-compatible embedding.
\end{lemma}

\begin{proof}
  If $\pskel$ has no edge-compatible embedding, then the rotation
  scheme around $u$ conflicts with the rotation scheme around~$v$.
  This implies that there is a triplet of skeleton edges
  $e_a,e_b,e_c$, for which the rotation scheme around~$u$ imposes a
  different cyclic order than the rotation scheme around~$v$.  We
  distinguish two cases.

  \emph{Case 1.} The graph $H$ has a cycle $C$ whose edges intersect
  two of the three skeleton edges, say $e_a$ and $e_b$.  Then the edge
  $e_c$ must contain a vertex $x$ whose prescribed embedding is to the
  left of $C$, as well as a vertex $y$ whose prescribed embedding is
  to the right of $C$.  Since $x$ and $y$ are connected by a path in
  $G_c -\{u,v\}$, we obtain obstruction~1.

  \emph{Case 2.} The graph $H$ has no cycle that intersects two of the
  three $\pskel$-edges $e_a,e_b,e_c$.  Each of the three
  $\pskel$-edges contains an edge of $H$ adjacent to $u$ as well as an
  edge of $H$ adjacent to $v$.  Since $G_i - \{u,v\}$ is connected for
  each $i$, it follows that each of the three skeleton edges contains
  a path from $u$ to $v$, such that the first and the last edge of the
  path belong to $H$.  Fix such paths $P_a$, $P_b$ and $P_c$,
  projecting into $e_a$, $e_b$ and $e_c$, respectively.

  At least two of these paths ($P_a$ and $P_b$, say) also
  contain an edge not belonging to $H$, otherwise they would form a
  cycle of $H$ intersecting two skeleton edges.  By relaxations and
  simple contractions, we may reduce $P_a$ to a path of length three,
  whose first and last edge belong to $E(H)$ and the middle edge
  belongs to $E(G)\setminus E(H)$.  The same reduction can be
  performed with $P_b$.  The path $P_c$ can then be contracted to a
  single vertex, to obtain obstruction~2.
\end{proof}

Next, we consider cycle-compatibility.  Assume that $H$ has at least
one facial cycle whose edges intersect two distinct skeleton edges.
It follows that $u$ and $v$ belong to the same connected component of
$H$; denote this component by $H_{uv}$.  We call a \emph{uv-cycle} any
facial cycle of $H$ that contains both $u$ and $v$.  Note that any
$uv$-cycle is also a facial cycle of~$H_{uv}$, and a facial cycle of
$H_{uv}$ that contains both $u$ and $v$ is a $uv$-cycle.  Following
the conventions of~\cite{abfjk-tppeg-10}, we assume that all facial
cycles are oriented in such a way that a face is to the left of its
facial cycles.  The next lemma shows that the vertices of $H_{uv}$
cannot violate any cycle-compatibility constraints without violating
edge-compatibility as well.

\newcommand{\lemsamecomptext}{%
  Assume that $C$ is a $uv$-cycle that intersects two distinct
  $\pskel$-edges $e_a$ and $e_b$, and that $x$ is a vertex of $H_{uv}$ not
  belonging to $C$.  In any edge-compatible embedding of $\pskel$,
  the vertex $x$ does not violate cycle-compatibility with respect
  to~$C$.%
}
\begin{lemma}\label{lem:same-comp}
  \lemsamecomptext
\end{lemma}

\begin{proof}
  The vertex $x$ belongs to a skeleton edge $e_x$ different from $e_a$
  and $e_b$, otherwise it cannot violate cycle-compatibility.  Note
  that since $x$ is in $H_{uv}$, $e_x$ must contain a path $P$ of $H$
  that connects $x$ to one of the poles $u$ and $v$.  In the graph
  $H$, all the vertices of $P$ must be embedded on the same side of
  $C$ as the vertex $x$.  The last edge of $P$ may not violate
  edge-compatibility, which forces the whole edge $e_x$, and thus $x$,
  to be embedded on the correct side of the projection of $C$, as
  claimed.
\end{proof}

The next lemma shows that for an obstruction-free $\peg$, all vertices
of $H$ projecting to the same $\pskel$-edge impose the same
cycle-compatibility constraints for the placement of this edge.

\newcommand{\lemsamevirtualedgetext}{%
  Let $x$ and $y$ be two vertices of $H$, both distinct from $u$ and~$v$.
  Suppose that $x$ and $y$ project to the same $\pskel$-edge
  $e_a$.  Let $C$ be a cycle of $H$ that is edge-disjoint from $G_a$.
  Then $x$ and $y$ are embedded on the same side of~$C$ in~$\H$.%
}

\begin{lemma}\label{lem:same-virtual-edge}
  \lemsamevirtualedgetext
\end{lemma}

\begin{proof}
  Since $G_a - \{u,v\}$ is a connected subgraph of $G$, there is a
  path $P$ in $G$ that connects $x$ to $y$ and avoids $u$ and~$v$.
  Since $C$ is edge-disjoint from $G_a$, the path $P$ avoids all the
  vertices of~$C$.  If $x$ and $y$ were not embedded on the same side
  of~$C$ in $\H$, we would obtain obstruction~1 by contracting $C$ and
  $P$.
\end{proof}

We now prove the main result of this subsection.

\begin{proposition}
  \label{prop:pnode}
  Let $(G,H,\H)$ be a biconnected obstruc\-tion-free \peg.  Then every
  P-skeleton $\pskel$ of the SPQR-tree of~$G$ admits a compatible
  embedding.
\end{proposition}

\begin{proof} 
  Fix an edge-compatible embedding that minimizes the number of
  violated cycle-compatibility constraints; more precisely, fix an
  embedding of $\pskel$ that minimizes the number of pairs $(C,x)$
  where $C$ is a facial cycle of $\H$ projecting to a cycle $C'$ of
  $\pskel$, $x$ is a vertex of $H-\{u,v\}$ projecting into a skeleton
  edge $e_x$ not belonging to $C'$, and the relative position of $C'$
  and $e_x$ in the embedding of $\pskel$ is different from the
  relative position of $C$ and $x$ in~$\H$.  We claim that the chosen
  embedding of $\pskel$ is compatible.

  For contradiction, assume that there is at least one pair $(C,x)$
  that violates cycle-compatibility in the sense described above.  Let
  $e_x$ be the $\pskel$-edge containing $x$.  Note that $e_x$ does not
  contain any edge of $H$ adjacent to $u$ or $v$.  If it contained
  such an edge, it would contain a vertex $y$ from the component
  $H_{uv}$, and this would contradict Lemma~\ref{lem:same-comp} or
  Lemma~\ref{lem:same-virtual-edge}.  Thus, the edge $e_x$ does not
  participate in any edge-compatibility constraints.

  It follows that $x$ does not belong to the component $H_{uv}$.  That
  means that in $\H$, the vertex $x$ is embedded in the interior of a
  unique face $F$ of $H_{uv}$.  We distinguish two cases, depending on
  whether the boundary of $F$ contains both poles $u$ and $v$ of
  $\pskel$ or not.

\begin{figure}
 \centering
\includegraphics{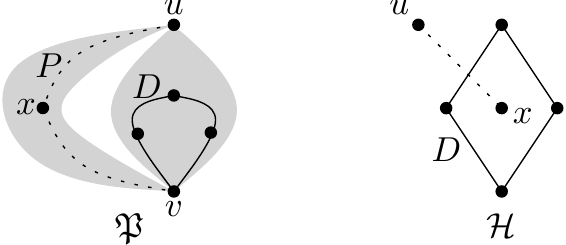}
\caption{Illustration of Case 1 in the proof of Proposition~\ref{prop:pnode}.
The shaded regions represent the edges of~$\pskel$.}\label{fig:prop8case1}
\end{figure}

  \emph{Case 1.} The boundary of $F$ contains at most one of the two
  poles $u$ and~$v$; see Fig.~\ref{fig:prop8case1}.  Without loss of
generality, the boundary of $F$ does
  not contain $u$.  Thus, $F$ has a facial cycle $D$ that separates
  $u$ from $x$.  The pertinent graph $G_x$ of $e_x$ contains a path
  $P$ from $x$ to $u$ that avoids~$v$.  The path $P$ does not contain
  any vertex of $H_{uv}$ except $u$, and in particular, it does not
  contain any vertex of~$D$.
  Contracting $D$ to a triangle and $P$ to an edge yields
  obstruction~1, which is a contradiction.

  \emph{Case 2.} The boundary of $F$ contains both poles $u$ and $v$
  of the skeleton.  In this case, since $u$ and $v$ belong to the same
  block of $H$, the face $F$ has a unique facial cycle $D$ that
  contains both $u$ and $v$.  The cycle $D$ is the only $uv$-cycle
  that has $x$ to its left (i.e., inside its corresponding face).

  The cycle $D$ may be expressed as a union of two paths $P$ and $Q$
  connecting $u$ and $v$, where $P$ is directed from $u$ to $v$ and
  $Q$ is directed from $v$ to $u$.  We distinguish two subcases,
  depending on whether the paths $P$ and $Q$ project to different
  $\pskel$-edges.

  \emph{Case 2.a} Both $P$ and $Q$ project to the same skeleton
  edge~$e_D$.  This case is depicted in Fig.~\ref{fig:prop8case2a}. 
  Each of the two paths $P$ and $Q$ has at least one internal
  vertex.  Since all these internal vertices are inside a single
  skeleton edge, there must be a path $R$ in $G$ connecting an
  internal vertex of $P$ to an internal vertex of $Q$ and avoiding
  both $u$ and $v$.  By choosing $R$ as short as possible, we may
  assume that no internal vertex of $R$ belongs to $D$.  Furthermore,
  since $\pskel$ by hypothesis has at least one violated
  cycle-compatibility constraint, it must contain at least two edges
  that contain an $H$-path from $u$ to~$v$.  In particular, there must
  exist a $\pskel$-edge $e_S$ different from $e_D$ that contains an
  $H$-path $S$ from $u$ to~$v$.
%
%   \begin{figure}
%     \centering
%     \subfigure[]{\includegraphics{prop8_1}\label{fig:prop8case1}}
%     \hfil
%     \subfigure[]{\includegraphics{prop8_2}\label{fig:prop8case2a}}
%     \caption{Illustration of Case 1~\subref{fig:prop8case1} and Case
%       2.a~\subref{fig:prop8case2a} in the proof of
%       Proposition~\ref{prop:pnode}.  Edges of the skeleton are shown
%       as shaded regions.}
%     \label{fig:prop-pnode-illustration}
%   \end{figure}
%

\begin{figure}
 \centering
\includegraphics{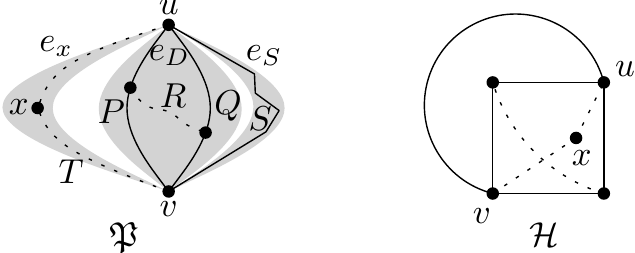}
\caption{Illustration of Case 2.a in the proof of Proposition~\ref{prop:pnode}.
The shaded regions represent the edges of~$\pskel$.}\label{fig:prop8case2a}
\end{figure}

Necessarily, the path $S$ is embedded outside the face $F$, i.e., the
right of $D$.  And finally, the edge $e_x$ must contain a $G$-path~$T$
from $u$ to $v$ that contains $x$.  Note that $e_x$ is different from
$e_D$ and $e_S$, because $e_x$ has no $H$-edge incident to $u$ or $v$.
Thus, the paths $P$, $Q$, $S$, $T$ are all internally disjoint.  The
five paths $P$, $Q$, $R$, $S$, and $T$ can then be contracted to form
obstruction~3.

  \emph{Case 2.b} The two paths $P$ and $Q$ belong to distinct
  skeleton edges $e_P$ and~$e_Q$.  That means that the facial cycle
  $D$ projects to a cycle $D'$ of the skeleton, formed by the two
  edges.  Modify the embedding of the skeleton by moving $e_x$ so that
  it is to the left of $D'$.  This change does not violate
  edge-compatibility, because $e_x$ has no $H$-edge adjacent to $u$
  or~$v$.

  We claim that in the new skeleton embedding, $x$ does not
  participate in any violated cycle-compatibility constraint.  To see
  this, we need to check that $x$ is embedded to the right of any
  facial cycle $B \ne D$ of $H_{uv}$ that projects to a cycle in the
  skeleton.  Choose such a cycle $B$ and let $B'$ be its projection; see
Fig.~\ref{fig:prop8case2b}.
  Let $e^+$ or $e^-$ denote the edges of $D$ incident to $u$ with
  $e^+$ being oriented towards $u$ and $e^-$ out of $u$.  Similarly,
  let $f^+$ and $f^-$ be the incoming and outgoing edges of $B$
  adjacent to $u$.  In $\H$, the four edges must visit $u$ in the
  clockwise order $(e^+,e^-,f^+,f^-)$, with the possibility that
  $e^-=f^+$ and $e^+=f^-$.

  \begin{figure}
    \centering
    \includegraphics{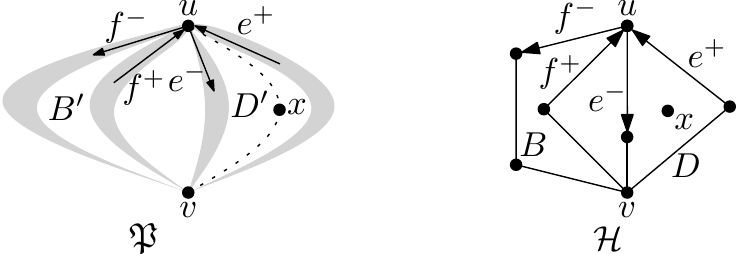}
    \caption{Illustration of Case 2.b in the proof of
      Proposition~\ref{prop:pnode}.  The left part represents the
      embedding of $\pskel$ after after $e_x$ has been moved to the
      left of~$D'$.  Edge-compatibility guarantees that $x$ is now on
      the correct side of every facial cycle
      in~$\pskel$.}\label{fig:prop8case2b}
  \end{figure}

  Since the embedding of the skeleton is edge-compatible, this means
  that any skeleton edge embedded to the left of $D'$ is also to the
  right of $B'$, as needed.  We conclude that in the new embedding of
  $\pskel$, the vertex $x$ does not violate any cycle-compatibility
  constraint, and by Lemma~\ref{lem:same-virtual-edge}, the same is
  true for all the other $H$-vertices in~$e_x$.  Moreover, the change
  of embedding of $e_x$ does not affect cycle-compatibility of
  vertices not belonging to $e_x$, so the new embedding violates fewer
  cycle-compatibility constraints than the old one, which is a
  contradiction.
  This proves that $\pskel$ has a compatible embedding.
\end{proof}

Let us remark that there are only finitely many obstructions that may arise
from a $P$-skeleton that lacks a compatible embedding. In fact, if $(G,H,\H)$ is
a non-planar \peg and if $G$ is a biconnected graph with no $K_4$-minor
(implying that the SPQR-tree of $G$ has no $R$-nodes), then we may conclude that
$(G,H,\H)$ contains obstruction 1 or 2, since all the other obstructions contain
$K_4$ as (ordinary) minor. 

\subsection{R-Nodes}\label{ssec-rnode}

Let us now turn to the analysis of R-nodes.  As in the case of
P-nodes, our goal is to show that if a skeleton $\rskel$ of an R-node
in the SPQR-tree of $G$ has no compatible embedding, then the
corresponding \peg $(G,H,\H)$ contains an obstruction.  The skeletons
of R-nodes have more complicated structure than the skeletons of
P-nodes, and accordingly, our analysis is more complicated as well.
Similar to the case of P-nodes, we will first show that an R-node of
an obstruction-free \peg must have an edge-compatible embedding, and
as a second step show that in fact it must also have an
edge-compatible embedding that also is cycle-compatible.

The skeleton of an R-node is a 3-connected graph.  We therefore start
with some preliminary observations about 3-connected graphs, which
will be used throughout this section.  Let $\rskel$ be a 3-connected
graph with a planar embedding $\rskelemb$, let $x$ be a vertex
of~$\rskel$.  A vertex $y$ of $\rskel$ is \emph{visible from $x$} if
$x\neq y$ and there is a face of $\rskelemb$ containing $x$ and $y$ on
its boundary.  An edge $e$ is visible from $x$ if $e$ is not incident
with $x$ and there is a face containing both $x$ and $e$ on its
boundary.  The vertices and edges visible from $x$ form a cycle in
$\rskel$.  To see this, note that these vertices and edges form a face
boundary in $\rskelemb - x$, and every face boundary in an embedding
of a 2-connected graph is a cycle.  We call this cycle \emph{the
  horizon of $x$}.

In the following, we will consider a fixed skeleton~$\rskel$ of an
R-node.  Since $\rskel$ is 3-connected, it has only two planar
embeddings, denoted by $\rskelemb$ and $\rskelnot$~\cite{w-cgcg-32}.
Suppose that neither of the two embeddings is compatible.  The
constraints on the embeddings either stem from a vertex whose incident
$H$-edges project to distinct edges of~$\rskel$ or from a cycle
of~$\rskel$ that is a projection of an $H$-cycle whose
cycle-compatibility constraints demand exactly one of the two
embeddings.  Since neither~$\rskelemb$ nor $\rskelnot$ are compatible,
there must be at least two such structures, one requiring
embedding~$\rskelemb$, and the other one requiring~$\rskelnot$.  If
these structures are far apart in~$\rskel$, for example, if no vertex
of the first structure belongs to the horizon of a vertex of the
second structure, it is usually not too difficult to find one of the
obstructions.  However, if they are close together, a lot of special
cases can occur.  A significant part of the proof therefore consists
in controlling the distance of objects and showing that either an
obstruction is present or close objects cannot require different
embeddings.

As before, we distinguish two main cases: first, we deal with the
situation in which both embeddings of $\rskel$ violate
edge-compatibility.  Next, we consider the situation in which $\rskel$
has at least one edge-compatible embedding, but no edge-compatible
embedding is cycle-compatible.

\subsubsection{$\rskel$ has no edge-compatible embedding}

Let $u$ be vertex of $\rskel$ that violates the edge-compatibility of
$\rskelemb$, and let $v$ be a vertex violating edge-compatibility of
$\rskelnot$.  If $u=v$, i.e., if a single vertex violates
edge-compatibility in both embeddings, the following lemma shows that
we can find an occurrence of obstruction~2 in $(G,H,\H)$.

\begin{lemma}\label{lem-singlevertex}
  Assume that an R-node skeleton $\rskel$ has a vertex $u$ that
  violates edge-compatibility in both embeddings of~$\rskel$.  Then
  $(G,H,\H)$ contains obstruction~2.
\end{lemma}

\begin{proof}
  Let $e'_1,\dots,e'_m$ be the $\rskel$-edges incident to $u$ that
  contain at least one $H$-edge incident to $u$.  Assume that these
  edges are listed in their clockwise order around~$u$ in the
  embedding $\rskelemb$.  Let $e_i$ be an $H$-edge incident to $u$
  projecting into $e'_i$.  By Lemma~\ref{lem:interval}, if a triple of
  edges $e'_i,e'_j,e'_k$ violates edge-compatibility in $\rskelemb$,
  then this violation is demonstrated by the edges $e_i,e_j,e_k$,
  i.e., the cyclic order of $e_i, e_j$ and $e_k$ in $\H$ is different
  from the cyclic order of $e'_i,e'_j$ and $e'_k$ in~$\rskelemb$.

  Choose a largest set $I \subseteq \{1,\dots,m\}$ such that the edges
  $\{e_i, i\in I\}$ do not contain any violation of edge-compatibility
  when embedded according to $\rskelemb$.  Clearly, $3 \le |I| < m$
  because if each triple violated edge-compatibility in~$\rskelemb$,
  then~$\rskelnot$ would be edge-compatible with~$u$.  Also~$|I|<m$,
  otherwise~$\rskelemb$ would be edge-compatible with~$u$.

  Choose an index $i \in \{1,\dots,m\}$ not belonging to $I$.  By
  maximality of~$I$, there are $j,k,\ell \in I$ such that, without
  loss of generality, $(e_i,e_j,e_k,e_\ell)$ appear clockwise in
  $\rskelemb$ and $(e_j,e_i,e_k,e_\ell)$ appear clockwise in $\H$
  (recall that $(e_j,e_k,e_\ell)$ have the same order in $\rskelemb$
  and $\H$, by the definition of~$I$).

  For $a \in \{1,\dots,m\}$ let $x_a$ be the endpoint of the skeleton
  edge $e'_a$ different from $u$.  The horizon of $u$ in $\rskelemb$
  contains two disjoint paths $P$ and $Q$ joining $x_i$ with $x_\ell$
  and $x_j$ with~$x_k$.  By obvious contractions we obtain
  obstruction~2.
\end{proof}

Let us concentrate on the more difficult case when $u$ and $v$ are
distinct.  To handle this case, we introduce the concept of
`wrung \pegs'.  A \emph{wrung \peg} is a \peg $(G,H,\H)$
with the following properties.
\begin{itemize}%[1)]
\itemsep=-.5ex
\item $G$ is a subdivision of a 3-connected planar graph, therefore it
  has two planar embeddings $\G^+$, $\G^-$.
\item $H$ has two distinct vertices $u$ and $v$ of degree~3.  Any
  other vertex of $H$ is adjacent to $u$ or $v$, and any edge of $H$
  is incident to $u$ or to $v$.  Hence, $H$ has five or six edges, and
  at most eight vertices.s
\item $H$ is not isomorphic to $K_{2,3}$ or to $K_4^-$ (i.e., $K_4$
  with an edge removed).  Equivalently, $H$ has at least one vertex of
  degree~1.
\item The embedding $\H$ of $H$ is such that its rotation scheme
  around $u$ is consistent with $\G^+$ and its rotation scheme around
  $v$ is consistent with $\G^-$.  Note that such an embedding exists
  due to the previous condition.
\end{itemize}

\begin{figure}
\hfil \includegraphics{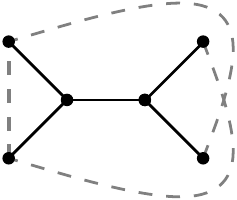}
\caption{An example of a minimal wrung \peg that is not a minimal planarity
obstruction (it contains obstruction 2).}\label{fig:wr_ex}
\end{figure}

Clearly, a wrung \peg is not planar, because neither $\G^+$ nor
$\G^-$ is an extension of~$\H$.  A \emph{minimal wrung \peg} is
a wrung \peg that does not contain a smaller wrung \peg
as a \peg-minor.  A minimal wrung \peg is not necessarily a
planarity obstruction---it may contain a smaller non-planar \peg
that is not wrung (see Fig.~\ref{fig:wr_ex}).  However, it turns out that
minimal wrung \pegs are close to being planarity obstructions.  The
key idea in using wrung \pegs is that they are characterized by
being subdivisions of 3-connected graphs, a property that is much
easier to control than non-embeddability of \pegs.

The following proposition summarizes the key property of wrung
\pegs.  In particular, it shows that there are only finitely
many minimal wrung \pegs.

\newcommand{\prowrungtext}{%
  If $(G,H,\H)$ is a minimal wrung \peg, then every vertex of $G$ also
  belongs to $H$ and the graph $H$ is connected.
}

\begin{proposition}\label{pro-wrung}
\prowrungtext
\end{proposition}

The proof of this proposition relies heavily on the notion of
`contractible edge', which is an edge in a 3-connected graph whose
contraction leaves the graph 3-connected.  This notion has been
intensely studied \cite{k-cscg-00,k-scegp-02}, and we are able to use
powerful structural theorems that guarantee that any `sufficiently
large' wrung \pegs must contain an edge that can be contracted
to yield a smaller wrung \peg.

\begin{proof}
  Let $G^\star$ be the 3-connected graph whose subdivision is~$G$.
  A \emph{subdivision vertex} is a vertex of $G$ of degree~2.  A
  \emph{subdivided edge} is a path in $G$ of length at least two whose
  every internal vertex is a subdividing vertex and whose endpoints
  are not subdividing vertices.  Therefore, each edge of $G^\star$
  either represents an edge of $G$ or a subdivided edge of~$G$.

  The proof of the proposition is based on several claims.
  \begin{claim}\label{claim:subdivision}
    Every subdividing vertex of $G$ is a vertex of $H$.  Every
    subdivided edge of $G$ contains at most one vertex adjacent to $u$
    and at most one vertex adjacent to $v$.  If $H$ is disconnected
    then $G$ has at most one subdivided edge, which (if it exists)
    connects $u$ and $v$ and is subdivided by a single vertex.
  \end{claim}
  If $G$ had a subdividing vertex $x$ not belonging to $H$, we could
  contract an edge of $G$ incident to $x$ to get a smaller \peg, which
  is still wrung.  

  To see the second part of the claim, note that two
  vertices adjacent to $u$ in the same subdivided edge would imply the
  existence of a loop or a multiple edge in $G^\star$.  

  For the last
  part of the claim, note that if $H$ is disconnected, then every
  vertex of $H$ except for $u$ and $v$ has degree~1 in $H$.  If a
  subdividing vertex adjacent to~$u$ were also adjacent to an
  $H$-neighbor of $v$, then the edge between them could be contracted.
  This proves the claim.

  A fundamental tool in the analysis of minimal wrung \pegs is
  the concept of contractible edges.  An edge $e$ in a 3-connected
  graph $F$ is \emph{contractible} if $F.e$ is also 3-connected, where
  $F.e$ is the graph obtained from $F$ by contracting $e$.  Note that
  an edge $e=xy$ in a 3-connected graph $F$ is contractible if and only if $F
  - \{x,y\}$ is biconnected.

  The next fact is a special case of a theorem by
  Kriesell~\cite{k-cscg-00}, see also~\cite[Theorem 3]{k-scegp-02}.

  \newtheorem{fact}{Fact}

  \begin{fact}\label{fact:contractible}
    If $F$ is a 3-connected graph and $w$ a vertex of $F$ that is not
    incident with any contractible edge and such that $F-w$ is not a
    cycle, then $w$ is adjacent to four vertices $x_1,x_2,y_1,y_2$,
    all having degree~3 in $F$, which induce two disjoint edges
    $x_1y_1$ and $x_2y_2$ of $F$, and both these edges are
    contractible.
  \end{fact}

  We are now ready to show that every vertex of $G$ also belongs to
  $H$.  Suppose for a contradiction that $G$ has a vertex $w$ not
  belonging to $H$.  By Claim~\ref{claim:subdivision}, $w$ is not a
  subdivision vertex, so $w$ is also a vertex of $G^\star$.  If $w$
  were incident to a contractible edge of $G^\star$, we could contract
  this edge to obtain a smaller wrung \peg.  Hence, $w$ is not
  incident to any contractible edge of $G^\star$.  Fix now the four
  vertices from Fact~\ref{fact:contractible}, and let $e_1=x_1y_1$ and
  $e_2=x_2y_2$ be the two contractible edges.  Necessarily all the
  four endpoints of $e_1$ and $e_2$ belong to $H$, otherwise we could
  contract one of them to get a smaller wrung \peg.  Moreover,
  the edges $e_1$ and $e_2$ cannot contain $u$ or $v$, because their
  endpoints have degree three and are adjacent to the vertex $w$ not
  belonging to $H$.  Therefore, each endpoint of $e_1$ and $e_2$ is
  adjacent to either $u$ or $v$ in $G^\star$ (and also in $G$ and in
  $H$).

  Assume without loss of generality that $x_1$ is adjacent to~$u$.
  Then $y_1$ cannot be adjacent to $u$, because then $u$ and $w$ would
  form a separating pair in $G^\star$, hence $y_1$ is adjacent to $v$.
  Analogously, we may assume that $x_2$ is adjacent to $u$ and $y_2$
  is adjacent to $v$.  The graph $H$ must be connected, otherwise we
  could contract $e_1$ or $e_2$.  This means that $H$, together with
  $e_1$ and $e_2$ and the two edges $wx_1$ and $wx_2$ form a
  subdivision of $K_4$, and therefore they form a wrung \peg
  properly contained in $(G,H,\H)$.  Therefore any vertex of $G$ also
  belongs to~$H$.

  It remains to prove that $H$ is connected.  For this we need another
  concept for dealing with subdivisions of 3-connected graphs.  Let
  $F$ be a 3-connected graph and let $e=xy$ be an edge of $F$.  The
  \emph{cancellation} of $e$ in $F$ is the operation that proceeds in
  the three steps
  \begin{inparaenum}[1)]
    \item Remove $e$ from $F$, to obtain $F-e$,
    \item If the vertex $x$ has degree~2 in $F-e$, then replace the
      subdivided edge containing $x$ by a single edge.  Do the same
      for $y$ as well.
    \item Simplify the graph obtained from step~2 by removing multiple
      edges.
  \end{inparaenum}

  Let $F \ominus e$ denote the result of the cancellation of $e$ in
  $F$.  Note that $F \ominus e$ may contain vertices of degree~2 if
  they arise in step 3 of the above construction.  An edge $e$ is
  \emph{cancellable} if $F \ominus e$ is 3-connected.  It is called
  \emph{properly cancellable} if it is cancellable, and moreover, the
  first two steps in the above definition produce a graph without
  multiple edges.

  \begin{claim}\label{claim:cancelorcontract}
    A cancellable edge $e$ in a 3-connected graph $F$ is either
    properly cancellable or contractible.
  \end{claim}

  Suppose that $e=xy$ is cancellable, but not properly cancellable.
  We show that it is contractible.  Since $e$ is not properly
  cancellable, one of its endpoints, say $x$, has degree~3 in $F$ and
  its two neighbors $x'$ and $x''$ besides $y$ are connected by an
  edge.  We show that between any pair of vertices $a$ and $b$ of $F -
  \{x,y\}$ there are two vertex-disjoint paths.  In $F$ there exist
  three vertex-disjoint $a-b$-paths $P_1,P_2$ and $P_3$.  If two of
  them avoid $x$ and $y$ then they are also present in $F-\{x,y\}$.
  Therefore, we may assume that $P_1$ contains $x$ and $P_2$ contains
  $y$.  Then $P_1$ contains the subpath $x'xx''$ which can be replaced
  by the single edge $x'x''$.  Again at most one of the paths contains
  vertices of $\{x,y\}$ and therefore we again find two
  vertex-disjoint $a-b$-paths in $F-\{x,y\}$.  This shows that
  $F-\{x,y\}$ is biconnected and therefore $e=xy$ is contractible.
  This concludes the proof of the claim.

  Moreover, we need the following result by Holton et
  al.~\cite{hjsw-recg-90}, which we present without proof.
  \begin{fact}\label{fact:cancel}
    If $F$ is a 3-connected graph with at least five vertices, then
    every triangle in $F$ has at least two cancellable edges.
  \end{fact}

  We proceed with the proof of Proposition~\ref{pro-wrung}, and show that $H$
  is connected.  Suppose for a
  contradiction that $H$ is disconnected, and let $H_u$ and $H_v$ be
  its two components containing $u$ and $v$.  Let $x_1,x_2$ and $x_3$
  be the three neighbors of~$u$ in $H$, and $y_1,y_2$ and $y_3$ the
  three neighbors of~$v$.  Recall from Claim~\ref{claim:subdivision}
  that $G$ has at most one subdividing vertex, and that the possible
  subdivided edge connects $u$ and $v$.

  Since $G^\star$ is 3-connected, it has three disjoint edges
  $e_1,e_2$ and $e_3$, each of them connecting a vertex of $H_u$ to a vertex of
  $H_v$.  At least one of them avoids both $u$ and $v$.  Assume
  without loss of generality that $e_1=x_1y_1$ is such an edge.  If
  $e_1$ were a contractible edge of $G^\star$, we would get a smaller
  wrung \peg.  Therefore the graph $G^\star - \{x_1,y_1\}$ has a
  cut-vertex $w$.  Note that $w$ is either $u$ or $v$.  Otherwise,
  $H_u - \{x_1,y_1,w\}$ would be connected, and $H_v - \{x_1,y_1,w\}$
  would be connected as well.  However, since (by disjointness) at
  least one of the edges $e_2, e_3$ avoids $w$, this implies that
  $G^\star - \{x_1,y_1,w\}$ would be connected as well, contradicting
  the choice of~$w$.

  So, without loss of generality, $G^\star$ has a separating triplet
  $\{x_1,y_1,u\}$.  Since at least one of the two edges $e_2,e_3$
  avoids this triplet, we see that one of the components of $G^\star -
  \{x_1,y_1,u\}$ consists of a single vertex $x' \in \{x_2,x_3\}$.
  Since each vertex in a minimal separator must be adjacent to each of
  the components separated by the separator, $G^\star$ contains the
  two edges $x'x_1$ and $x'y_1$.  Consequently, $x'$, $x_1$ and $y_1$
  induce a triangle in $G^\star$ (and in $G$), and by
  Fact~\ref{fact:cancel}, at least one of the two edges $x_1y_1$ and
  $x'y_1$ is cancellable, and by Claim~\ref{claim:cancelorcontract},
  at least one of the two edges is contractible or properly
  cancellable, contradicting the minimality of $(G,H,\H)$.

  This completes the proof of Proposition~\ref{pro-wrung}.
\end{proof}

Proposition~\ref{pro-wrung} implies that a minimal wrung \peg
has at most seven vertices.  Therefore, to show that each wrung
\peg contains one of the obstructions from
Fig.~\ref{fig:obstructions} is a matter of a finite (even if a bit
tedious) case analysis.  We remark that a minimal wrung \peg
may contain any of the exceptional obstructions of
Fig.~\ref{fig:obstructions}, except obstructions 18--22, obstruction
3, $K_5$, and $K_{3,3}$.  A minimal wrung \peg does not contain
any $k$-fold alternating chain for~$k\ge 4$.  As the analysis requires
some more techniques, we defer the proof to Lemma~\ref{lem:min-wrung}.

Let us show how the concept of wrung \pegs can be applied in the
analysis of R-skeletons.  Consider again the skeleton $\rskel$, with
two distinct vertices $u$ and $v$, each of them violating
edge-compatibility of one of the two embeddings of~$\rskel$.  This
means that $u$ is incident to three $H$-edges $e_1, e_2, e_3$
projecting into distinct $\rskel$-edges $e'_1, e'_2, e'_3$, such that
the cyclic order of $e_i$'s in $\H$ coincides with the cyclic order of
$e'_i$'s in $\rskelnot$, and similarly $v$ is adjacent to $H$-edges
$f_1, f_2, f_3$ projecting into $\rskel$-edges $f'_1,f'_2,f'_3$, whose
order in $\rskelemb$ agrees with~$\H$.  We have the following
observation.

\begin{observation}
 \label{obs:wrung}
 If all the~$e_i'$ and~$f_i'$ for~$i=1,2,3$ are distinct, then~$G$
 contains a wrung \peg.
\end{observation}

If all the~$e_i'$ and~$f_i'$ for~$i=1,2,3$ are distinct, then it is
fairly easy to see that $G$ must contain a wrung \peg, obtained simply
by replacing each edge of $\rskel$ with a path of $G$, chosen in such
a way that all the six edges $e_i$ and $f_i$ belong to these
paths. Such a choice is always possible and yields a wrung \peg.  In
particular, this is always the case if~$u$ and~$v$ are not adjacent
in~$\rskel$.  Thus the observation holds.

If, however, $u$ and $v$ are connected by an $\rskel$-edge $g'$, and
if, moreover, we have $e'_i=g'=f'_j$ for some $i$ and $j$, the
situation is more complicated, because there does not have to be a
path in $G$ that contains both edges $e_i$ and $f_j$ and projects
into~$g'$.  In such a situation, we do not necessarily obtain a wrung
\peg.  This situation is handled separately in
Lemma~\ref{lem:non-wrung}.  Altogether, we prove the following
proposition.

\begin{proposition}\label{prop:biconnected}
  Let~$(G,H,\H)$ be a biconnected obstruction-free \peg, and
  let~$\rskel$ be the skeleton of an R-node of the SPQR-tree of~$G$.
  Then~$\rskel$ has an edge-compatible embedding.
\end{proposition}

We already know that if~$\rskel$ does not have an edge-compatible
embedding, then by Lemma~\ref{lem-singlevertex} it either contains
obstruction~2, or two distinct vertices~$u$ and~$v$ requiring
different embeddings of~$\rskel$.  In this case~$(G,H,\H)$ either
contains a wrung \peg, or it does not, and~$u$ and~$v$ are connected
by an edge.

In the remainder of this subsection, we prove that in either
case~$(G,H,\H)$ contains one of the obstructions from
Fig.~\ref{fig:obstructions}.  We first show that if~$\rskel$ does
not contain a wrung \peg, then it contains one of the
obstructions~4,5 and~6; see Lemma~\ref{lem:non-wrung}.  Finally, we
also present a detailed analysis showing that any minimal wrung
\peg (of which there are only finitely many) contains one of the obstructions
from
Fig.~\ref{fig:obstructions}; see Lemma~\ref{lem:min-wrung}.

The following technical lemma is a useful tool, which we employ in
both proofs. To state the lemma, we use the following notation: let
$x_1,x_2,\dotsc,x_k$ be (not necessarily distinct) vertices of a
graph~$F$. We say that a path $P$ in $F$ is \emph{a path of the form
  $x_1\to x_2\to\dotsb\to x_k$}, if $P$ is a simple path that is
obtained by concatenating a sequence of paths
$P_1,P_2,\dotsc,P_{k-1}$, where $P_i$ is a path connecting $x_i$ to
$x_{i+1}$ (note that if $x_i=x_{i+1}$, then $P_i$ consists of a single
vertex).

\begin{lemma}\label{lem-H}
  Let $\rskel$ be a 3-connected graph with a fixed planar embedding
  $\rskelemb$.  Let $u$ and $v$ be two vertices of $\rskel$ connected
  by an edge~$B$. Let $u_1$ and $u_2$ be two distinct neighbors of $u$,
  both different from $v$, such that $(v,u_1,u_2)$ appears
  counter-clockwise in the rotation scheme of $u$. Similarly, let $v_1$ and
  $v_2$ be two neighbors of $v$ such that $(u,v_1,v_2)$ appears
  clockwise around~$v$. (Note that we allow some of the $u_i$ to
  coincide with some $v_j$.) Then at least one of the following
  possibilities holds:

  \begin{enumerate}
    \itemsep=-.5ex
  \item The graph $\rskel$ contains a path of the form $v\to u_1\to u_2\to
v_1\to v_2\to u$.
  \item The graph $\rskel$ contains a path of the form $u\to v_1\to v_2\to
u_1\to u_2\to v$. (This is symmetric to the previous case.)
  \item The graph $\rskel$ has a vertex $w$ different from $u_2$ and
    three paths of the forms $w\to u_2\to v_2$, $w\to u_1$ and $w\to v_1$,
respectively. These
    paths only intersect in $w$, and none of them contains $u$ or~$v$.
  \item The graph $\rskel$ has a vertex $w$ different from $v_2$ and
    three paths of the forms $w\to v_2\to u_2$, $w\to v_1$ and $w\to u_1$,
respectively. These paths only intersect in $w$, and none of them contains $u$
    or~$v$. (This is again symmetric to the previous case.)
  \end{enumerate}
\end{lemma}

\begin{proof}
  Let $C_u$ be the horizon of $u$ and $C_v$ the horizon of $v$. Orient
  $C_u$ counterclockwise and split it into three internally disjoint
  oriented paths $v\to u_1$, $u_1\to u_2$ and $u_2\to v$, denoted by
  $C_u^1$, $C_u^2$, and $C_u^3$ respectively. Similarly, orient $C_v$
  clockwise, and split it into $C_v^1= u\to v_1$, $C_v^2= v_1\to v_2$,
  and $C_v^3= v_2\to u$.

  Let $F_1$ and $F_2$ be the two faces of $\rskel$ incident with the
  edge $uv$, with $F_1$ to the left of the directed edge
  $\vec{uv}$. Note that each vertex on the boundary of $F_1$ except
  $u$ and $v$ appears both in $C_u^1$ and in~$C_v^1$.  Similarly, the
  vertices of $F_2$ (other than $u$ and $v$) appear in $C_u^3$
  and~$C_v^3$. There may be other vertices shared between $C_u$
  and~$C_v$ and we have no control about their position. However, at
  least their relative order must be consistent, as shown by the
  following claim.

  \begin{claim} 
    \label{claim:clockwise}
    Suppose that $x$ and $y$ are two vertices from
    $C_u\cap C_v$, at most one of them incident with $F_1$ and at most
    one of them incident with~$F_2$. Then $(v,x,y)$ are
    counter-clockwise on $C_u$ if and only if $(u,x,y)$ are clockwise
    on~$C_v$; see Fig.~\ref{fig-H1}.
  \end{claim}

  To prove the claim, draw two curves $\gamma_x$ and $\gamma_y$
  connecting $u$ to $x$ and to $y$, respectively. Draw similarly two
  curves $\delta_x$ and $\delta_y$ from $v$ to $x$ and to~$y$. The
  endpoints of each of the curves appear in a common face of
  $\rskelemb$, so each curve can be drawn without intersecting any
  edge of~$\rskel$. Also, the assumptions of the claim guarantee that
  at most two of these curves can be in a common face of $\rskelemb$,
  and this happens only if they share an endpoint, so the curves can
  be drawn internally disjoint. Consider the closed curve formed by
  $\gamma_x$, $\delta_x$, and the edge $uv$, oriented in the direction
  $u\to x\to v\to u$. Suppose, without loss of generality, that $y$ is
  to the left of this closed curve. Then $\gamma_y$ is also to the
  left of it, and $(uv,\gamma_x, \gamma_y)$ appear in
  counter-clockwise order around $u$, so $(v,x,y)$ are
  counter-clockwise on~$C_u$. By analogous reasoning, $(u,x,y)$ are
  clockwise around~$v$. The claim is proved.

  \begin{figure}
    \hfil \includegraphics{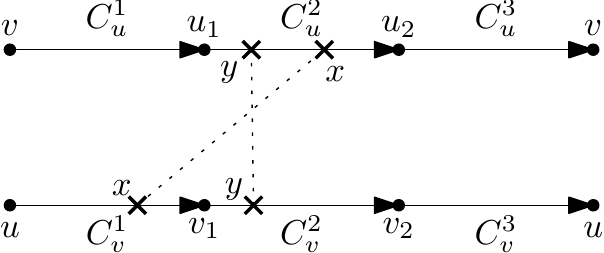}
    \caption{The two directed horizontal lines represent $C_u$ and
      $C_v$. A vertex $x$ appearing on both $C_u$ and $C_v$ is
      represented by a dotted line connecting its position on $C_u$
      with its position on $C_v$. Here is an example of a situation
      forbidden by the Claim, where two shared vertices $x$ and $y$
      appear in different order on the directed cycles $C_u$
      and~$C_v$.}\label{fig-H1}
  \end{figure}

  We now consider several cases depending on whether various parts of
  $C_u$ share vertices with parts of $C_v$.

  \textbf{Case A.} $C_u^3$ shares a vertex $x$ with~$C_v^1$.  Consider
  a walk starting in $v$, following $C_u$ counter-clockwise through
  $u_1$ and $u_2$ until $x$, then following $C_v$ clockwise from $x$
  through $v_1$ and $v_2$ till~$u$; see Fig.~\ref{fig-HA}.  The above
  Claim~\ref{claim:clockwise} guarantees that this walk is actually a
  path (note that $x$ cannot belong to either $F_1$ or $F_2$).  This
  path corresponds to the first case in the statement of the lemma.
  Similarly, if $C_u^1\cap C_v^3$ is nonempty, a symmetric argument
  yields the second case of the lemma.

  \begin{figure}
    \hfil \includegraphics{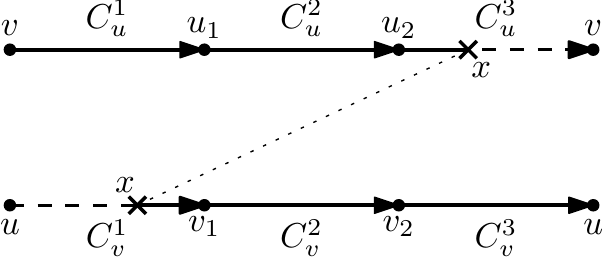}
    \caption{Case A of the proof of Lemma~\ref{lem-H}. The thick line
      represents the constructed walk.}\label{fig-HA}
  \end{figure}

  Assume for the rest of the proof that $C_u^3\cap
  C_v^1=\emptyset=C_u^1\cap C_v^3$.

  \textbf{Case B.} No internal vertex of $C_u^2$ belongs to $C_v^1\cup
  C_v^3$.  Define a walk in $\rskel$ by starting in $v_1$, following
  $C_v$ counter-clockwise until we reach the first vertex (call it
  $x$) that belongs to $C_u$, then following $C_u$ counter-clockwise
  through $u_1$ and $u_2$, until we reach the first vertex $y$ from
  $C_u^3\cap C_v$, then following $C_v$ from $y$ towards $v_2$ while
  avoiding $v_1$ and~$u$. Note that the vertices $x$ and $y$ must
  exist, because $F_1$ and $F_2$ each have at least one vertex from
  $C_u\cap C_v$. Note also that $x\in C_u^1$ and $y\not\in C_v^1$,
  otherwise we are in Case A.

  The walk defined above is again a path, it avoids $u$ and $v$, and
  by putting $w:=u_1$, we are in the situation of the third case of
  the lemma. Symmetrically, if no internal vertex of $C_v^2$ belongs
  to $C_u^1\cup C_u^3$, we obtain the fourth case of the lemma.

  \begin{figure}
    \hfil \includegraphics{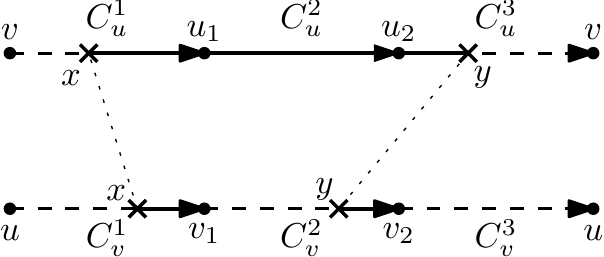}
    \caption{Case B of the proof of Lemma~\ref{lem-H}. Note that the
      vertex $y$ may also belong to~$C_v^3$.}\label{fig-HB}
  \end{figure}

  Suppose that none of the previous cases (and their symmetric
  variants) occurs.  What is left is the following situation.

  \textbf{Case C.} $C_u^2$ has an internal vertex $x$ belonging to
  $C_v^1\cup C_v^3$, and $C_v^2$ has an internal vertex $y$ belonging
  to $C_u^1\cup C_u^3$.  We cannot simultaneously have $x\in C_v^1$
  and $y\in C_u^1$ as that would contradict
  Claim~\ref{claim:clockwise}.  So assume that $x\in C_v^1$ and $y\in
  C_u^3$, the other case being symmetric. Consider a walk $W_1$ from
  $u_1$ along $C_u$ counter-clockwise through $u_2$, and let $z$ be
  the first vertex on $C_u$ after $u_2$ that belongs to~$C_v$. We must
  have $z\in C_v^2$, otherwise $z$ and $y$ violate
  Claim~\ref{claim:clockwise}.  Continue $W_1$ from $z$ until $v_2$
  clockwise along~$C_v$. Start another walk $W_2$ from $v_1$
  counterclockwise along $C_v$ and let $w$ be the first vertex of
  $C_u$ encountered.  Necessarily $w\in C_u^2$, otherwise $x$ and $w$
  violate Claim~\ref{claim:clockwise}. Therefore, $w\in W_1\cap W_2$,
  and $w$ is the only such vertex. This results in case 3 from the
  lemma, completing the proof.
\end{proof}

\begin{figure}
 \hfil \includegraphics{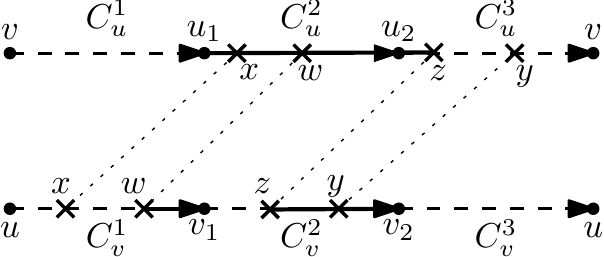}
\caption{Case C of the proof of Lemma~\ref{lem-H}. The thick lines
represent $W_1$ and~$W_2$.}\label{fig-HC}
\end{figure}

\begin{lemma}
  \label{lem:non-wrung}
  Let~$(G,H,\H)$ be a \peg and let $\rskel$ be the skeleton of an
  R-node of~$G$ such that~$\rskel$ has two distinct vertices~$u$
  and~$v$, each violating edge-compatibility of one of the embeddings
  of~$\rskel$.  If~$(G,H,\H)$ does not contain a wrung \peg, then $(G,H,\H)$
contains obstruction 4,
  5, or 6.
\end{lemma}

\begin{proof}
  As we have seen, $u$ must be incident to three $H$-edges
  $e_1=ux_1,e_2=ux_2,e_3=ux_3$ projecting to distinct $\rskel$-edges
  $e_1'=uu_1,e_2'=uu_2,e_3'=uu_3$, such that the cyclic order of the
  $e_i$'s in~$\H$ coincides with the cyclic order of $e_1',e_2',e_3'$
  in~$\rskelemb$.  We may assume that $(e_1,e_2,e_3)$ appear counter-clockwise
around $u$ in $\H$ and $(e'_1,e'_2,e'_3)$ appear counter-clockwise
around $u$ in $\rskelemb$.

Similarly, $v$ is incident to three $H$-edges
  $f_1=vy_1,f_2=vy_2,f_3=vy_3$ projecting to distinct $\rskel$-edges
  $f_1'=vv_1,f_2'=vv_2,f_3'=vv_3$ whose order in $\rskelnot$ agrees
  with~$\H$. Assume that $(f_1,f_2,f_3)$ appear counter-clockwise around $v$ in
$\H$ and $(f'_1,f'_2,f'_3)$ appear counter-clockwise around $v$ in $\rskelnot$,
and therefore clockwise around $v$ in~$\rskelemb$.

If all the edges $e_1',e_2',e_3'$ and $f_1',f_2',f_3'$ are distinct,
then $(G,H,\H)$ contains a wrung \peg by Observation~\ref{obs:wrung}.
Hence, one edge $e_i'$ must coincide with an edge $f_j'$.  After
possibly renaming the edges, we can assume $e_3' = f_3'$, and hence
$u_3=v$ and $v_3=u$.  Moreover, we may assume that the pertinent graph
of $e_3'$ does not contain a path from $u$ to $v$ that would contain
both $e_3$ and $f_3$, otherwise we again obtain a wrung \peg. Since
any vertex of the pertinent graph of $e_3'$ lies on a $G$-path from
$u$ to $v$ that projects into $e_3'$, we conclude that the pertinent
graph of $e_3'$ has a cycle that contains the two edges $e_3$ and
$f_3$.

After performing relaxations in $(G,H,\H)$ if necessary, we may assume that the
graph $H$ only contains the edges $e_i$ and $f_i$ for $i=1,2,3$, and the
vertices incident to these edges. We may also assume, after performing deletions
and contractions if necessary, that the pertinent graph of $e_3'$ is the
four-cycle formed by the edges $ux_3$, $x_3v$, $vy_3$ and $y_3u$. Furthermore,
we may assume that the pertinent graph of $e'_1$ is a path of length at most
two containing the edge $e_1=ux_1$, and similarly for $e'_2$, $f'_1$ and
$f'_2$. In fact, if the four edges $e'_1$, $e'_2$, $f'_1$ and $f'_2$ do not
form a four-cycle in $\rskel$, we may assume that the pertinent graph of each of
them is a single edge $e_i$ or $f_i$; however, if the four edges form a
four-cycle in $\rskel$, we may not contract their pertinent graphs to a single
edge, as that would form a new cycle in $H$, which our contraction rules do not
allow. Lastly, the pertinent graph of any edge of $\rskel$ different from
$e'_i$ and $f'_i$ for $i=1,2,3$ may be contracted to a single edge. This means,
in particular, that any path in $\rskel$ that does not contain the edge $e'_3$
is the projection of a unique path of~$G$.

We now apply Lemma~\ref{lem-H} to the embedded graph $\rskelemb$
($u,v,u_1,u_2,v_1,v_2$ are named as in the lemma). Let us treat the
four cases of the lemma separately. In the first case, $\rskel$
contains a path $P'$ of the form $v\to u_1\to u_2 \to v_1\to v_2\to
u$. The existence of such path implies that $u_1$ is not the same
vertex as $v_1$ or $v_2$, and $v_2$ is not the same as $u_1$ or
$u_2$. In particular, the four edges $e_1$, $e_2$, $f_1$ and $f_2$ do
not form a four-cycle, and each of them has for pertinent graph a
single edge of~$H$. The path $P'$ is a projection of a $G$-path~$P$.
Performing contractions as needed, we may assume that $P$ does not
contain any other vertices apart from $v$, $u_1=x_1$, $u_2=x_2$,
$v_1=y_1$, $v_2=y_2$, and $u$.  Moreover, if the vertices $x_2$ and
$y_1$ are distinct, this implies that they belong to different
components of $H$, and we may contract the edge of $P$ that connects
them, to create a single vertex~$w$. After these contractions are
performed, we are left with a \peg $(G',H',\H')$, where
$V(G')=V(H')=\{u,v,x_1,x_3,w,y_2,y_3\}$, $E(H')=\{ux_1,ux_3,uw,
vw,vy_2,vy_3\}$, and
$E(G')=E(H')\cup\{vx_1,x_1w,wy_2,y_2u,uy_3,vx_3\}$. This \peg is the
obstruction~6.

Since the case 2 of Lemma~\ref{lem-H} is symmetric to case 1, let us proceed
directly to case 3. In this case, $\rskel$ has a vertex $w$ different from
$u_2$, and three paths $P_1\colon w\to u_2\to v_2$, $P_2\colon w\to u_1$, and
$P_3\colon w\to v_1$ which only share the vertex~$w$.  We may again assume that
the paths do not contain any other vertex except those listed above.  Let us
distinguish two subcases, depending on whether the four edges $e'_1$, $e'_2$,
$f'_1$ and $f'_2$ form a cycle in $\rskel$ or not.  

Assume first that the edges form a cycle $C\subseteq\rskel$, that is,
$u_1=v_1$; see Fig.~\ref{fig:lem-nonwrung-cycle}.  Suppose that at
least one of the four edges of $C$ contains more than one edge
of~$G$. After performing contractions in $G$, we may assume that there
is only one edge of $C$ that contains two $G$-edges, and three edges
of $C$ containing one $G$-edge each. This produces obstruction~5.

\begin{figure}[tb]
  \centering
  \subfigure[]{\includegraphics[page=1]{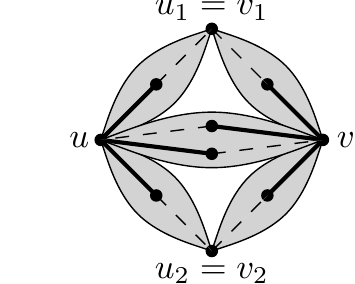}\label{fig:lem-nonwrung-cycle}}\hfil
  \subfigure[]{\includegraphics[page=2]{lem13-case3-cycle}\label{fig:lem-nonwrung-nocycle-1}}\hfil
  \subfigure[]{\includegraphics[page=3]{lem13-case3-cycle}\label{fig:lem-nonwrung-nocycle-2}}\hfil
  \caption{Illustration of case 3 in the proof of
    Lemma~\ref{lem:non-wrung}.  All figures show the planar
    embedding~$\mathcal{G}^+$ of~$G$.}
  \label{fig:lem-nonwrung-illustration}
\end{figure}

Finally, consider the case when the four edges $e'_1$, $e'_2$, $f'_1$ and
$f'_2$ do not form a cycle. Then
each of the four edges is a projection of a unique edge of~$G$. If
$u_2\neq v_2$ (Fig.~\ref{fig:lem-nonwrung-nocycle-1}), then we may
contract $P_2\cup P_3$ into a single vertex and assume that
$u_1=v_1$. On the other hand, if $u_2=v_2$
(Fig.~\ref{fig:lem-nonwrung-nocycle-2}), we may contract $P_2\cup P_3$
into a single edge~$u_1v_1$. Both situations correspond to
obstruction~5.  Note that this case is independent of the vertex~$w$
and whether it coincides with other vertices.  Since case 4 of
Lemma~\ref{lem-H} is symmetric to case 3, this covers all
possibilities.
\end{proof}

The following lemma shows that any minimal wrung \peg contains
one of the obstructions in Fig.~\ref{fig:obstructions}.  Although
the analysis is straight-forward, there exist many different cases.
For the sake of completeness, we provide a detailed proof.

\begin{lemma}
  \label{lem:min-wrung}
  A minimal wrung \peg contains one of the obstructions~1,2, or
  4--17.
\end{lemma}

\begin{proof}
  Let~$(G,H,\H)$ be a minimal wrung \peg, and let~$\G^+$
  be the planar embedding of~$G$.  By the definition of a wrung
  \peg, the graph~$H$ contains a vertex~$u$ with three adjacent
  vertices~$x_1,x_2,x_3$ occurring in this counterclockwise order
  around~$u$ both in~$\G^+$ and in~$\H$.  Similarly, $H$ has
  a vertex~$v$ distinct from~$u$ that has three neighbors $y_1,y_2,y_3$ 
occurring in clockwise order around $v$
  in~$\G^+$ but in counter-clockwise order in~$\H$.  Moreover, $G$ is a
  subdivision of a 3-connected graph~$G^*$. By Proposition~\ref{pro-wrung} the
graph~$H$ is connected, and all vertices of~$G$ also belong to~$H$.  

The proof of the lemma is split into two parts.  
In the first part, we assume that~$H$ contains the
  edge~$uv$, whereas in the second part we assume that this edge does not  belong to~$H$. 

To prove the first part, assume without loss of generality that $v=x_3$ and
$u=y_3$. Let $\G^*$ be the embedding of $G^*$ inherited from~$\G^+$. Clearly, both
$u$ and $v$ are vertices of $G^*$, but some of the vertices $x_i$ and $y_i$ may
be subdivision vertices in $G$ and therefore do not appear in $G^*$. Let $x_1'$
be the vertex of $G^*$ defined as follows: if $x_1$ belongs to $G^*$, then put
$x_1'=x_1$, otherwise choose $x_1'$ in such a way that $x_1$ is subdividing the
edge $ux_1'$ of $G^*$. The vertices $x_2'$, $y_1'$ and $y_2'$ are defined
analogously. Note that $x_1'$ can only be equal to $x_1$, $y_1$ or $y_2$ ---
if $x_1'$ were equal to $u$, $v$ or $x_2$, then $G^*$ would
have a multiple edge or a loop. Analogous restrictions hold for $x'_2$, $y'_1$
and~$y'_2$ as well. Also, $x'_1$ is not equal to $x'_2$, and
$y'_1$ is not equal to $y'_2$, otherwise we again get a
multiple edge in~$G^*$.

By construction, $v$, $x_1'$, and $x_2'$ are neighbors of $u$ in $G^*$
and their order around $u$ in $\G^*$ is the same as the order of $v$, $x_1$ and
$x_2$ in~$\G^+$, and similarly for the $y_i$'s. We apply
Lemma~\ref{lem-H} to the graph $G^*$, with $x'_i$ playing the role of $u_i$ and
$y'_i$ playing the role of~$v_i$. 

If the first case of Lemma~\ref{lem-H} occurs, we find in $G^*$
a path $P$ of the form $v\to x_1'\to x_2'\to y_1'\to y_2' \to u$. Let us distinguish
several possibilities, depending on whether the vertices $x_2'$ and $y_1'$ are
distinct or not. If the two vertices are distinct, this implies that all four
vertices $x'_1$, $x'_2$, $y'_1$ and $y'_2$ are distinct, and therefore $G^*=G$.
The path $P$ then shows that $(G,H,\H)$ contains obstruction~7. Suppose now
that $x'_2=y'_1$. Then at least one of the two vertices $x_2$ and $y_1$ must be
equal to $x'_2$. If both are equal to $x'_2$, we get obstruction~17, while if
exactly one of them is equal to $x'_2$, we obtain obstruction~10.

The second case of Lemma~\ref{lem-H} is symmetric to the first one. Let us deal
with the third case. We then have a vertex $w\in G^*$ and three paths
$P_1\colon w\to x'_2\to y'_2$, $P_2\colon w\to x'_1$ and $P_3\colon w\to y'_1$.
The three paths only share the vertex~$w$, which is distinct from~$x'_2$, 
and therefore also of $y'_2$. This
means, in particular, that neither $x'_1$ nor $y'_1$ may coincide with any of
$x'_2$ or $y'_2$. Note also that the vertex $w$ is equal to $x'_1$ or to $y'_1$,
because the three paths $P_1$, $P_2$ and $P_3$ avoid $u$ and $v$ by Lemma~\ref{lem-H}.
If $x'_1$ and $y'_1$ are distinct, then $x_1=x'_1$ and $y_1=y'_1$, and moreover
$x_1$ and $y_1$ are connected by a $G$-edge. If $x'_1=y'_1$, then either
$x_1=y_1$ or $x_1$ is connected by a $G$-edge to~$y_1$. In any case, $x_1$ and
$y_1$ are either identical or adjacent in~$G$. By a similar reasoning, $x_2$ and
$y_2$  are either identical or adjacent in~$G$. 

If $x_1$ and $y_1$ are identical, then $x_2$ and $y_2$ are different,
because in a wrung \peg $(G,H,\H)$ the graph $H$ must have at least one vertex
of degree one. Therefore, if $x_1=y_1$ then $x_2y_2\in E(G)$ and since we also
have $x_1x_2\in E(G)$ because of path $P_1$, we get obstruction~1. The case
when $x_2=y_2$ is analogous. If $x_1\neq y_1$ and $x_2\neq y_2$, then
$x_1y_1\in E(G)$ and $x_2y_2\in E(G)$, and by contracting the edge $uv$ we get
obstruction~2. 

Since case 4 of Lemma~\ref{lem-H} is symmetric to case 3, this completes the
analysis of wrung \pegs with the property $uv\in E(H)$.

In the second part of the proof, we consider minimal wrung \pegs
where~$u$ and~$v$ are not adjacent in~$H$.  Since~$H$ is connected,
one of the~$x_i$ must coincide with one of the $y_j$ and after
renumbering them, we may assume that~$x_3=y_3$.  To obtain a more
symmetric notation where this vertex is not notationally biased
towards~$u$ or~$v$, we denote it by~$w$.  We now make a case
distinction, based on which of the vertices~$x_1,x_2,y_1$ and~$y_2$
are distinct and which ones are subdivision vertices.  Note that~$w$
may not be a subdivision vertex, otherwise we could contract its
incident edges to the edge~$uv$ to obtain a smaller wrung \peg. The
overall case analysis works as follows.

  \begin{enumerate}[I)]
  \item Some of the vertices~$x_1,x_2,y_1,y_2$ coincide.
    
    As before, by symmetry we can assume that~$x_1$ coincides
    with one of~$y_1$ or~$y_2$.

    \begin{enumerate}[a)]
    \item If~$x_1=y_1$, then $(G,H,\H)$ contains obstruction~1 or~4.
    \item If~$x_1=y_2$, then $(G,H,\H)$ contains obstruction~16.
    \end{enumerate}

  \item The vertices~$x_1,x_2,y_1,y_2$ are all distinct and the wrung
    \peg has a vertex of degree~2.  By symmetry, we may assume that~$x_1$ is a  subdivision vertex.  We consider several subcases.
    
    \begin{enumerate}[a)]
    \item If the vertex~$y_1$ is also subdividing and~$G$ contains
      the edge~$x_1y_1$, then~$(G,H,\H)$ contains obstruction~2.
    \item If the vertex~$y_2$ is also subdividing and~$G$ contains the
      edge~$x_1y_2$, then~$(G,H,\H)$ contains obstruction~12.

      By symmetry these two cases cover all situations in which two
      subdivision vertices are connected by an edge.

    \item The graph~$G$ contains the edge~$x_1v$.

      Here we distinguish several subcases, depending on the position of $y_1$
and $y_2$ relative to the cycle $C=ux_1vwu$ in~$\G^+$.

      \begin{enumerate}[1)]
      \item If $y_1$ and $y_2$ are separated from $x_2$ by $C$, we obtain
obstruction~14.
      \item If $y_1$ is separated from $x_2$ but $y_2$ is not, we obtain
obstruction~5.
      \item If both $y_1$ and $y_2$ are on the same side of $C$ as $x_2$ we use
Lemma~\ref{lem-H} to get obstruction 5, 9 or~13.
      \end{enumerate}

    \item The vertex~$x_1$ is subdividing and adjacent to~$y_1$,
      but~$y_1$ is not subdividing.  Again, we consider subcases,
      based on other subdividing vertices and their adjacencies.

      \begin{enumerate}[1)]
      \item If~$x_2$ is subdividing and adjacent to~$y_2$, we
        find obstruction~2.
      \item If~$y_2$ is subdividing and adjacent to~$x_2$, we find
        obstruction~2.

        Note that this covers all the cases where any other vertex is
        subdividing.  If~$x_2$ was subdividing and adjacent to~$v$, we
        could exchange~$x_2$ with~$x_1$ to obtain an instance of~case
        IIb.  Analogously for~$y_2$ subdividing and adjacent to~$u$.
        Further,~$x_2$ cannot be subdividing and adjacent to~$y_1$ as
        this would create parallel subdivided edges~$ux_1y_1$
        and~$ux_2y_1$.  Therefore this covers all subcases where
        another vertex except~$x_1$ is a subdivision vertex.

      \item If no vertex besides~$x_1$ is subdividing, we find
        obstruction~5 or~15.
      \end{enumerate}

    \item If~$x_1$ is subdividing and adjacent to~$y_2$, but~$y_2$ is
      not subdividing, we find obstruction~9 or~13.

      This does not need specific subcases, as no other vertex can be
      subdividing.  If~$x_2$ was subdividing, it would be adjacent
      to~$v$ and by mirroring the embedding and exchanging~$x_1$
      with~$x_2$ and~$y_1$ with~$y_2$, we would arrive in case~IIc.
      Analogously for~$y_2$, which would have to be adjacent to~$u$.
      Hence we can assume that~$x_1$ is the only subdivision vertex.
    \end{enumerate}

  \item The vertices~$x_1,x_2,y_1,y_2$ are all distinct and all
    vertices in the wrung \peg have degree at least~3.

    \begin{enumerate}[a)]
    \item $G$ contains the edge~$uv$.

      We distinguish cases, based on the embedding of the edge~$uv$,
      by considering the relative positions of the cycle~$C=uwvu$ and
      the vertices~$x_1,x_2,y_1$ and~$y_2$ in~$\G^+$.
      \begin{enumerate}[1)]
      \item If all these vertices are on the same side of~$C$, we find
        obstruction~2,7,8 or~15.
      \item If one of these vertices is on one side and the others are
        on the other side, we obtain obstruction~5,9 or~13.  In this
        case, we may assume without loss of generality that~$C$
        separates~$x_1$ from the other vertices.
      \item If the cycle~$C$ separates~$x_1$ and~$x_2$ from the other
        vertices, we find obstruction~11.
      \item If the cycle~$C$ separates~$x_1$ and~$y_1$ from the other
        vertices, we find obstruction~2.

        All other cases are symmetric to one of these.
      \end{enumerate}
    \item $G$ does not contain the edge~$uv$.
      
      Here, we use the fact that~$G$ is 3-connected, and thus contains
      three vertex-disjoint paths~$p_1,p_2$ and~$p_3$
      from~$\{x_1,u,x_2\}$ to~$\{y_1,v,y_2\}$.  We distinguish cases,
      based on which vertex is connected to which.

      \begin{enumerate}[1)]
      \item The path~$p_1$ connects~$x_1$ to~$y_1$, $p_2$ connects~$u$
        to~$v$; we obtain obstruction~2.
      \item The path~$p_1$ connects~$x_1$ to~$y_1$, $p_2$ connects~$u$
        to~$y_2$; we obtain obstruction~2,5 or~9.
      \item The path~$p_1$ connects~$x_1$ to~$v$, $p_2$ connects~$x_2$
        to~$y_1$; we obtain obstruction~2,9,11,12 or~13.

        This covers all cases where the paths connect~$x_1$ to~$y_1$
        or to~$v$.  However, in the case where~$x_1$ is connected
        to~$y_2$, it is necessary that~$u$ connects to~$y_2$ and~$x_2$
        to~$v$, which after renaming the vertices is symmetric to the
        last case.
      \end{enumerate}
    \end{enumerate}
  \end{enumerate}

  Now we treat the above cases individually. 

  \textbf{Case I:} Some of the vertices~$x_1,x_2,y_1$ and~$y_2$
  coincide.  Since~$x_1 \ne x_2$ and~$y_1 \ne y_2$ and because~$H$ has
  at least one vertex of degree~1, at most two of these vertices can
  coincide.  By symmetry, we can assume that~$x_1$ coincides with one
  of the vertices~$y_1$ or~$y_2$.

  \textbf{Case Ia:} We have $x_1=y_1$.  The rotation schemes at~$u$
  and~$v$ imply that~$x_2$ and~$y_2$ are embedded on the same side of
  the cycle~$C$ formed by~$ux_1vwu$, and therefore are embedded on
  different sides of~$C$ in~$\H$.  Hence, if~$G$ contains the
  edge~$x_2y_2$, we obtain obstruction~1.

  Now assume that~$x_2y_2$ is not in~$G$.  Note that~$x_2$
  and~$y_2$ cannot both be subdivision vertices, as both would be part of a
  subdivision of the edge~$uv$.  Without loss of generality, we assume
  that~$x_2$ is not a subdivision vertex. 
  Assume first that $G$ has the edge~$y_2u$.  Then~$x_2$
  must connect to two vertices in the set~$\{x_1,v,w\}$.  The
  embedding of the edge~$y_2u$ implies that~$x_2$ cannot be adjacent
  to both~$w$ and~$x_1$, and hence we either have edges~$x_2x_1$
  and~$x_2v$ or~$x_2w$ and~$x_2v$.  In both cases, the fact
  that~$\{u,v\}$ is not a separator implies that the edge~$wx_1$ is
  in~$G$.  The cycle~$C$ together with the edges~$uy_2$, $y_2v$, $ux_2$, $x_2v$, and $wx_1$ then forms obstruction~4. 

We can therefore assume that~$y_2u$ is not an edge of $G$, and hence $y_2$ is
not a subdivision vertex. It follows that $y_2$ is adjacent to $w$ and~$x_1$.
Planarity then implies that $x_2$ cannot be adjacent to~$v$, so  $x_2$ is
adjacent to $w$ and to~$x_1$. To prevent $\{x_1,w\}$ from being a separator, $G$
must contain the edge~$uv$.
The cycle $C$ together with the edges $x_1y_2$, $y_2w$, $x_1x_2$, $x_2w$ and $uv$ again forms obstruction~4. This
closes the case~$x_1=y_1$.

\textbf{Case Ib:} We have~$x_1=y_2$. This means that $H$ contains a
four-cycle $C$ formed by the edges $uw$, $wv$, $vx_1$ and $x_1u$. We
will show that $C$ has two diagonally opposite vertices that are
adjacent to $x_2$, while the other two of its vertices are adjacent to
$y_1$. This will imply that $(G,H,\H)$ has obstruction~16.

We first show that every vertex of $C$ must be adjacent to at least one of
$x_2$ and $y_1$. This is clear for the vertices $u$ and $v$. To see this for
$x_1$ and $w$, note that neither of them may have degree two in $G$, as it
could then be contracted, contradicting the minimality of $(G,H,\H)$. Suppose
now that $wx_1$ is an edge of~$G$, and suppose without loss of generality that
it is embedded on the same side of $C$ as~$y_1$ in~$\G^+$. Then $y_1$ is
embedded inside the triangle $vwx_1$, and therefore it is adjacent to both $w$
and $x_1$. We conclude that each vertex of $C$ is adjacent to at least one of
$x_1$ and~$w$. 

Consequently, if e.g., $x_2$ is a subdivision vertex, then it is adjacent to
two diagonally opposite vertices of $C$, and therefore $y_1$ must be adjacent
to the other two vertices of $C$ forming obstruction~16. 

Suppose that neither $x_2$ nor $y_1$ is subdividing. If one of these
two vertices is adjacent to all the vertices of $C$, we easily obtain
obstruction~16. If both $x_2$ and $y_1$ are adjacent to three vertices
of $C$, and the vertex of $C$ not adjacent to $x_2$ is diagonally
opposite to the vertex of $C$ not adjacent to $y_1$, we get a
contradiction with 3-connectivity. The only remaining possibility is
that the vertex of $C$ not adjacent to $x_2$ is connected to the
vertex of $C$ not adjacent to $y_1$ by an edge of~$C$. This also
yields obstruction~16.

This concludes the treatment of the cases where~$x_1,x_2,y_1$ and~$y_2$
  are not distinct.

  \textbf{Case II:} The vertices~$x_1,x_2,y_1$ and~$y_2$ are distinct
  and one of the vertices is subdividing.  Without loss of generality
  we assume that~$x_1$ is subdividing, and we consider subcases based
  on the adjacencies of~$x_1$.

  \textbf{Case IIa:} The vertex~$y_1$ is also subdividing, and~$G$
  contains the edge~$x_1y_1$.  If~$x_2$ was a subdivision vertex, it
  could not be adjacent to~$v$, as the corresponding edge would be
  parallel to the edge subdivided by~$x_1$.  Therefore it would have
  to be adjacent to~$y_2$, which would give obstruction~2, by
  contracting the path~$uwv$ to a single vertex.  Hence, we can assume
  that~$x_2$ is not subdividing, and by a symmetric argument also
  that~$y_2$ is not subdividing.  Hence, each of them needs degree~3
  and since they are embedded on the same side of the
  cycle~$ux_1y_1vu$, they must be adjacent, which again results in
  obstruction~2; see Figure~\ref{fig:lem-minwrung-caseIIa}.

  \begin{figure}[tb]
    \centering
    \subfigure[]{\includegraphics[page=1]{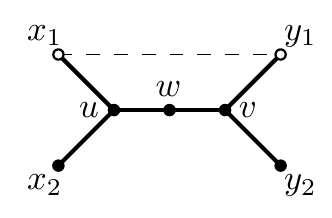}\label{fig:lem-minwrung-caseIIa}}\hfil
    \subfigure[]{\includegraphics[page=2]{lem14-illustration}\label{fig:lem-minwrung-caseIIIa1}}
    \subfigure[]{\includegraphics[page=3]{lem14-illustration}\label{fig:lem-minwrung-caseIIIa1-2}}

    \subfigure[]{\includegraphics[page=4]{lem14-illustration}\label{fig:lem-minwrung-caseIIIa2}}\hfil
    \subfigure[]{\includegraphics[page=5]{lem14-illustration}\label{fig:lem-minwrung-caseIIIa2-2}}\hfil
    \subfigure[]{\includegraphics[page=6]{lem14-illustration}\label{fig:lem-minwrung-caseIIIa2-3}}\hfil

    \subfigure[]{\includegraphics[page=7]{lem14-illustration}\label{fig:lem-minwrung-caseIIIa3}}\hfil
   \subfigure[]{\includegraphics[page=8]{lem14-illustration}\label{fig:lem-minwrung-caseIIIa4}}\hfil
   \subfigure[]{\includegraphics[page=9]{lem14-illustration}\label{fig:lem-minwrung-caseIIIb3}}\hfil

   \subfigure[]{\includegraphics[page=10]{lem14-illustration}\label{fig:lem-minwrung-caseIIIb3-2}}\hfil
   \subfigure[]{\includegraphics[page=11]{lem14-illustration}\label{fig:lem-minwrung-caseIIIb3-3}}\hfil

   \caption{Illustration of several cases of
     Lemma~\ref{lem:min-wrung}; hollow vertices are subdivision
     vertices, which have degree~2.  The edge~$uv$ from case IIIa is
     drawn as a dotted curve, all other edges of~$G$ are dashed.}
    \label{fig:lem-minwrung}
  \end{figure}

  \textbf{Case IIb:} The vertex~$y_2$ is also subdividing and~$G$
  contains the edge~$x_1y_2$.  As in the previous case it can be seen
  that neither of~$x_2$ and~$y_1$ may be subdividing.  Since they must
  be embedded on different sides of the cycle~$ux_1y_2vwu$
  in~$\G^+$, the graph~$G$ must contain the
  edges~$x_2w,x_2v,y_1u$ and~$y_1w$, in order for them to have
  degree~3.  Altogether, this yields obstruction~12.

  \textbf{Case IIc:} The vertex~$x_1$ is subdividing and adjacent
  to~$v$.  Therefore, $ux_1vwu$ is a four-cycle in~$G$. Let us call it $C$ and
orient it in the direction $u\to
x_1\to v\to w\to u$. Note that the vertex $x_2$ is to the left of~$C$. We
distinguish several cases based on the position of $y_1$ and $y_2$ relative
to~$C$ in~$\G^+$. Note that if $y_1$ is left of $C$ then so is $y_2$, because $y_1$, $y_2$ and~$w$ must be clockwise around~$v$.

\textbf{Case IIc1:} Both $y_1$ and $y_2$ are to the right of $C$.  This means
that $x_2$ cannot be subdividing and it must be adjacent to $w$ and to $v$.  We
also easily check that $G$ must contain all the edges $y_1y_2$, $wy_1$ and
$uy_2$, otherwise we get contradiction with 3-connectivity. This creates
obstruction~14. 

\textbf{Case IIc2:} The vertex $y_1$ is to the right of $C$ but $y_2$ is left.
Then $y_1$ must be adjacent to both $u$ and $w$.  We now show that $G$ has the
edge $x_2y_2$. If $x_2$ is a subdividing vertex, it cannot be adjacent to $v$
or $w$, as that would form a multiple edge, so it must be adjacent to $y_2$.
Similarly, if $y_2$ is subdividing, it must be adjacent to $x_2$. If neither
$x_2$ nor $y_2$ is subdividing, they must also be adjacent, otherwise they
would both have to be adjacent to $u$, $v$ and $w$, forming a $K_{3,3}$
subgraph in~$G$.  Thus, $x_2y_2\in E(G)$.  Moreover, at least one of $x_2$ and
$y_2$ is adjacent to $w$, otherwise $u$ and $v$ form a 2-cut. We then obtain
obstruction~5.

\textbf{Case IIc3:} Consider again the 3-connected graph $G^*$, and
define $x_2'$, $y_1'$ and $y_2'$ as in the first part of the proof of
this lemma.  Note that the vertex $w$ cannot be subdividing, and
therefore belongs to~$G^*$. Apply Lemma~\ref{lem-H} to $G^*$, taking
$u_1=x_2'$, $u_2=w$, $v_1=y_1'$ and $v_2=y_2'$.

The first case of Lemma~\ref{lem-H} yields a path $v\to x_2'\to w\to
y_1'\to y_2'\to u$. This implies that $x_2'=x_2$, $y_1'=y_1$ and
$y_2'=y_2$. This case however cannot occur, because the edges $wv$ and
$vy_1$ together with the path $w\to y_1$ form a cycle, and $y_2$ is on
the other side of this cycle than $u$, making it impossible to embed
the edge $y_2u$ in a planar way.  Thus this would contradict the
assumption that $\G^+$ is a planar embedding of~$G$.

In the second case of Lemma~\ref{lem-H}, we get a path $u\to y_1'\to
y_2'\to x_2'\to w\to v$. This implies $y_1'=y_1$.  We may have
$x_2'=y_2'$ or not, and if $x_2'=y_2'$, then we may have either
$x_2=x_2'$ or $y_2=x_2'$, but in any case, we know that $x_2$ and
$y_2$ must be adjacent in $G$ and at least one of them must be
adjacent to~$w$. The edges $x_1v$, $y_1u$, $x_2y_2$ and one of $wx_2$ or $wy_2$
together form obstruction~5.

In the third case of the lemma, we have a vertex $w'\in\{x'_2,y'_1\}$ and three
paths $w'\to w \to y_2'$, $w'\to x_2'$ and $w'\to y_1'$ sharing only the vertex
$w'$. This implies that $y_2'=y_2$. Note also that the paths $w'\to w$ and
$w'\to y'_1$ together with the (possibly subdivided) edges $y'_1v$ and $vw$ form
a cycle that separates $y_2$ from $u$, showing that $y_2$ and $u$ cannot be
adjacent. Consequently, $y_2$ must be adjacent to at least one of $x_2$ and
$y_1$.  Moreover, $x_2$ and $y_1$ must be adjacent to each other, as shown by
the existence of a path $x_2'\to y_1'$. Together with the edge $wy_2$, this
forms obstruction 9 or~13.

In the fourth case of the lemma, we have again $w'\in\{x'_2,y'_1\}$, this time
with paths $w'\to y'_2\to w$, $w'\to y'_1$ and $w'\to x'_2$. This again shows
that $y'_2=y_2$, and that $w$ is adjacent to $y_2$, $x_2$ is adjacent to $y_1$,
and $y_2$ is adjacent to $x_2$ or~$y_1$. This forms again obstruction 9
or~13.

  \textbf{Case IId:} The vertex~$x_1$ is subdividing and adjacent
  to~$y_1$, but~$y_1$ is not subdividing.  We distinguish subcases
  based on whether other vertices are subdividing.

  \textbf{Case IId1:} The vertex~$x_2$ is subdividing and adjacent
  to~$y_2$, then by contracting the path~$uwv$ to a single vertex, we
  obtain obstruction~2.
  
  \textbf{Case IId2:} The vertex~$y_2$ is subdividing and adjacent
  to~$x_2$.  As in the previous case, we obtain obstruction~2.  As
  argued in the description of the case analysis, this covers all
  instances, where a vertex besides~$x_1$ is subdividing.

  \textbf{Case IId3:} The vertex~$x_1$ is subdividing and adjacent
  to~$y_1$, and no other vertex is subdividing.  Clearly, if~$x_2y_2
  \in G$, we obtain obstruction~2.  We can therefore assume that this
  is not the case.  Hence,~$x_2$ must be adjacent to at least two
  vertices in the set~$\{w,v,y_1\}$.

  First, assume that~$x_2w$ and~$x_2v$ are in~$G$.  The edge~$x_2v$
  admits two distinct embeddings, however in one of them the
  cycle~$x_2wvx_2$ would enclose only the vertex~$y_2$, which would
  imply the existence of the excluded edge~$x_2y_2$.  We can therefore
  assume that all vertices that do not belong to the cycle are on the
  same side of it.  Since~$y_2$ has degree~3 in~$G$, this implies the
  existence of the edges~$y_1y_2$ and~$y_2u$.  Since~$\{u,v\}$ still
  would form a separating pair, the edge~$wy_1$ must also be present.
  Omitting the edge~$x_2w$ yields obstruction~5.

  Next, assume that~$x_2w$ and~$x_2y_1$ are in~$G$.  Since~$y_2$ has
  degree~3, and~$y_1y_2$ is excluded, $G$ must contain the
  edges~$y_1y_2$ and~$y_2w$.  Until now~$\{w,y_1\}$ would still form a
  separating pair.  Since~$x_2y_2$ is excluded~$G$ contains the
  edge~$uv$.  Altogether, this is obstruction~15.

  Finally, assume that~$x_2v$ and~$x_2y_1$ are in~$G$.  Again, only
  one of the embeddings of~$x_2v$ does not force the edge~$x_2y_2$ to
  be present.  Since~$y_2$ has degree~3, it must be adjacent to~$u$
  and~$w$.  As~$\{u,v\}$ would still be a separating pair, we also
  obtain the edge~$wy_1$.  Omitting the edges~$y_2w$ and~$x_2y_1$
  yields obstruction~5.  This finishes the case where~$x_1$ is
  subdividing and adjacent to~$y_1$.

  \textbf{Case IIe:} The vertex~$x_1$ is subdividing and adjacent
  to~$y_2$, which is not subdividing.  No other vertex may be
  subdividing, as this would be symmetric to case IIc.  Clearly,~$y_1$
  must be adjacent to two vertices in the set~$\{u,w,y_2\}$.  We
  consider several cases.

  Assume that~$y_1$ is adjacent to~$y_2$ and~$w$.  If~$G$ contains
  $x_2v$, this forms obstruction~9.  Therefore, assume that~$x_2v$ is
  not in~$G$.  It follows that~$x_2$ is adjacent to~$w$ and~$y_2$.
  Since~$\{w,y_2\}$ still forms a separator, we also get edge~$y_1u$.
  By omitting~$y_1y_2$ and $y_1w$, we obtain obstruction~13.

  Next, assume that~$y_1$ is adjacent to~$y_2$ and~$u$ but not~$w$.  
  Then $x_2$ must be adjacent to $y_2$, otherwise $uv$ would form a separator or $x_2$ would have degree~2. Also, $w$ must be adjacent to $x_2$, otherwise $x_2$ or $w$ would have degree~2.  This yields obstruction~13.

  Finally, assume that~$y_1$ is adjacent to~$u$ and to~$w$, and not
  to~$y_2$ (otherwise one of the previous cases would apply).
  Clearly, $x_2$ must be adjacent to two of the three vertices~$w,v$
  and~$y_2$.  It is not possible that~$x_2$ is only adjacent to~$w$
  and~$v$, since~$\{u,v\}$ would still form a separating pair.
  Hence~$x_2y_2 \in E(G)$.  Also $x_2$ must be adjacent to $w$, otherwise $uv$ would be a separator or $x_2$ would have degree 2. This gives obstruction~13.

 This closes the
  case where~$x_1$ is subdividing and adjacent to~$y_2$, and thus all
  cases where~$G$ has a subdividing vertex.

  \textbf{Case III:} The graph~$G$ does not have any subdividing
  vertices, and thus is 3-connected.  We distinguish two subcases,
  based on whether~$G$ contains the edge~$uv$.

  \textbf{Case IIIa:} The graph~$G$ is 3-connected and contains the
  edge~$uv$.  The edge~$uv$ can be embedded in a variety of different
  ways in~$\G^+$.  We distinguish cases, based on this embedding,
  in particular, the relative position of the cycle~$C = uwvu$ and the
  vertices~$x_1,x_2,y_1$ and~$y_2$.

  \textbf{Case IIIa1:} First, assume that all these vertices are
  embedded on the same side of~$C$.  We apply Lemma~\ref{lem-H} on the
  vertices~$u$ and~$v$ with~$u_i = x_i$ and~$v_i=y_i$ for~$i=1,2$.
  Suppose that case~1 of the lemma applies.  Then we obtain a simple
  path~$v \rightarrow x_1 \rightarrow x_2 \rightarrow y_1 \rightarrow
  y_2 \rightarrow u$.  If~$w$ is not contained in any of the subpaths,
  we can contract~$wv$ and obtain obstruction~7.  Further, by
  planarity,~$w$ cannot subdivide any of the subpaths~$x_1 \rightarrow
  x_2, y_1 \rightarrow y_2$ or~$x_2 \rightarrow y_2$.  Hence, it must
  subdivide~$x_1 \rightarrow v$ or $y_2 \rightarrow u$, which is
  completely symmetric, and we assume without loss of generality
  that~$w$ subdivides the path~$x_1 \rightarrow v$ and all other
  subpaths consist of a single edge, each.  Again, contracting~$wv$
  yields obstruction~7.

  Case~2 of the lemma is completely symmetric, we therefore continue
  with case~3.  We obtain a vertex~$w' \ne x_2,y_2,u,v$ and paths~$w'
  \rightarrow x_2 \rightarrow y_2, w' \rightarrow x_1, w' \rightarrow
  y_1$.  Further,~$w$ may subdivide one of these paths.  First of
  all,~$w'$ must coincide with either of~$x_1,y_1$ or~$w$.

  If~$w' = x_1$, we obtain obstruction~2, unless~$w$ subdivides the
  subpath~$x_1 \rightarrow y_1$ (or symmetrically~$x_2 \rightarrow
  y_2$).  We therefore assume that~$w$ subdivides~$x_1y_1$, and thus
  all other paths consist of single edges; see
  Fig.~\ref{fig:lem-minwrung-caseIIIa1}.  Since $y_1$ must have degree
  at least 3, at least one of the three edges $y_1y_2$, $y_1x_2$ or
  $y_1x_1$ (note that~$y_1u$ is not possible in a planar way) must be
  present, resulting in obstructions 8, 15 and 2, respectively.

  If~$w'=y_1$, we again obtain obstruction~2, unless~$w$ subdivides
  either~$x_1 \rightarrow y_1$ or~$x_2 \rightarrow y_2$.  Again these
  cases are symmetric, and we assume~$w$ subdivides~$x_1 \rightarrow
  y_1$.  Since~$x_1$ has degree~3, it is either adjacent to~$y_1$ or
  to~$x_2$.  This leads to obstructions~2 and~15, respectively.

  If~$w'=w$, the situation is illustrated in
  Fig.~\ref{fig:lem-minwrung-caseIIIa1-2}, and $x_1x_2$ must be in~$G$
  because of planarity and since~$x_1$ has degree at least~3.
  Further, $y_1$ has degree at least~3 and therefore~$G$ either
  contains~$y_1y_2$ or~$y_1x_2$, which yields obstructions~8 and~15,
  respectively.

  Since case~4 of the
  Lemma~\ref{lem-H} is completely symmetric, this closes the case
  where the cycle~$C = uwvu$ bounds an empty triangle.

  \textbf{Case IIIa2:} The graph~$G$ is 3-connected, contains the
  edge~$uv$ and the cycle~$C=uwvu$ is embedded so that it
  separates~$x_1$ from the vertices~$x_2,y_1$ and~$y_2$. In this case,
  $x_1$ must be adjacent to $v$ and $w$ in~$G$. We apply
  Lemma~\ref{lem-H} to vertices~$u$ and~$v$
  with~$u_1=x_2,u_2=w,v_1=y_1$ and~$v_2=y_2$.

  In case~1 of the lemma, we obtain a path $v\to x_2\to w\to y_1\to
  y_2\to u$; see Fig.~\ref{fig:lem-minwrung-caseIIIa2}, where the
  edge~$uv$ is drawn as a dotted curve.  This path cannot be embedded
  in a planar way into $G$ (without changing the embedding of~$uv$,
  which is assumed to be fixed), so this case does not occur.

  In case~2 of the lemma, we obtain edges~$uy_1,y_1y_2,y_2x_2,x_2w$
  and~$wv$.  Together with the edge~$x_1v$,  this
  forms an instance of obstruction~5.

  In case~3 of the lemma, we obtain a vertex~$w' \ne w$ paths~$w'
  \rightarrow w \rightarrow y_2$, $w' \rightarrow x_2$ and~$w'
  \rightarrow y_1$.  The vertex~$x_1$ cannot subdivide any of these
  paths and also~$w'=x_1$ is not possible by planarity.

  Suppose that~$w'=x_2$, we thus obtain edges~$x_2w,y_2w$
  and~$x_2y_1$.  We already know that $x_1v$ is an edge of~$G$.
  Figure~\ref{fig:lem-minwrung-caseIIIa2-2} illustrates the situation,
  where~$x_1v$ is drawn dotted, the remaining edges are drawn as
  dashed curves.  Further,~$y_2$ must have another adjacency, which
  can by planarity either be~$y_1$ or~$x_2$.  In the former case, we
  find obstruction~9, in the latter obstruction~13.

  Suppose that~$w'=y_1$, we then have edges~$x_2y_1,y_2w$ and~$y_1w$;
  see Fig.~\ref{fig:lem-minwrung-caseIIIa2-3}.  Planarity and the
  degree of~$y_2$ imply that~$y_1y_2$ is in~$G$, and as above~$x_1v$
  is in~$G$.  Together, this yields obstruction~9.

  In case~4 of the lemma, we obtain a vertex~$w' \ne y_2$ and
  paths~$w' \rightarrow y_2 \rightarrow w, w' \rightarrow y_1$ and~$w'
  \rightarrow x_2$.  Again, by planarity,~$x_1$ may neither coincide
  with~$w'$ nor subdivide any of the paths.

  Suppose that~$w' = x_2$, and we thus have edges~$x_2y_2,y_2w$
  and~$x_2y_1$.  As above we find the edge~$x_1v$.  This gives obstruction~13.

  Suppose that~$w'=y_1$.  This yields edges~$x_2y_1,y_2w$
  and~$y_1y_2$.  As above we find edge~$x_1v$, and thus obstruction~9.
  This closes the case where exactly one of the vertices is enclosed
  by the cycle~$uwvu$.

  \textbf{Case IIIa3:} The edge~$uv$ is embedded so that the
  cycle~$C=uwvu$ separates~$x_1$ and~$x_2$ from~$y_1$ and~$y_2$.  We
  may assume that $x_1$ and $x_2$ are to the right of the directed
  cycle $u\to v\to w\to u$, while $y_1$ and $y_2$ are to its left.
  Note that in this case, the vertices $x_1$ and $x_2$ must be
  adjacent, because otherwise both of them would have to be adjacent
  to $u,v$ and $w$, contradicting
  planarity. Figure~\ref{fig:lem-minwrung-caseIIIa3} illustrates the
  current situation.  By 3-connectedness, both $v$ and $w$ must be
  adjacent to at least one of $x_1$ and $x_2$, and each $x_i$ must be
  adjacent to at least one of $v$ and~$w$. Planarity implies that
  $x_1w$ and $x_2v$ must both be edges of~$G$.  An analogous argument
  for $y_1$ and $y_2$ implies that $uy_1$, $y_1y_2$ and $y_2w$ are all
  edges of~$G$. This forms obstruction~11.

  \textbf{Case IIIa4:} The edge~$uv$ is embedded so that the
  cycle~$C=uwvu$ separates~$x_1$ and~$y_1$ from~$x_2$ and~$y_2$, see Fig.~\ref{fig:lem-minwrung-caseIIIa4}
  Clearly, $x_1$ and~$y_1$ both need degree~3.  However,~$C$ has only
  vertices~$u,w$ and~$v$, and $x_1$ and~$y_1$ cannot both be
  connected to all vertices of~$C$ in a planar way (otherwise we could
  find a planar drawing of~$K_{3,3}$).  Hence, $G$ must contain the
  edge~$x_1y_1$, and by a symmetric argument also~$x_2,y_2$, which
  results in obstruction~2.

  \textbf{Case IIIb:} The graph~$G$ is 3-connected and does not
  contain the edge~$uv$.  Since~$G$ is 3-connected it contains three
  vertex-disjoint paths~$p_1,p_2$ and~$p_3$ from~$\{x_1,u,x_2\}$
  to~$\{y_1,v,y_2\}$.  We distinguish cases, based on which endpoints
  are connected by the paths.

  \textbf{Case IIIb1:} The path~$p_1$ connects~$x_1$ to~$y_1$
  and~$p_2$ connects~$u$ to~$v$.  Clearly,~$p_3$ must then
  connect~$x_2$ to~$y_2$.  Since~$G$ does not contain the edge~$uv$,
  $p_2$ must contain the vertex~$w$, which implies that~$p_1$
  and~$p_3$ consist of a single edge, each.  This yields
  obstruction~2.

  \textbf{Case IIIb2:} The path~$p_1$ connects $x_1$ to~$y_1$
  and~$p_2$ connects~$u$ to~$y_2$.  Clearly,~$p_3$ must then
  connect~$y_1$ to~$v$. We
  distinguish cases, based on whether~$w$ is contained in one of these
  paths.

  First, suppose that none of these paths contains~$w$, that is,
  each of them consists of a single edge.  The edges~$x_2v$ and~$y_2u$
  admit two different embeddings that are completely symmetric.  We
  therefore assume that~$uwvy_2u$ bounds a face in the graph
  consisting of~$H$ and the paths~$p_i, i=1,2,3$ with the embedding
  inherited from~$\G^+$.  Then the cycle~$x_2vwux_2$
  separates~$y_2$ from~$x_1$ and~$y_1$.  Since~$y_2$ needs degree at
  least~3, we either have~$x_2y_2$ or~$y_2w$ in~$G$.  The former would
  yield obstruction~2, thus we assume the latter.  However,
  still~$\{u,v\}$ would form a separating pair, thus implying that~$w$
  needs to be adjacent to either~$x_1$ or to~$y_1$.  In both cases we
  obtain obstruction~5.  

Hence, we can assume that~$w$ is contained in one of the paths $p_1$, $p_2$ and~$p_3$.
  First, assume that~$w$ is contained in~$p_1$.  Again, the embedding
  choices offered by the edges~$x_2v$ and~$y_2u$ are completely
  symmetric, and as above we assume that the cycle~$x_2vwux_2$
  separates~$y_2$ from~$x_1$ and~$y_1$.  If~$x_1y_1$ was in~$G$, we
  could replace the path~$p_1$ by this edge and the previous case
  would apply; we therefore assume that this is not the case.
  If~$x_1v$ was in~$G$, then the fact that~$y_1$ needs another
  adjacency would again imply that~$x_1y_1$ is in~$G$.  Since~$x_1$
  needs one more adjacency, the only option is the edge~$x_1x_2$.
  Similarly, for~$y_1$ the only option is the edge~$y_1x_2$.  The
  graph then contains obstruction~9.

  Second, assume that~$w$ is contained in~$p_2$.  We then have
  edges~$x_1y_1,y_2w$ and~$x_2v$, which have a unique embedding, in
  which the cycle~$x_2vwux_2$ separates~$y_2$ from~$x_1$ and~$y_1$.
  Since~$y_2$ needs degree at least~3, it must be adjacent to
  either~$x_2$ or to~$u$.  In the former case, we again get
  obstruction~2.  In the latter case, we can replace the path~$p_2$
  with the edge~$uy_2$ and we are in the case that~$w$ is not
  contained in any of the three paths.

  Finally, the case that~$w$ is contained in~$p_3$ is completely
  symmetric to the previous case.  This closes the case where~$p_1$
  connects~$x_1$ to~$y_1$ and~$p_2$ connects~$u$ to~$y_2$.

  \textbf{Case IIIb3:} The path~$p_1$ connects~$x_1$ to~$v$, $p_2$
  connects~$u$ to~$y_2$, and thus~$p_3$ connects~$x_2$ to~$y_1$.
  Again, we consider cases based on whether~$w$ is contained in one of
  these paths.  First, suppose that~$w$ is not contained in any of
  these paths.  Then $G$ contains the edges~$x_2y_1,x_1v$ and~$y_2u$,
  whose embedding in $\G$ is uniquely determined.  Since~$w$ has
  degree at least~3 it must be connected to~$x_1$ or~$y_2$ in~$G$.
  Both cases are completely symmetric, and we assume~$wy_2$ is in~$G$;
  see Fig.~\ref{fig:lem-minwrung-caseIIIb3}.  So far, the
  set~$\{u,v\}$ would still form a separating pair.  Hence~$G$
  contains at least one of the edges~$y_1y_2$,~$y_2x_2$ and~$wx_1$,
  resulting in obstruction 9, 13 and 12, respectively.

  We can therefore assume that~$w$ subdivides one of the paths.
  Suppose that~$w$ subdivides~$p_1$, that is, we have
  edges~$x_1w,uy_2$ and~$x_2y_1$; see
  Fig.~\ref{fig:lem-minwrung-caseIIIb3-2}.  Clearly,~$x_1$ must be
  adjacent to one of~$x_2,v$ and~$y_1$.  However,~$x_1x_2$ produces
  obstruction~9,~$x_1y_1$ produces obstruction~13, and if~$x_1v$ is
  in~$G$, we replace the path~$p_1$ by~$x_1v$ and are in the case
  where~$w$ is not contained in any of the three paths.

  The case that~$w$ is contained in~$p_2$ is completely symmetric, we
  therefore assume it is contained in~$p_3$, and we thus have
  edges~$x_1v,y_2u,x_2w$ and~$y_1w$.  There are two possible planar
  embeddings.  One, in which~$x_1v$ comes after~$vw$ and before~$vy_2$
  in counter-clockwise direction around~$v$, and one in which~$x_1v$
  comes after~$vy_2$ and before~$vy_1$ in counter-clockwise direction.

  We start with the first case; see
  Fig.~\ref{fig:lem-minwrung-caseIIIb3-3}.  Suppose that~$G$ contains
  the edge~$x_1y_2$.  The vertex~$x_2$ needs another adjacency, which
  can either be~$x_1$ or~$v$.  Analogously,~$y_1$ needs to be adjacent
  to either~$u$ or~$y_2$.  The
  combination~$x_1x_2$ and~$y_1y_2$ gives obstruction~11, any of the
  combination~$x_2v,y_1y_2$ and~$x_1x_2,y_1u$ gives obstruction~9, and
  the combination $x_2v$ and $y_2u$ gives obstruction 12.

  Now assume that~$G$ does not contain the
  edge~$x_1y_2$.  Then~$x_1$ needs to be adjacent to~$x_2$; its only
  options are~$x_2$ and~$w$, but in the latter case~$x_2$ still needs
  an edge and the only possibility is~$x_1$.  Analogously,~$y_1y_2$
  must be in~$G$.  Together this forms obstruction~11.

  Now suppose that the second embedding applies, and~$x_1v$ comes
  after~$vy_2$ and before~$vy_1$ in counter-clockwise direction
  around~$v$.  If~$G$ contains~$x_1y_1$, then it cannot
  contain~$x_2y_2$ as well, otherwise we would get obstruction~2.
  Hence,~$x_2$ must be adjacent to~$v$.  Since~$y_2$ is not adjacent
  to~$x_2$, it must be adjacent to~$x_1$, and we thus obtain
  obstruction~13.  We can hence assume that~$G$ does not
  contain~$x_1y_1$, and symmetrically also not~$x_2y_2$.  We hence get
  edges~$x_2v$ and~$y_1u$.  But then still~$\{u,v\}$ forms a
  separator, which shows that either~$x_1y_1$ or~$x_2y_2$ is in~$G$; a
  contradiction.  This finishes the last case, and thus concludes the
  proof.
\end{proof}

\subsubsection{$\rskel$ has an edge-compatible embedding} 

Assume now that the embedding $\rskelemb$ of the skeleton $\rskel$ is
edge-compatible but not cycle-compatible.  We first give a sketch of
our general proof strategy.  Our analysis of this situation strongly
relies on the concept of $C$-bridge, which has been previously used
by Juvan and Mohar in the study of embedding extensions on surfaces of
higher genus~\cite{jm-repeg-05}, and which is also employed (under the
name \emph{fragment}) by Demoucron, Malgrange and Pertuiset in their
planarity algorithm~\cite{dmp-rcrpt-64}.

Let $F$ be a graph and $C$ a cycle of~$F$.  A \emph{$C$-bridge} is
either a chord of~$C$, (i.e., an edge not belonging to $C$ whose vertices are on $C$) or a connected component of $F-C$, together with
all vertices and edges that connect it to~$C$.  A vertex of $C$ that is incident to an edge of a $C$-bridge $X$ is called an \emph{attachment} of~$X$. Let $\att(X)$ denote the set of attachments of~$X$.  A bridge that
consists of a single edge is \emph{trivial}.  

In our argument, we focus on cycles in $\rskel$ that are projections of cycles
in $H$.  Notice that in this
case, any nontrivial bridge in $\rskel$ has at least three attachments,
because $\rskel$ is 3-connected.
If $\rskelemb$ violates cycle-compatibility, it means that $H$ must
contain a cycle $C'$ that projects to a cycle $C$ of $\rskel$, and $\rskelemb$
has a $C$-bridge that is embedded on the `wrong' side of~$C$.  We concentrate on
the substructures that enforce such `wrong' position for a given $C$-bridge, and
use them to locate planarity obstructions. 

Let us describe the argument in more detail.  Suppose again that $C'$ is a cycle
of $H$ that projects to a cycle $C$ of~$\rskel$.  Let $x$ be a vertex of $H$
that does not belong to any $\rskel$-edge belonging to~$C$.  We say that $x$ is
\emph{happy} with $C'$, if its embedding in $\rskelemb$ does not violate
cycle-compatibility with respect to the cycle $C'$, i.e., $x$ is to the left of
$C'$ in $\H$ if and only if $x$ is to the left of $C$ in~$\rskelemb$.
Otherwise we say that $x$ is unhappy with~$C'$.  We say that a $C$-bridge $B$ of
$\rskel$ is happy with $C'$ if there is a vertex $x$ happy with $C'$ that
projects into $B$, and similarly for unhappy bridges.  A $C$-bridge that is
neither happy nor unhappy is \emph{indifferent}.

In our analysis of cycle-incompatible skeletons, we establish the
following facts.
\begin{itemize}
\itemsep=-.5ex
\item With $C$ and $C'$ as above, if a single $C$-bridge is both happy
  and unhappy with $C'$, then $(G,H,\H)$ contains obstruction 1 or 4
  (Lemma~\ref{lem-schizo}).
\item Let us say that the cycle $C'$ is \emph{happy} if at least one
  $C$-bridge is happy with $C'$, and it is \emph{unhappy} if at least
  one $C$-bridge is unhappy with~$C'$.  If $C'$ is both happy and
  unhappy, then $(G,H,\H)$ contains obstruction 4, obstruction 16, or
  an alternating-chain obstruction (Lemma~\ref{lem-schizo-cyc}).
\item Assume that the situation described above does not arise.
  Assume further that $C'$ is an unhappy cycle of~$\H$.  Then any edge
  of $H$ incident to a vertex of $C'$ must project into an
  $\rskel$-edge belonging to $C$, unless $(G,H,\H)$ contains
  obstruction~3 or one of the obstructions from the previous
  item. Note that this implies, in particular, that the vertices of
  $C$ impose no edge-compatibility constraints
  (Lemma~\ref{lem-cyc-fork}).
\item If $C'_1$ and $C'_2$ are two facial cycles of $\H$ whose
  projection is the same cycle $C$ of $\rskel$, then any $C$-bridge is
  happy with $C'_1$ if and only if it is happy with $C'_2$, unless the
  graph $G$ is non-planar, or the \peg $(G,H,\H)$ contains
  obstruction~1 (Lemma~\ref{lem-2cyc1}).
\item If $H$ contains a happy facial cycle as well as an unhappy one,
  we obtain obstruction 18 (Lemma~\ref{lem-2cyc2}).
\item If $H$ contains an unhappy facial cycle, and if at least one
  vertex of $\rskel$ imposes any non-trivial edge-compati\-bility
  constraints, then $(G,H,\H)$ contains one of the obstructions 19--22
  (Lemma~\ref{lem-edge-cyc}).
\end{itemize}

Note that these facts guarantee that if $(G,H,\H)$ is obstruction-free
then $\rskel$ has a compatible embedding.  To see this, assume that
$\rskelemb$ is an edge-compatible but not cycle-compatible embedding
of~$\rskel$. This means that at least one facial cycle of $\H$ is
unhappy.  This in turn implies that no cycle may be happy, and no
vertex of $\rskel$ may impose any edge-compatibility restrictions.
Consequently, the embedding $\rskelnot$ is compatible.  In order to
prove the above claims we need some technical machinery, in particular
the concept of conflict graph of~$C$-bridges and its properties.

\paragraph{Conflict graph of a cycle and minimality of alternating
  chains}

For a cycle $C$ and two distinct vertices $x$ and $y$ of~$C$, an
\emph{arc} of $C$ with endvertices $x$ and $y$ is a path in~$C$
connecting $x$ to $y$.  Any two distinct vertices of a cycle determine
two arcs.  Let $u,v,x,y$ be four distinct vertices of a cycle~$C$.  We
say that the pair $\{x,y\}$ \emph{alternates with} $\{u,v\}$ if each
arc determined by $x$ and $y$ contains exactly one of the two vertices
$\{u,v\}$.  If $U$ and $X$ are sets of vertices of a cycle $C$, we say
that $X$ \emph{alternates} with $U$ if there are two pairs of vertices
$\{u,v\}\subseteq U$ and $\{x,y\}\subseteq X$ that alternate with each
other.

Let now~$F$ be a graph containing a cycle~$C$.  Intuitively, a bridge
represents a subgraph, whose internal vertices and edges must all be
embedded on the same side of $C$ in any embedding of~$F$.  Thus, a
$C$-bridge may be embedded in two possible positions relative to~$C$.
Moreover, if two bridges $B_1$ and $B_2$ have three common
attachments, or if the attachments of $B_1$ alternate with the
attachments of $B_2$, then in any planar embedding, $B_1$ and $B_2$
must appear on different sides of~$C$.  This motivates the definition
of two types of conflicts between bridges.  We say that two
$C$-bridges $X$ and $Y$ of~$F$ have a \emph{three-vertex conflict} if
they share at least three common attachments, and they have a
\emph{four-vertex conflict} if $\att(X)$ alternates with~$\att(Y)$.
Two $C$-bridges have a conflict if they have a three-vertex conflict
or a four-vertex conflict.  This gives rise to a conflict graph of~$F$
with respect to~$C$.  For a cycle $C$, define the \emph{conflict
  graph} $\con C$ to be the graph whose vertices are the $C$-bridges,
and two vertices are connected by an edge of $\con C$ if and only if
the corresponding bridges conflict.  Define the \emph{reduced conflict
  graph} $\conr C$ to be the graph whose vertices are bridges of $C$,
and two bridges are connected by an edge if they have a four-vertex
conflict.

As a preparation, we first derive some basic properties of conflict graphs.

\begin{lemma}\label{lem-bip}
  If $F$ is a planar graph, then for any cycle $C$ of $F$ the conflict
  graph $\con C$ is bipartite (and hence $\conr C$ is bipartite as
  well).
\end{lemma}
\begin{proof}
  In any embedding of $F$, each $C$-bridge must be completely embedded
  on a single side of~$C$.  Two conflicting bridges cannot be embedded
  on the same side of~$C$.
\end{proof}

Consider now the situation when $C$ is a cycle of length at least~4 in
a 3-connected graph~$F$.  The goal is to show that in this case also
the reduced conflict graph $\conr C$ is connected.  To prove this we
need some auxiliary lemmas.  The first one states that if the attachments
of a set of bridges alternate with two given vertices $x$ and~$y$
of~$C$, then the set must contain a $C$-bridge whose attachments
alternate with $x$ and~$y$, provided that the set of bridges is
connected in the reduced conflict graph $\conr C$.

\begin{lemma}\label{lem-alternate}
  Let $F$ be a graph and let $C$ be a cycle in~$F$.  Let $K$ be a
  connected subgraph of the reduced conflict graph $\conr C$ and let $\att(K)$
  be the set of all attachment vertices of the $C$-bridges in $K$,
  that is, $\att(K) = \bigcup_{X\in K} \att(X)$.  If $\{x,y\}$ is a
  pair of vertices of $C$ that alternates with $\att(K)$, then there
  is a bridge $X \in K$ such that the pair $\{x,y\}$ alternates
  with~$\att(X)$.
\end{lemma}
\begin{proof}
  Let $\alpha$ and $\beta$ be the two arcs of $C$ with endvertices $x$
  and $y$.  Let $K_\alpha$ be the set of $C$-bridges from $K$ whose
  all attachments belong to $\alpha$, and let $K_\beta$ be the set of
  bridges from $K$ with all their attachments in~$\beta$. Note that
  both $K_\alpha$ and $K_\beta$ are proper subsets of $K$, because
  $\{x,y\}$ alternates with $\att(K)$.

  Since no bridge in $K_\alpha$ conflicts with any bridge in
  $K_\beta$, and since $K$ is a connected subgraph in the reduced
  conflict graph, there must exist a bridge $X\in K$ that belongs to
  $K\setminus (K_\alpha\cup K_\beta)$.  Clearly, $X$ has at least one
  attachment in the interior of $\alpha$ as well as at least one
  attachment in the interior of $\beta$.  Thus, $\att(X)$ alternates
  with~$\{x,y\}$.
\end{proof}

Next, we show that in a 3-connected graph, unless~$C$ is a triangle,
its reduced conflict graph~$\conr C$ is connected.

\begin{lemma}\label{lem-conn}
  Let $C$ be a cycle of length at least~4 in a 3-connected
  graph~$F$. Then the reduced conflict graph $\conr C$ is connected
  (and hence $\con C$ is connected as well).
\end{lemma}

\begin{proof}
  We first show that for a cycle~$C$ of length at least~4 and a set of
  $C$-bridges~$K$ that form a connected component in $\conr C$, every
  vertex of~$C$ is an attachment of at least one bridge in~$K$.

  \begin{claim}\label{claim:att}
    Let $C$ be a cycle of length at least four in a 3-connected
    graph~$F$.  Let $K$ be a connected component of the graph $\conr
    C$, and let $\att(K)$ be the set $\bigcup_{X\in K} \att(X)$.  Then
    each vertex of $C$ belongs to $\att(K)$.
  \end{claim}
  Suppose that some vertices of $C$ do not belong to $\att(K)$.  Then
  there is an arc $\alpha$ of $C$ of length at least~2, whose
  endvertices belong to $\att(K)$, but none of its internal vertices
  belongs to~$\att(K)$.  Let $x$ and $y$ be the endvertices
  of~$\alpha$.  Let $\beta$ be the other arc determined by $x$
  and~$y$. Observe that, since~$|\att(K)| \geq 3$ in any 3-connected
  graph, $\beta$ also has length at least 2.

  Since $F$ is 3-connected, $F-\{x,y\}$ is connected, and in
  particular, there is a $C$-bridge $Y$ that has at least one
  attachment $u$ in the interior of the arc $\alpha$ and at least one
  attachment $v$ in the interior of~$\beta$.  Clearly $Y\not\in K$,
  since $Y$ has an attachment in the interior of~$\alpha$.

  Since the pair $\{u,v\}$ alternates with $\{x,y\}\subseteq\att(K)$,
  Lemma~\ref{lem-alternate} shows that there is a bridge $X\in K$
  whose attachments alternate with $\{u,v\}$. Then $X$ and $Y$ have a
  four-vertex conflict, which is impossible because $K$ is a connected
  component of $\conr C$ not containing~$Y$.  This finishes the proof
  of the claim.

  We are now ready to prove the lemma.  Let $K$ and $K'$ be two
  distinct connected components of~$\conr C$.  Choose a bridge $X\in
  K$.  Let $u$ and $v$ be any two attachments of $X$ that are not
  connected by an edge of~$C$.  By Claim~\ref{claim:att}, each vertex
  of $C$ is in $\att(K')$, so $\att(K')$ alternates with $\{u,v\}$,
  and hence by Lemma~\ref{lem-alternate}, the set $K'$ has a bridge
  $Y$ whose attachments alternate with the attachments of~$X$.  Hence,
  $X$ and~$Y$ have a four-vertex conflict and belong to the same
  connected component in $\conr C$.
\end{proof}

Next, we show that if we have an induced path in the conflict graph,
then we can find a corresponding sequence of bridges and pairs of
their attachment vertices such that consecutive pairs alternate.  This
lemma will be the main tool for extracting alternating chains
from non-planar \pegs.

\begin{lemma}\label{lem-path}
  Let $C$ be a cycle of length at least~4 in a graph~$F$
  and let $P$ be an induced path with $k\ge 2$ vertices in the graph
  $\conr C$.  Let $X_1, X_2,\dotsc, X_k$ be the vertices of $P$, with
  $X_i$ adjacent to $X_{i+1}$ for each $i=1,\dots,k-1$.  Then for each
  $i\in\{1,\dotsc,k\}$ we may choose a pair of vertices
  $\{x_i,y_i\}\subseteq \att(X_i)$, such that for each $i=1,\dots,k-1$
  the pair $\{x_i,y_i\}$ alternates with the pair
  $\{x_{i+1},y_{i+1}\}$.
\end{lemma}
\begin{proof}
  For each $j \le k$, select a set $S_j\subseteq\att(X_j)$ in such a
  way that for each $i<k$ the set $S_i$ alternates with
  $S_{i+1}$.  Such a selection is possible, e.g., by taking
  $S_j=\att(X_j)$.  Assume now that we have selected $\{S_j \mid
  j=1,\dotsc,k\}$ so that their total size $\sum_{j\le k} |S_j|$ is as
  small as possible.  We claim that each set $S_j$ consists of a pair
  of vertices $\{x_j,y_j\}$.

  Assume for contradiction that this is not the case.  Since obviously
  each $S_j$ has at least two vertices, assume that for some $j$ we
  have $|S_j|\ge 3$.  Clearly, this is only possible for $1<j<k$.
  Select a pair of vertices $\{x_{j-1},y_{j-1}\}\subseteq S_{j-1}$ and
  a pair of vertices $\{x_{j+1},y_{j+1}\}\subseteq S_{j+1}$ such that
  both these pairs alternate with~$S_j$. The sets $S_{j-1}$ and
  $S_{j+1}$ do not alternate because $P$ was an induced path. Therefore, there is an arc
  $\alpha$ of $C$ with endvertices $\{x_{j-1},y_{j-1}\}$ that has no
  vertex from $S_{j+1}$ in its interior, and similarly there is an arc
  $\beta$ with endvertices $\{x_{j+1},y_{j+1}\}$ and no vertex of
  $S_{j-1}$ in its interior.

  Since both $\{x_{j-1},y_{j-1}\}$ and $\{x_{j+1},y_{j+1}\}$ alternate
  with $S_j$, there must be a vertex $x_j\in S_j$ that belongs to the
  interior of $\alpha$, and a vertex $y_j\in S_j$ belonging to the
  interior of~$\beta$.  The pair $\{x_j,y_j\}$ alternates with both
  $S_{j-1}$ and $S_{j+1}$, contradicting the minimality of our choice
  of~$S_j$.
\end{proof}

Our next goal is to link the conflict graph with the elements
of~$\ach{k}$.  Recall that an element of $\ach{k}$ consists of an
$H$-cycle of length~$k+1$ and $k$ edge-disjoint paths $P_1,\dots,P_k$
such that consecutive pairs have alternating endpoints on~$C$.
Moreover,~$P_2,\dots, P_{k-1}$ are single edges, while~$P_1$ and~$P_k$
are subdivided by a single isolated $H$-vertex.  Note that for all
elements $(G_k,H_k,\H_k)$ of $\ach{k}$, the conflict graph of the
unique $H_k$-cycle forms a path of length~$k$.  To establish a link,
we consider pairs of a graph and a cycle such that the conflict graph
forms a path.  Let $F$ be a graph, and let $C$ be a cycle in~$F$.  We
say that the pair $(F,C)$ forms a \emph{conflict path}, if each
$C$-bridge of $F$ has exactly two attachments and the conflict graph
$\con C$ is a path. (Note that if each $C$-bridge has two attachments,
then the conflict graph is equal to the reduced conflict graph.)

Note that if $(G_k,H_k,\H_k)$ is an element of $\ach{k}$ and $C$ the unique cycle of $H_k$, then $(G_k,C)$ forms a conflict path.  However, not every conflict path arises this way.  Suppose that $(F,C)$ forms a conflict
path.  Let $e=uv$ be an edge of~$C$.  The edge $e$ is called
\emph{shrinkable} if no $C$-bridge attached to $u$ conflicts with any
$C$-bridge attached to~$v$. Note that a shrinkable edge may be
contracted without modifying the conflict graph.

Before we can show that the elements of $\ach k$ are minimal
non-planar \pegs, we first need a more technical lemma about conflict
paths.

\begin{lemma}\label{lem-att3}
  Assume that $(F,C)$ forms a conflict path.  Then each vertex of $C$
  is an attachment for at most two $C$-bridges.
\end{lemma}
\begin{proof}
  Suppose that $(F,C)$ forms a conflict path, and a vertex $v\in C$ is
  an attachment of three distinct bridges $X$, $Y$ and $Z$. These
  three bridges do not alternate, so there must be at least five
  bridges to form a path in~$\con C$. Let $x$, $y$ and $z$ be the
  attachments of $X$, $Y$ and $Z$ different from~$v$. The three
  vertices $x$, $y$ and $z$ must be all distinct, because a pair of
  bridges with the same attachments would share the same neighbors in
  the conflict graph, which is impossible if the conflict graph is a
  path with at least five vertices.

  Choose an orientation of $C$ and assume that the four attachments
  appear in the order $(v,x,y,z)$ with respect to this
  orientation. Let $\alpha_{vx}$, $\alpha_{xy}$, $\alpha_{yz}$, and
  $\alpha_{zx}$ be the four internally disjoint arcs of $C$ determined
  by consecutive pairs of these four attachments.
  
  For a subgraph $P'$ of $P$, let $\att(P')$ denote the set of all the
  attachments of the bridges that  belong to $P'$. Let $P_{xz}$ be the subpath
  of $\con C$ that connects $X$ to $Z$. At least one vertex of $\att(P_{xz})$
  must belong to the interior of $\alpha_{vx}$ and at least one vertex of
  $\att(P_{xz})$ must belong in the interior of $\alpha_{zv}$. Hence the set
  $\att(Y)$ alternates with $\att(P_{xz})$ and by Lemma~\ref{lem-alternate}, at
  least one bridge in $P_{xz}$ conflicts with $Y$. This means that $Y$ is an
  internal vertex of~$P_{xz}$.

  Consider now the graph $P_{xz}-Y$. It consists of two disjoint paths $P_x$ and
  $P_z$ containing $X$ and $Z$ respectively. We know that $P_x$ has a  vertex
  adjacent to $X$ as well as a vertex adjacent to $Y$, but no vertex adjacent to
  $Z$. Consequently, $\att(P_x)$ contains at least one vertex from the interior
  of $\alpha_{vx}$ as well as at least one vertex from  the interior of
  $\alpha_{yz}$.  Similarly, $\att(P_z)$ has  a vertex from the interior of
  $\alpha_{xy}$ and from the interior of $\alpha_{zv}$. Hence, the set
  $\att(P_x)$ alternates with $\att(P_z)$. Using Lemma~\ref{lem-alternate}, we
  easily deduce that at least one bridge of $P_x$ must conflict with a bridge
  of $P_z$, which is a contradiction.
\end{proof}

Next, we show that the attachment vertices on the
cycle~$C$ of a conflict path $(F,C)$ without shrinkable edges have a
structure very similar to that of an alternating chain.

\begin{lemma}\label{lem-contract}
  Assume that $(F,C)$ forms a conflict path with $k\ge 4$ $C$-bridges.
  Let $X_1,\dotsc,X_k$ be the $C$-bridges, listed in the order in
  which they appear on the path~$\con C$.  Let $\{x_i,y_i\}$ be the
  two attachments of $X_i$.  Assume that $C$ has no shrinkable
  edge. Then
  \begin{enumerate}
    \itemsep=-.5ex
  \item The two attachments $\{x_1,y_1\}$ of $X_1$ determine an arc
    $\alpha_1$ of length 2, and the unique internal vertex $z_1$ of
    this arc is an attachment of $X_2$ and no other bridge.
  \item The two attachments $\{x_k,y_k\}$ of $X_k$ determine an arc
    $\alpha_k$ of length 2 different form $\alpha_1$, and the unique
    internal vertex $z_k$ of this arc is an attachment of $X_{k-1}$
    and no other bridge.
  \item All the vertices of $C$ other than $z_1$ and $z_k$ are
    attachments of exactly two bridges.
  \end{enumerate}
\end{lemma}
\begin{proof}
  We know from Lemma~\ref{lem-att3} that no vertex of $C$ is an
  attachment of more than two bridges.

  Let $\alpha$ and $\beta$ be the two arcs of $C$ determined by
  $\{x_1,y_1\}$.  The bridges $X_3,\dotsc,X_k$ do not alternate with
  $X_1$, so all their attachments belong to one of the two arcs, say
  $\beta$. The arc $\alpha$ then has only one attachment $z_1$ in its
  interior, and this attachment belongs to $X_2$ and no other
  bridge. It follows that $\alpha$ has only one internal vertex.  This
  proves the first claim; the second claim follows analogously.

  To prove the third claim, note first that any vertex of $C$ must be
  an attachment of at least one bridge. Suppose that there is a vertex
  $v$ that is an attachment of only one bridge~$X_j$. Let $u$ and $w$
  be the neighbors of $v$ on~$C$. By assumption, both $u$ and $w$ are
  attachments of at least one bridge that conflicts with $X_j$.

  Assume first, that a single bridge $Y$ conflicting with $X_j$ is
  attached to both $u$ and~$w$. Since the arc determined by $u$ and
  $w$ and containing $v$ does not have any other attachment in the
  interior, this means that $Y$ conflicts only with the bridge
  $X_j$. Then $Y\in\{X_1,X_k\}$ and $v\in\{z_1,z_k\}$. Next, assume
  that the bridge $X_{j-1}$ is attached to $u$ but not to $w$, and the
  bridge $X_{j+1}$ is attached to $w$ but not~$u$.  We then easily
  conclude that $X_{j-1}$ conflicts with $X_{j+1}$, which is a
  contradiction.
\end{proof}

\newtheorem{corollary}{Corollary}

This directly implies that non-planar \pegs that form a conflict path
and do not have shrinkable edges are $k$-fold alternating chains.

\begin{corollary}\label{cor:ach}
  Let~$(G,H,\H)$ be a non-planar \peg for which~$H$ consists of a
  single cycle~$C$ of length at least~4 and two additional
  vertices~$u$ and~$v$ that do not belong to~$C$, such that~$(G,C)$
  forms a conflict path with bridges~$X_1,\dots,X_k$ along the path,
  each with attachments~$\{x_i,y_i\}$.  Let further~$X_i$ consist of
  the single edge~$x_iy_i$ for~$i=2,\dots,k-1$ and let~$X_1$ consist
  of~$x_1uy_1$ and~$X_k$ of~$x_kvy_k$.  If~$C$ does not contain
  shrinkable edges then~$(G,H,\H)$ is an element of~$\ach{k}$.
\end{corollary}

\begin{proof}
  The non-planarity of~$G$ implies that~$u$ and~$v$ must be embedded
  on different sides of~$C$ if~$k$ is even, and on the same side
  if~$k$ is odd.

  Clearly, the graphs~$G$ and~$H$ have the same vertex set.  By
  assumption, each bridge~$X_i$ forms a path~$P_i$, which satisfy the
  properties for $k$-fold alternating chains; they have the right
  lengths and contain the right vertices.  Further, since~$(G,C)$
  forms a conflict path their endpoints alternate in the required way.

  Finally, as~$C$ has no shrinkable edges, Lemma~\ref{lem-contract}
  implies that all vertices of~$C$ have degree~4, with the exception
  of one of the attachments of~$X_2$ and~$X_{k-1}$, which have
  degree~3.  This also implies that the length of the cycle is~$k+1$,
  and thus~$(G,H,\H)$ thus is an element of~$\ach{k}$.
\end{proof}

We now employ the observations we made so far to show that every
element of~$\ach k$ is indeed an obstruction.

\begin{lemma}\label{lem-ach}
  For each $k\ge 3$, every element of $\ach k$ is an
  obstruction.
\end{lemma}
\begin{proof}
  As observed before, $\ach 3$ contains a single element, which is the
  obstruction~4.  Assume $k\ge 4$, and choose $(G',H',\H')\in\ach
  k$.  Let $C$ be the unique cycle of $H$, and let $u$ and $v$ be the
  two isolated vertices of~$H$.  Observing that $(G',H',\H')$ is not
  planar is quite straightforward: since no two conflicting bridges
  can be embedded into the same region of $C$, all the odd bridges
  $X_1, X_3, X_5,\dotsc $ must be in one region while all the even
  bridges must be in the other region, and this guarantees that $u$ or
  $v$ will be on the wrong side of~$C$.

  Let us prove that $(G',H',\H')$ is minimal non-planar.  The least
  obvious part is to show that contracting an edge of a cycle $C$
  always gives a planar \peg.  If the cycle $C$ contained a shrinkable
  edge $e=xy$, we might contract the edge into a single vertex
  $x_e$.  After the contraction, the new graph still forms a conflict
  path, but the vertex $x_e$ is an attachment of at least three
  bridges, which contradicts Lemma~\ref{lem-att3}.  We conclude that
  $C$ has no shrinkable edge.

  By contracting a non-shrinkable edge $C$, we obtain a new \peg
  $(G'',H'',\H'')$ where $H$ consists of a cycle $C'$ and two isolated
  vertices. The conflict graph of $C'$ in $G''$ is a proper subgraph
  of the conflict graph of $C$ in~$G'$.  In particular, the bridges
  containing $u$ and $v$ belong to different components of the
  conflict graph of~$C'$.  We may then assign each bridge to one of the
  two regions of the cycle $C'$, in such a way that the bridges
  containing $u$ and $v$ are assigned consistently with the embedding
  $\H''$, and the remaining bridges are assigned in such a way that no
  two bridges in the same region conflict.

  It is easy to see that any collection of $C'$-bridges that does not
  have a conflict can be embedded inside a single region of $C'$
  without crossing. Thus, $(G'', H'',\H'')$ is planar.

  By analogous arguments, we see that removing or relaxing an edge or
  vertex of $H'$ yields a planar \peg.  Contracting an edge incident
  to $u$ or $v$ yields an planar \peg as well.  Thus, $(G',H',\H')$ is
  an obstruction.
\end{proof}

\paragraph{At least one of the embeddings is edge-compatible}

Finally, we use all this preparation to analyze the skeletons of
R-nodes.  In all the following lemmas we suppose that $(G,H,\H)$ is a
2-connected obstruction-free \peg, and that $\rskel$ is an R-skeleton
of $G$ with at least one edge-compatible embedding~$\rskelemb$, which
we assume to be fixed.  We denote this hypothesis (HP1).

Let $C$ be a cycle of $\rskel$ that is a projection of a cycle $C'$
of~$\H$.  Recall that a vertex~$x$ of~$H$ that does not belong to an
edge of~$C$ is happy with $C'$ if it is embedded on the correct side
of~$C$ in $\rskelemb$, and that it is unhappy otherwise.  Recall
further that a $C$-bridge is happy with $C'$, if it contains a happy
vertex, and it is unhappy if it contains an unhappy vertex and that a
bridge that is neither happy nor unhappy is indifferent.  We first
show that a $C$-bridge cannot be happy and unhappy at the same time.

\begin{lemma}\label{lem-schizo}
  In the hypothesis (HP1), if $C$ is a cycle of $\rskel$ that is a
  projection of a cycle $C'$ of~$\H$, then no $C$-bridge can be both
  happy and unhappy with~$C'$.
\end{lemma}
\begin{proof}
  Assume a $C$-bridge $X$ contains a happy vertex $u$ and an unhappy
  vertex~$v$.

  If there exists a $G$-path from $u$ to $v$ that avoids all the
  vertices of $C'$, then we obtain obstruction~1. Assume then that
  there is no such path.  This easily implies that the bridge $X$ is a
  single $\rskel$-edge $B$ with two attachments $x$ and~$y$.
  Figure~\ref{fig:schizo} shows this situation and illustrates the
  following steps.  Since both $u$ and $v$ are connected to $x$ and to
  $y$ by a $G$-path projecting into $B$, there is a cycle $D$ of $G$
  containing both $u$ and~$v$, and which is contained in~$B$.  Since
  every $G$-path from $u$ to $v$ inside $B$ intersects $x$ or $y$, we
  conclude that $D$ can be expressed as a union of two $G$-paths $P$
  and $Q$ from $x$ to $y$, with $u\in P$ and $v\in Q$.

  \begin{figure}[tb]
    \centering
    \includegraphics{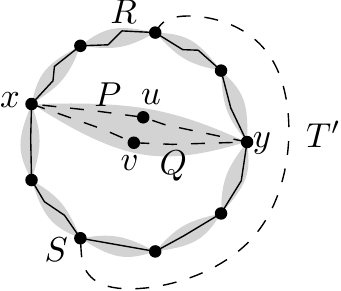}
    \caption{Illustration of Lemma~\ref{lem-schizo}, the bridge
      embedded in the cycle contains a happy vertex~$u$ and an unhappy
      vertex~$v$ that are not connected by a path avoiding~$x$
      and~$y$.  In this case, the \peg contains obstruction~4.}
    \label{fig:schizo}
  \end{figure}

  Similarly, the cycle $C$ of $\rskel$ can be expressed as a union of
  two $\rskel$-paths $R$ and $S$, each with at least one internal
  vertex. The paths $R$ and $S$ are projections of two $H$-paths $R'$
  and~$S'$.  Since $\rskel$ is 3-connected, it has a path $T$ that
  connects an internal vertex of $R$ to an internal vertex of $S$, and
  whose internal vertices avoid~$C$.  The path $T$ is a projection of
  a $G$-path~$T'$.  The paths $P$, $Q$, $R'$, $S'$ and $T'$ can be
  contracted to form obstruction~4.
\end{proof}

Recall that a cycle~$C'$ in $H$ that projects to a cycle $C$  in $\rskel$  is happy, if there is at least one $C$-bridge that is
happy with $C'$ and it is unhappy, if at least one $C$-bridge is
unhappy with~$C'$.  Again, as the following lemma shows, cycles cannot
be both happy and unhappy at the same time.

\begin{lemma}\label{lem-schizo-cyc}
  In the hypothesis (HP1), if $C'$ is a cycle of $H$ whose projection
  is a cycle $C$ of $\rskel$, then $C'$ cannot be both happy and
  unhappy.
\end{lemma}
\begin{proof}
  Suppose $C$ has a happy bridge $X$ containing a happy vertex $u$,
  and an unhappy bridge $Y$ with an unhappy vertex $v$.

  If $C$ is a triangle, then $X$ and~$Y$ cannot be chords of~$C$ and
  therefore they have three attachments, each.  This implies that they
  are embedded on different sides of the triangle and all vertices of
  the triangle are attachments of both~$X$ and~$Y$.  Since $Y$ is
  unhappy, it contains a vertex that is prescribed on the same side of
  $C'$ as~$X$.  This yields obstruction~17.
  Otherwise, $C$ has length at least~4, and we know that the reduced
  conflict graph $\conr C$ is connected by Lemma~\ref{lem-conn}.  We
  find a shortest path $X_1,\dotsc, X_k$ in $\conr C$ connecting
  $X=X_1$ to $Y=X_k$.  If the path is a single edge, we obtain
  obstruction~16.  Otherwise we use Lemma~\ref{lem-path} to choose for
  each $X_i$ a pair of attachments $\{x_i,y_i\}\subseteq\att(X_i)$,
  such that $\{x_i,y_i\}$ alternates with $\{x_{i+1},y_{i+1}\}$.

  Since each $C$-bridge of the skeleton represents a connected
  subgraph of $G$, we know that for every $i=2,\dotsc,k-1$ the graph
  $G$ has a path from $x_i$ to $y_i$ whose internal vertices avoid
  $C'$ and which projects to the interior of~$X_i$.  We also know that
  there is a $G$-path $Q_1$ from $x_1$ to $u$, and a $G$-path $R_1$
  from $y_1$ to $u$ whose internal vertices avoid $C'$ and which
  project into~$X_1$.  Similarly, there are $G$-paths $Q_k$ and $R_k$
  from $x_k$ to $v$ and from $y_k$ to $v$, internally disjoint with
  $C'$ and projecting into~$X_k$.  Performing contractions if necessary,
  we may assume that all these paths are in fact single edges.

  Consider the sub-\peg $(G',H',\H')$, where $H'$ consists of the
  cycle $C'$ and the two vertices $u$ and $v$, and $G$ has in addition
  all the edges obtained by contracting the paths defined above. If
  $C'$ has shrinkable edges, we may contract them, until no shrinkable
  edges are left.  Then we either obtain obstruction~4 (if $k=3$), or
  Corollary~\ref{cor:ach} implies that we have obtained an occurrence of
  $\ach k$ for some~$k\ge 4$.
\end{proof}

Next, we show that it is not possible that one cycle is
happy and another one is unhappy.  However, this is complicated if the
cycles are too close in~$\rskel$, in particular if they share
vertices.  Therefore, we first show that an unhappy cycle $C'$
projecting to a cycle~$C$ may not have an incident $H$-edge that does
not belong to~$C$.  Such an edge~$e$, if it existed, would either be a
chord of~$C'$, or it would be part of a bridge containing a vertex
of~$H$ (e.g., the endpoint of~$e$ not belonging to~$C'$).  The next
two lemmas exclude these two cases separately.

In the former case, where $e$ is a chord of~$C'$ that hence projects
to a chord of~$C$, we also call $e$ a \emph{relevant chord}.  Note
that if $B$ is an edge of $\rskel$ containing a relevant chord, then
in an edge-compatible embedding of $\rskel$, $B$ must always be
embedded on the correct side of~$C$.  For practical purposes, such an
edge $B$ behaves as a happy bridge, as shown by the next lemma.

\begin{lemma}\label{lem-chord}
  In the hypothesis (HP1), let $C'$ be a cycle of $H$ that projects to
  a cycle $C$ of $\rskel$.  Let $e$ be a relevant chord of $C'$ that
  projects into an $\rskel$-edge $B$.  Then $C'$ cannot be unhappy.
\end{lemma}
\begin{proof}
  Let $u$ and $v$ be the two vertices of $e$, which are also the two
  poles of~$B$.  Let $\alpha'$ and $\beta'$ be the two arcs of $C'$
  determined by the two vertices $u$ and~$v$, and let $\alpha$ and
  $\beta$ be the two arcs of $C$ that are projections of $\alpha'$ and
  $\beta'$, respectively.  Note that each of the two arcs $\alpha$ and
  $\beta$ has at least one internal vertex, otherwise $B$ would not be a chord.

  Suppose for contradiction that $C$ has an unhappy bridge $X$
  containing an unhappy vertex~$x$.  We distinguish two cases,
  depending on whether $B$ is part of $X$ or not.

  \begin{figure}[tb]
    \centering
    \subfigure[]{\includegraphics[page=1]{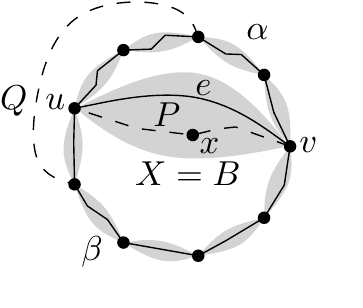}\label{fig:relchord-a}}
    \subfigure[]{\includegraphics[page=2]{relevant-chord}\label{fig:relchord-b}}
    \caption{Illustration of Lemma~\ref{lem-chord}.  The edge~$e$ is a
      relevant chord of the shown cycle, and the bridge~$X$ is assumed
      to be unhappy.  If the skeleton edge of~$e$ also contains an
      unhappy vertex~$x$, we obtain
      obstruction~3~\subref{fig:relchord-a}.  Otherwise, the skeleton
      edge of~$e$ together with the arc between its attachments forms
      a smaller cycle~$C_\alpha$, for which~$X$ is still unhappy, but
      the remainder of the cycle is part of a happy
      bridge~\subref{fig:relchord-b}, which contradicts the fact that
      a cycle cannot be both happy and unhappy.}
    \label{fig:relevant-chord}
  \end{figure}

  First, assume that the bridge $X$ contains the $\rskel$-edge $B$.
  Then $X$ is a trivial bridge whose only edge is~$B$; see
  Fig.~\ref{fig:relchord-a}.  The edge $B$ then contains a
  $G$-path~$P$ from $u$ to $v$ containing $x$.  The graph $G$ also
  has a path $Q$ connecting an internal vertex of $\alpha$ to an
  internal vertex of~$\beta$ and avoiding both $u$ and~$v$.  Together,
  the edge $e$, the paths $P$ and $Q$, and the arcs $\alpha$ and
  $\beta$ can be contracted to form obstruction~3.

  Assume now that the bridge $X$ does not contain $B$.  Consider two
  $H$-cycles $C'_\alpha=\alpha'\cup e$ and $C'_\beta=\beta'\cup e$,
  and their respective projections $C_\alpha=\alpha\cup B$ and
  $C_\beta=\beta\cup B$.  It is not hard to see that the vertex $x$
  must be unhappy with at least one of the two cycles $C'_\alpha$
  and~$C'_\beta$. Let us say that $X$ is unhappy with~$C'_\alpha$; see
  Fig.~\ref{fig:relchord-b}.  Thus, $C_\alpha$ has at least one
  unhappy bridge.  We claim that $C_\alpha$ also has a happy bridge.
  Indeed, let $Y$ be the bridge of $C_\alpha$ that contains~$\beta$.
  Since $\beta$ has at least one internal vertex, the bridge $Y$ is
  not indifferent.  The bridge $Y$ must be happy, otherwise the
  vertices $u$ and $v$ would violate edge-compatibility.  This means
  that $C_\alpha$ has both a happy bridge and an unhappy bridge,
  contradicting Lemma~\ref{lem-schizo-cyc}.
\end{proof}

\begin{lemma}\label{lem-cyc-fork}
  In the hypothesis (HP1), let $C'$ be a cycle of $H$ that projects to
  a cycle $C$ of~$\rskel$.  If $C'$ is unhappy, then every edge of $H$
  that is incident to a vertex of $C$ projects into an $\rskel$-edge
  that belongs to~$C$.
\end{lemma}
\begin{proof}
  For contradiction, assume that an edge $e=uv$ of $H$ is incident to
  a vertex $u\in C$, but projects into an $\rskel$-edge $B\not\in
  C$. If $v$ is also a vertex of $C$, then $e$ is a relevant chord and
  $C$ may not have any unhappy bridges by Lemma~\ref{lem-chord}. If
  $v\not\in C$, then $v$ is an internal vertex of a $C$-bridge, and
  from edge-compatibility it follows that $v$ is happy with $C'$. Thus
  $C$ has both happy and unhappy bridges, contradicting
  Lemma~\ref{lem-schizo-cyc}.
\end{proof}

The previous two lemmas show that for an unhappy cycle~$C'$ of $H$
projecting to a cycle~$C$ of~$\rskelemb$, no $C$-bridge contains an
$H$-edge incident to a vertex of~$C$.  In particular, the projection
of any happy $H$-cycle is either disjoint from~$C$ (that is they are
far apart) or it is identical to~$C$.  We now exclude the latter case.

\begin{lemma}\label{lem-2cyc1}
  In the hypothesis (HP1), let $C_1'$ and $C'_2$ be two distinct
  facial cycles of $\H$, which project to the same (undirected) cycle
  $C$ of~$\rskel$. Then any $C$-bridge that is happy with $C'_2$ is
  also happy with~$C'_1$.
\end{lemma}
\begin{proof}
  Let $F_1$ and $F_2$ be the faces of $\H$ corresponding to facial
  cycles $C_1'$ and $C_2'$, respectively.

  Suppose for contradiction that at least one $C$-bridge $X$ is unhappy
  with $C'_1$ and happy with $C'_2$. In view of
  Lemma~\ref{lem-schizo}, we may assume that $X$ contains in its
  interior a vertex $x\in H$, such that $x$ is unhappy with $C'_1$ and
  happy with~$C'_2$. Refer to Figure~\ref{fig-2cyc1}.

  Suppose that the two facial cycles $C'_1$ and $C'_2$ are oriented in
  such a way that their corresponding faces are to the left of the
  cycles.  Note that any vertex of $C$ is a common vertex of $C'_1$
  and $C'_2$.  This shows that the two facial cycles have at least
  three common vertices, which implies that they correspond to
  different faces of~$\H$.

  Let $a$, $b$ and $c$ be any three distinct vertices of $C$, and
  assume that these three vertices appear in the cyclic order
  $(a,b,c)$ when the cycle $C'_1$ is traversed according to its
  orientation.  The interior of the face $F_2$ lies to the right of
  the cycle $C'_1$, and in particular, the three vertices $a,b,c$
  appear in the cyclic order $(c,b,a)$ when the boundary of $F_2$ is
  traversed in the orientation of $C'_2$.  Thus, $C_1'$ and $C'_2$
  induce opposite orientations of their common projection~$C$.  Since
  $x$ is happy with exactly one of the two cycles $C'_1$ and $C'_2$,
  it means that in the graph $H$ with embedding $\H$, the two cycles
  either both have $x$ on their right, or both have $x$ on their left.
  It is impossible that both facial cycles have $x$ on their left,
  because the region left of $C'_1$ is disjoint from the region left of $C'_2$.
  Hence $x$ is to the right of $C_1'$ and $C_2'$.

  Let $H_C$ be the connected component of $H$ containing the vertices
  of $C$, and let $\H_C$ be its embedding inherited from~$\H$.  By
  Lemma~\ref{lem-cyc-fork}, the bridge $X$ contains no edge of $H$
  adjacent to $C$, so $x\not\in H_C$.  Let $F_3$ be the face of $\H_C$
  that contains $x$ in its interior.  Note that $F_3$ is distinct
  from~$F_1$ and~$F_2$, as~$x$ is contained in it, which is not the
  case for~$F_1$ and~$F_2$.  All the attachments of the bridge $X$
  must belong to the boundary of $F_3$ (as well as $F_1$ and $F_2$),
  otherwise we would obtain obstruction~1, using the fact that $X$
  contains a $G$-path from $x$ to any attachment of~$X$.  If $X$ has
  at least three attachments, this leads to contradiction, because no
  three faces of a planar graph can share three common boundary
  vertices --- to see this, imagine inserting a new vertex into the
  interior of each of the three faces and connecting the new vertices
  by edges to the three common boundary vertices, to obtain a planar
  drawing of~$K_{3,3}$.

\begin{figure}
 \centering\includegraphics{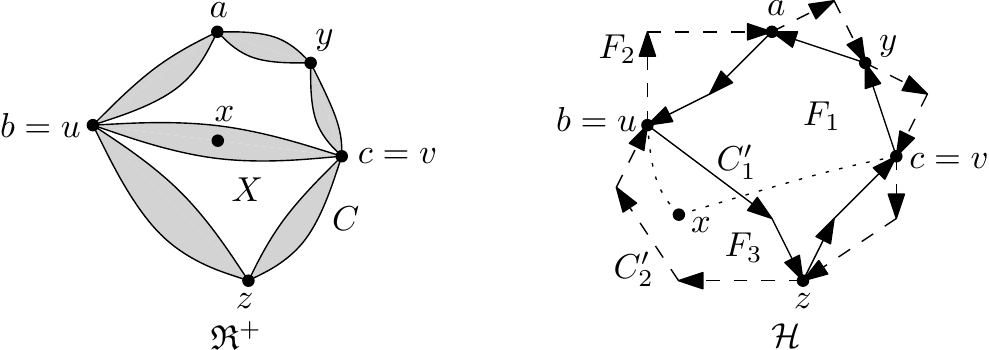}
\caption{Illustration of the proof of Lemma~\ref{lem-2cyc1}. In the right part,
the solid lines correspond to $C'_1$, the dashed lines represent $C'_2$, and
the dotted lines represent the paths from $x$ to the two attachments
of~$X$.}\label{fig-2cyc1}
\end{figure}

  Suppose now that $X$ only has two attachments $u$ and $v$, which
  means that $X$ is a trivial bridge.  Each of the two arcs of $C$
  determined by $u$ and $v$ must have an internal vertex.  Let $y$ and
  $z$ be such internal vertices of the two arcs.  To get a
  contradiction, insert a new vertex $w$ into the interior of face
  $F_1$ in $\H$ and connect it by edges to all the four vertices
  $u,v,y,z$.  Then draw an edge $uv$ inside face $F_3$ and an edge
  $yz$ inside~$F_2$.  The new edges together with the cycle $C'_1$
  form a subdivision of~$K_5$.
\end{proof}

We are now ready to show that $\rskelemb$ may not have a happy and an
unhappy cycle.

\begin{lemma}\label{lem-2cyc2}
  In the hypothesis (HP1), let $C'_1$ and $C'_2$ be two cycles of $H$
  that project to two distinct cycles $C_1$ and $C_2$ of~$\rskel$. If
  $C'_1$ is unhappy, then $C'_2$ cannot be happy.
\end{lemma}
\begin{proof}
  Suppose that $C'_1$ is unhappy and $C'_2$ is happy.  By
  Lemma~\ref{lem-cyc-fork}, this means that no $C_1$-bridge may
  contain an edge of $H$ incident to a vertex of~$C_1$.  Consequently,
  the two cycles $C_1$ and $C_2$ are vertex-disjoint.  Since $\rskel$
  is 3-connected, it contains three disjoint paths $P_1$, $P_2$ and
  $P_3$, each connecting a vertex of $C_1$ to a vertex of~$C_2$.  Each
  path $P_i$ is a projection of a $G$-path $P'_i$ connecting a vertex
  of $C'_1$ to a vertex of~$C'_2$.  Note that $C_1$ is inside a happy
  bridge of $C_2$, and $C_2$ is inside an unhappy bridge of~$C_1$.
  Thus, contracting the cycles $C'_1$ and $C'_2$ to triangles and
  contracting the paths $P'_i$ to edges, we obtain obstruction~18.
\end{proof}

The next lemma shows that if any vertex $u$ of~$\rskel$ that requires
the embedding~$\rskelemb$, then no cycle can be unhappy.

\begin{lemma}\label{lem-edge-cyc}
  In the hypothesis (HP1), assume that $H$ has three edges $e_1$,
  $e_2$ and $e_3$ that are incident to a common vertex $u$ and project
  into three distinct $\rskel$-edges $B_1$, $B_2$ and $B_3$
  of~$\rskel$.  Then no cycle of $H$ that projects to a cycle
  of~$\rskel$ can be unhappy.
\end{lemma}
\begin{proof}
  Proceed by contradiction. Assume that there is an unhappy cycle $C'$
  of $\H$, which projects to a cycle $C$ of~$\rskel$. From
  Lemma~\ref{lem-cyc-fork} it then follows that $u$ does not belong
  to~$C$, and hence $u$ must belong to an unhappy $C$-bridge. From the
  same lemma we also conclude that the vertex $u$ and the three edges
  $e_i$ belong to a different component of $H$ than the cycle~$C'$.

  For $i\in\{1,2,3\}$, suppose that the $H$-edge $e_i$ connects vertex
  $u$ to a vertex $v_i$, and is contained in an $\rskel$-edge $B_i$
  that connects vertex $u$ to a vertex~$w_i$. These vertices,
  $H$-edges and $\rskel$-edges are distinct, except for the
  possibility that $v_i=w_i$.

  Let $D$ be the horizon of $u$ in $\rskelemb$. The
  three vertices $w_1$, $w_2$ and $w_3$ split $D$ into three
  internally disjoint arcs $\alpha_{12}$, $\alpha_{13}$ and
  $\alpha_{23}$, where $\alpha_{ij}$ has endvertices $w_i$ and~$w_j$.

  As $\rskel$ is 3-connected, it contains three disjoint paths $P_1$,
  $P_2$ and $P_3$, where $P_i$ connects $w_i$ to a vertex of~$C$.  We
  now distinguish two cases, depending on whether the paths $P_i$ can
  avoid $u$ or not.

  First, assume that it is possible to choose the paths $P_i$ in such
  a way that all of them avoid the vertex~$u$.  We may then add $B_i$
  to the path $P_i$ to obtain three paths from $u$ to $C$, which only
  share the vertex~$u$. It follows that the graph $G$ contains three
  paths $R'_1$, $R'_2$ and $R'_3$ from $u$ to $C'$ which are disjoint
  except for sharing the vertex~$u$, and moreover, each $R'_i$
  contains the edge~$e_i$. This yields obstruction~19.

  Next, assume that it is not possible to choose $P_1$, $P_2$ and
  $P_3$ in such a way that all the three paths avoid~$u$.

  For $i\in\{1,2,3\}$, let $x_i$ be the last vertex of $P_i$ that
  belongs to $D$, assuming the path $P_i$ is traversed from $w_i$
  towards~$C$. Let $Q_i$ be the subpath of $P_i$ starting in $x_i$ and
  ending in a vertex of~$C$ (so $Q_i$ is obtained from $P_i$ by
  removing vertices preceding~$x_i$). Let $y_1$, $y_2$ and $y_3$ be
  the endvertices of $P_1$, $P_2$ and $P_3$ that belong to~$C$. We may
  assume that $y_i$ is the only vertex of $P_i$ belonging to~$C$,
  otherwise we could replace $P_i$ with its proper subpath.

  We claim that one of the three arcs $\alpha_{12}$, $\alpha_{13}$,
  and $\alpha_{13}$ must contain all the three vertices $x_i$,
  possibly as endvertices. If the vertices $x_i$ did not belong to the
  same arc, we could connect each $x_i$ to a unique vertex $w_j$ by
  using the edges of $D$, and we would obtain three disjoint paths
  from $w_i$ to $C$ that avoid~$u$. Assume then, without loss of
  generality, that $\alpha_{12}$ contains all the three
  vertices~$x_i$.

  We may also see that if the cycles $C$ and $D$ share a common vertex
  $y$, then $y$ belongs to $\alpha_{12}$. If not, we could connect
  $w_3$ to $y$ by an arc of $D$ that avoids $w_1$ and $w_2$, and we
  could connect $w_1$ and $w_2$ to two distinct vertices $x_i$ and
  $x_j$ by disjoint arcs of $D$, thus obtaining three disjoint paths
  from $w_i$ to $C$ avoiding~$u$.

  Fix distinct indices $p,q,r\in\{1,2,3\}$ so that the three vertices
  $x_1$, $x_2$ and $x_3$ are encountered in the order $x_p,x_q,x_r$
  when $\alpha_{12}$ is traversed in the direction from $w_1$ to
  $w_2$. Let $\beta$ be the arc of $D$ contained in $\alpha_{12}$
  whose endpoints are $x_p$ and~$x_r$. Clearly $x_q$ is an internal
  vertex of~$\beta$.

  We claim that at least one internal vertex of $\beta$ is connected
  to $u$ by an edge of $\rskel$. Assume that this is not the
  case. Then we may insert into the embedding $\rskelemb$ a new edge
  $f$ connecting $x_p$ and $x_r$ and embedded inside the face of
  $\rskelemb$ shared by $x_q$ and~$u$. Let $\gamma$ be the arc of $C$
  with endvertices $y_p$ and $y_r$ that does not contain~$y_q$. The
  arc $\gamma$, the paths $Q_p$ and $Q_r$ and the edge $f$ together
  form a cycle in the (multi)graph $\rskel\cup \{f\}$. The vertex
  $x_q$ and the vertex $w_q$ are separated from each other by this
  cycle. Thus, the path $P_q$ must share at least one vertex with this
  cycle, but that is impossible, since $P_q$ is disjoint from $Q_p$,
  $Q_r$ and $\gamma$. We conclude that $\rskel$ has an $\rskel$-edge
  $B_4$ connecting $u$ to a vertex $x_4$ in the interior of~$\beta$.

  We define three paths $R_1$, $R_2$ and $R_3$ of the graph $G$ as
  follows. The path $R_1$ starts in the vertex $u$, contains the edge
  $e_1=uv_1$, proceeds from $v_1$ to $w_1$ inside $B_1$, then goes
  from $w_1$ to $x_p$ inside the arc $\alpha_{12}$, then follows $Q_p$
  until it reaches the vertex~$y_p$. Similarly, the path $R_2$ starts
  in $u$, contains the edge $e_2$, follows from $v_2$ to $w_2$ inside
  $B_2$, from $w_2$ to $x_r$ inside $\alpha_{12}$, and then along
  $Q_r$ to~$y_r$. The path $R_3$ starts at the vertex $w_3$, proceeds
  towards $v_3$ inside $B_3$, then using the edge $e_3$ it reaches
  $u$, proceeds from $u$ to $x_4$ inside $B_4$, then from $x_4$ to
  $x_q$ inside $\beta$, then from $x_q$ towards $y_q$ along~$Q_q$. If
  any of the three paths $R_i$ contains more than one vertex of $C'$,
  we truncate the path so that it stops when it reaches the first
  vertex of $C'$.

  \begin{figure*}[tb]
    \centering
   
\subfigure[]{\includegraphics[page=2]{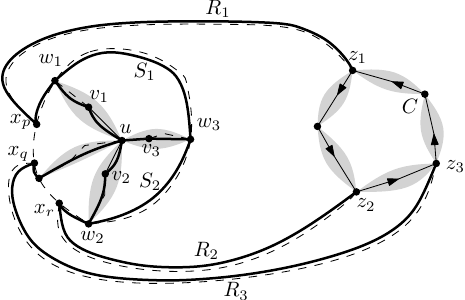}\label{fig:edge-cyc-a}}\hfil
\subfigure[]{\includegraphics[page=3]{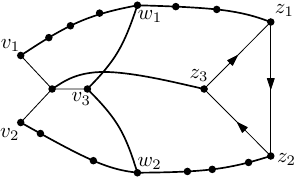}\label{fig:edge-cyc-b}}\hfil
    \caption{Illustration of Lemma~\ref{lem-edge-cyc}, the paths
      constructed in the proof~\subref{fig:edge-cyc-a} and an
      intermediate step in obtaining one of the obstructions 20, 21 or
      22~\subref{fig:edge-cyc-b}.}
    \label{fig:edge-cyc}
  \end{figure*}

  We also define two more paths $S_1$ and $S_2$ of $G$, where each
  $S_i$ connects the vertex $w_i$ to the vertex $w_3$ and projects
  into the arc~$\alpha_{i3}$; see Fig.~\ref{fig:edge-cyc-a} for an
  illustration of the constructed paths.

  Note that the three paths $R_i$ only intersect at the vertex $u$, a
  path $S_i$ may only intersect $R_j$ at one of the vertices $w_1$,
  $w_2$ or $w_3$, and the cycle $C'$ may intersect $S_i$ only in the
  vertex~$w_i$.

  Consider the \peg $(G',H',\H')$ formed by the union of the cycle
  $C'$, the three paths $R_i$, and the two paths $S_j$, where only the
  cycle $C'$ and the three edges $e_1$, $e_2$ and $e_3$ with their
  vertices have prescribed embedding, and their embedding is inherited
  from~$\H$.

  It can be easily checked that the graph $G'$ is a subdivision of a
  3-connected graph, so it has a unique edge-compatible
  embedding~$\G'$.  Consider the subgraph $\rskel'$ of $\rskel$ formed
  by all the vertices of $\rskel$ belonging to $G'$ and all the
  $\rskel$-edges that contain at least one edge of~$G'$. The graph
  $G'$ is a subdivision of~$\rskel'$. Thus, the embedding of $\rskel'$
  inherited from $\rskelemb$ must have the same rotation schemes as
  the embedding~$\G'$. Let $z_i$ be the endpoint of $R_i$ belonging
  to~$C'$.  Orient $C'$ so that $z_1,z_2,z_3$ appear in this cyclic
  order on~$C'$.  Suppose that $e_1$, $e_2$ and $e_3$ appear in this
  clockwise order in $\H$.  Then the four vertices $u$, $v_1$, $v_2$
  and $v_3$ are to the left of $C'$ in $\G'$, and hence also in
  $\rskelemb$. Since the four vertices are in an unhappy $C$-bridge of
  $\rskel$, they are to the right of $C'$ in $\H'$. This determines
  $(G',H',\H')$ uniquely.

  We now show that $(G',H',\H')$ contains one of the obstructions 20,
  21 or~22.  First, we contract each of $S_1$ and $S_2$ to a single
  edge. We also contract the cycle $C'$ to a triangle with vertices
  $z_1$, $z_2$ and~$z_3$. We contract the subpath of $R_3$ from $w_3$
  to $v_3$ to a single vertex, and we contract the subpath of $R_3$
  from $u$ to $z_3$ to a single edge.  After reversing the order of
  the vertices on the cycle to make it happy, we essentially obtain
  the \peg shown in Fig.~\ref{fig:edge-cyc-b}, except that for
  $i=1,2$ it may be that $w_i = v_i$ or $w_i = z_i$, but not both
  since $v_i \ne z_i$.  This is already very close to obstructions
  20--22.

  To contract $R_1$, we distinguish two cases. First, assume that
  $w_1$ belongs to~$C'$. This means that $z_1=w_1\neq v_1$, because we
  know that $v_1$ is not in the same component of $H$ as~$C'$. In this
  case, we contract the subpath of $R_1$ from $v_1$ to $w_1$ to a
  single edge. On the other hand, if $w_1$ does not belong to $C$, we
  contract the subpath of $R_1$ from $v_1$ to $w_1$ to a single
  vertex, and we contract the subpath from $w_1$ to $z_1$ to a single
  edge.

  The contraction of $R_2$ is analogous to the contraction of $R_1$,
  and it again depends on whether $w_2$ belongs to $C$ or not. After
  these contractions are performed, we end up with one of the three
  obstructions 20, 21, or~22.
\end{proof}

With the lemmas proven so far, we are ready to prove the following
proposition.

\begin{proposition}\label{pro-cyc}
  Let $(G,H,\H)$ be an obstruction-free \peg, with $G$ biconnected.
  Let $\rskel$ be a skeleton of an R-node of the SPQR tree of~$G$.  If
  $\rskel$ has at least one edge-compatible embedding, then it has a
  compatible embedding.
\end{proposition}
\begin{proof}
  Let $\rskelemb$ be an edge-compatible embedding.  If this embedding
  is not cycle-compatible, then $\H$ has an unhappy facial cycle $C'$
  projecting to a cycle $C$ of~$\rskel$.  The previous lemmas then
  imply that every facial cycle of $\H$ projecting to a cycle of
  $\rskel$ can only have unhappy or indifferent bridges.  Besides,
  Lemma~\ref{lem-edge-cyc} implies that no vertex $u$ of $\rskel$ can
  be incident to three $\rskel$-edges, each of them containing an edge
  of $H$ incident to~$u$.  Hence, the skeleton $\rskel$ has no
  edge-compatibility constraints.  Consequently, we may flip the
  embedding $\rskelemb$ to obtain a new embedding that is compatible.
\end{proof}

This concludes our treatment of R-nodes and thus also the proof of the
main theorem for biconnected \pegs.  We now turn to 1-connected \pegs,
i.e., \pegs that are connected but not necessarily biconnected and to
disconnected \pegs.

\section{Disconnected and 1-Connected PEGs}
\label{sec:disconnected}

We have shown that a biconnected obstruction-free \peg is planar.  We
now extend this characterization to arbitrary \pegs.  To do this, we
will first show that an obstruction-free \peg $(G,H,\H)$ is planar if
and only if each connected component of $G$ induces a planar sub-\peg.
Next, we provide a more technical argument showing that a connected
obstruction-free \peg $(G,H,\H)$ is planar, if and only if all the
elements of a certain collection of 2-connected \peg-minors of
$(G,H,\H)$ are planar.

\paragraph{Reduction to $G$ connected}
Angelini et al.~\cite{abfjk-tppeg-10} proved the following lemma.
\begin{lemma}[cf. Lemma~3.4 in~\cite{abfjk-tppeg-10}]
  Let $(G,H,\H)$ be a \peg.  Let $G_1,\dots,G_t$ be the connected
  components of~$G$.  Let $H_i$ be the subgraph of~$H$ induced by the
  vertices of $G_i$, and let $\H_i$ be $\H$ restricted to $H_i$.  Then
  $(G,H,\H)$ is planar if and only if
  \begin{enumerate}[1)]
  \item each $(G_i,H_i,\H_i)$ is planar, and
  \item for each $i$, for each facial cycle $\vec{C}$ of $H_i$ and for
    every $j \ne i$, no two vertices of $H_j$ are separated
    by~$\vec{C}$, in other words, all the vertices of $H_j$ are embedded on
    the same side of~$C$.
  \end{enumerate}
\end{lemma}
A \peg that does not satisfy the second condition of the lemma must
contain obstruction~1.  Thus, if Theorem~\ref{thm:main} holds for
\pegs with $G$ connected, it holds for all \pegs.

\paragraph*{Reduction to $G$ biconnected}
Next, we consider connected \pegs, i.e., \pegs $(G,H,\H)$ where $G$ is
connected.  In contrast to planarity of ordinary graphs, it is not in
general true that a \peg is planar if and only if each sub-\peg
induced by a biconnected component of $G$ is planar.  However, for
\pegs satisfying some additional assumptions, a similar
characterization is possible.

Let $(G,H,\H)$ be a connected \peg and let $v$ be a cut-vertex of $G$.
We say that $v$ is $H$-separating if at least two connected components
of $G-v$ contain vertices of $H$.

Let $(G,H,\H)$ be a connected \peg that avoids obstruction~1.  Let $v$
be an $H$-separating cut-vertex of $G$ that does not belong to $H$.
Let $x$ and $y$ be two vertices of~$H$ that belong to different
connected components of $G-v$, chosen in such a way that there is a
path in $G$ connecting $x$ to $y$ whose internal vertices do not
belong to~$H$.  The existence of such a path implies that $x$ and $y$
share a face~$F$ of~$\H$, otherwise $H$ would contain a cycle
separating $x$ from $y$, creating obstruction~1.  The face~$F$ is
unique, because $x$ and $y$ belong to distinct components of $H$.  It
follows that any planar embedding of $G$ that extends $\H$ must embed
the vertex $v$ in the interior of the face $F$.  We define $H' = H
\cup v$ and let $\H'$ be the embedding of $H'$ obtained from $\H$ by
inserting the isolated vertex $v$ into the interior of the face~$F$.
As shown above, any planar embedding of $G$ that extends $\H$ also
extends $\H'$.  We say that $(G,H',\H')$ is obtained from $(G,H,\H)$
by \emph{fixing} the cut-vertex $v$.

Let $(G,H^+,\H^+)$ be a \peg that is obtained from $(G,H,\H)$ by
fixing all the H-separating cut-vertices of $G$ not belonging to~$H$.
Note that each $H^+$-separating cut-vertex is also $H$-separating, and
vice versa.  A planar embedding of~$G$ that extends $\H$ also extends
$\H^+$ and in particular, $(G,H,\H)$ is planar if and only if
$(G,H^+,\H^+)$ is planar.  We now show that this operation cannot create a new
obstruction in $(G,H^+,\H^+)$.

\begin{lemma}\label{lem:fix-cutvertices}%
  Let $(G,H,\H)$ be a connected \peg that avoids obstruction~1, and
  let $(G,H^+,\H^+)$ be the \peg obtained by fixing all the
  $H$-separating cut-vertices of~$G$.  Then $(G,H,\H)$ contains a
  minimal obstruction~$X$ if and only if $(G,H^+,\H^+)$ contains $X$.
\end{lemma}
\begin{proof}
  Since $(G,H,\H)$ is a \peg-minor of $(G,H^+,\H^+)$,
  it suffices to prove that if $(G,H^+,\H^+)$
  contains an obstruction $X=(G_X,H_X,\H_X)$ then we can efficiently
  find the same obstruction in~$(G,H,\H)$.  This clearly holds in the case when
  $H_X$ does not contain isolated vertices, because then any sequence
  of deletions, contractions and relaxations that produces $X$ inside
  $(G,H^+,\H^+)$ will also produce $X$ inside $(G,H,\H)$.

  Suppose now that $H_X$ contains isolated vertices. Assume first that
  $G_X$ is 2-connected. Let $G_1,\dotsc G_t$ be the 2-connected
  blocks of $G$, let $H_i$ be the subgraph of $H$ induced by the
  vertices of $G_i$, let $\H_i$ be the embedding of $H_i$ inherited
  from $\H$, and similarly for $H^+_i$ and~$\H^+_i$. If $(G,H^+,\H^+)$
  contains $X$, then for some $i$, $(G_i,H^+_i,\H^+_i)$ contains $X$
  as well (here we use the fact that each $H^+$-separating cut-vertex
  of $G$ belongs to $H^+$). However, each $(G_i,H^+_i,\H^+_i)$ is a
  \peg-minor of $(G,H,\H)$ --- this is because any vertex $v$ of
  $H^+_i$ that is not a vertex of $H_i$ is connected to a vertex of
  $H$ by a path that internally avoids $G_i$. By contracting all such
  paths, we obtain a copy of $(G_i,H^+_i,\H^+_i)$ inside $(G,H,\H)$.
  Since $(G_i,H^+_i,\H^+_i)$ contains $X$, so does $(G,H,\H)$.

  It remains to deal with the case when $X$ is not 2-connected and $H_X$
  contains an isolated vertex. This means that $X$ is
  obstruction~1. By assumption, $(G,H,\H)$ does not contain
  obstruction~1. Suppose for contradiction that $(G,H^+,\H^+)$
  contains obstruction~1. This means that $H^+$ contains a cycle $C$
  and a pair of vertices $v$ and $w$ separated by this cycle, and that
  there exists a path $P$ of $G$ that connects $v$ and $w$ and has no
  vertex in common with~$C$.

  If $v$ is not a vertex of $H$, then $v$ is an $H$-separating
  cut-vertex.  Therefore, there are two vertices $x$ and $y$ of $H$ in
  distinct components of $G-v$ that both share a face $F$ with $v$ and
  are connected to $v$ by paths $P_x$ and $P_y$ of $G$ which do not
  contain any other vertex of~$H$. Since $x$ and $y$ are in distinct
  components of $H$, at least one of them, say $x$, does not belong to
  the cycle~$C$. Since $x$ shares a face with $v$, it must be on the
  same side of $C$ as~$v$. By the same reasoning, the vertex $w$
  either belongs to $H$ or there is a vertex $z\in H$ that appears on
  the same side of $C$ as $w$ and is connected to $w$ by a $G$-path
  $P_z$ whose internal vertices do not belong to~$H$. In any case, we
  find a pair of vertices of $H$ that are separated by $C$ and are
  connected by a $G$-path that avoids~$C$. This shows that $(G,H,\H)$
  contains obstruction~1, which is a contradiction.
\end{proof}

Lemma~\ref{lem:fix-cutvertices} shows that we can without loss of
generality restrict ourselves to \pegs $(G,H,\H)$ in which every
$H$-separating cut-vertex belongs to $H$.  For \pegs having this
property, we can show that planarity can be reduced to planarity of
biconnected components.

First, we need a definition. Let $H$ be a graph with planar embedding
$\H$, let $v$ be a vertex of $H$, and let $H_1$ and $H_2$ be two
edge-disjoint subgraphs of~$H$. We say that $H_1$ and $H_2$
\emph{alternate} around $v$ in $\H$, if there exist edges $e, e'\in
E(H_1)$ and $f,f'\in E(H_2)$ which are all incident with $v$ and
appear in the cyclic order $(e,f,e',f')$ in the rotation scheme of $v$
in the embedding~$\H$.

The following lemma is analogous to Lemma~3.3
of~\cite{abfjk-tppeg-10}, except that the assumption ``every
non-trivial $H$-bridge is local'' is replaced with the weaker
condition ``every $H$-se\-pa\-ra\-ting cut-vertex of $G$ is in~$H$''.  This
new assumption is weaker, because a separating cut-vertex not
belonging to~$H$ necessarily belongs to a non-local $H$-bridge.
However, the proof in~\cite{abfjk-tppeg-10} uses only this weaker
assumption and therefore we have the following lemma.

\begin{lemma}
  Let $(G,H,\H)$ be a connected \peg with the property that every
  $H$-separating cut-vertex of~$G$ is in $H$.  Let $G_1,\dots,G_t$ be
  the blocks of $G$, let $H_i$ be the subgraph of~$H$ induced by the
  vertices of~$G_i$ and let $\H_i$ be $\H$ restricted to $H_i$.  Then,
  $(G,H,\H)$ is planar if and only if
  \begin{enumerate}[1)]\itemsep=0ex
  \item $(G_i,H_i,\H_i)$ is a planar \peg for each $i$,
  \item no two distinct graphs $H_i$ and $H_j$ alternate around any vertex of
$\H$, and
  \item for every facial cycle $\vec{C}$ of $\H$ and for any two
    vertices $x$ and $y$ of~$\H$ separated by $\vec{C}$, any path in
    $G$ connecting $x$ and $y$ contains a vertex of $\vec{C}$.
  \end{enumerate}
\end{lemma}

Note that the last two conditions are always satisfied when $(G,H,\H)$
avoids obstructions~1 and~2.  We can also efficiently test whether the
two conditions are satisfied and produce an occurrence of an
obstruction when one of the conditions fails.  This concludes the
proof of Theorem~\ref{thm:main}.

\section{Other minor-like operations}\label{sec:other}

Let us remark that our definition of \peg-minor operations is not the only one
possible. In this paper, we preferred to work with a weaker notion of
\peg-minors, since this makes the resulting characterization theorem stronger. 
However, in many circumstances, more general minor-like operations may be
appropriate, providing a smaller set of obstructions.

For example, the $G$-edge contraction rules may be relaxed to allow
contractions in more general situations. Here is an example of such a relaxed
$G$-edge contraction rule: given a \peg $(G,H,\H)$, assume $e=uv$ is an edge of
$G$ but not of $H$, assume that $u$ and $v$ have a unique common face $F$ of
$\H$, and assume furthermore that each of the two vertices is visited only once
by the corresponding facial walk of~$F$.  If $u$ and $v$ are in distinct
components of $H$, or if the graph $H$ is connected, we embed the edge $uv$ into
$F$ and then contract it, resulting in a new \peg $(G',H',\H')$.

It is not hard to see that this relaxed contraction preserves the planarity of a
\peg, and that $\H'$ is uniquely determined. It also subsumes the `complicated
$G$-edge contraction' we introduced. With this stronger contraction rule, most
of the exceptional obstructions can be further reduced, leaving only the
obstructions 1, 2, 3, 4, 6, 11, 14, 16, and 17, as well as $K_5$ and $K_{3,3}$.
However, even this stronger contraction cannot reduce the obstructions from
$\ach{k}$.

To reduce the obstructions to a finite set, we need an operation that
can be applied to an alternating chain.  We now present an example of such an
operation.  See Fig.~\ref{fig-facesplit}.

\begin{figure}
\centering
\includegraphics{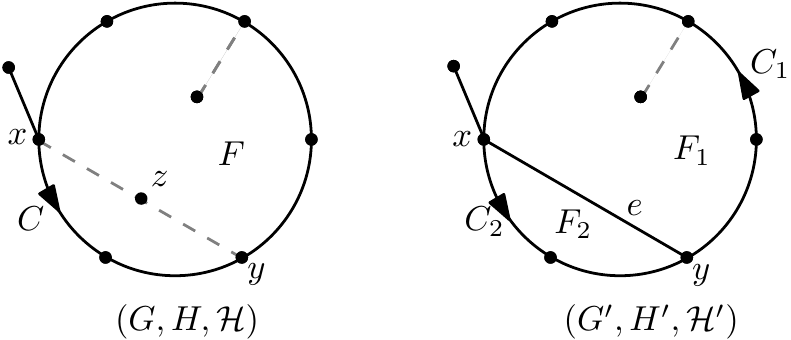}
\caption{A reduction rule transforming $(G,H,\H)$ into
$(G',H',\H')$.}
\label{fig-facesplit}
\end{figure}

Suppose that $(G,H,\H)$ is a \peg, let $F$ be a face of $\H$, let $C$ be a
facial cycle of $F$ oriented in such a way that the interior of $F$ is to
the left of $C$, let $x$ and $y$ be two vertices of $C$ that are not connected
by an edge of $G$, and let $z$ be a vertex of~$H$ not belonging to~$C$. Assume
that the following conditions hold.
\begin{enumerate}
 \item The vertex $z$ is adjacent to $x$ and to $y$ in~$G$.
 \item The vertex $z$ is embedded to the left of $C$ in the embedding~$\H$, and
is incident to the face~$F$.
 \item Any connected component of $H$ that is embedded to the left of $C$ in
$\H$ is connected to a vertex of $C\setminus\{x,y\}$ by an edge of~$G$.
 \item Any edge of $H$ that is incident to $x$ or to $y$ and does not
belong to $C$ is embedded outside of $F$ (i.e., to the right of $C$) in~$\H$.
\end{enumerate}
We define a new \peg by the following steps.
\begin{itemize}
 \item Remove vertex $z$ and all its incident edges from $G$ and~$H$.
 \item Add to $G$, $H$ and $\H$ a new edge $e=xy$. The edge $e$ is embedded
inside~$F$. (Note that the position of $e$ in the rotation schemes of $x$ and
$y$ is thus determined uniquely, because of condition 4 above.)
 \item The edge $e$ splits the face $F$ into two subfaces $F_1$ and $F_2$.  Let
$C_1$ and $C_2$ be the facial cycles of $F_1$ and $F_2$ such that $C_1\cup
C_2=C\cup\{e\}$. For any connected component $B$ of $H$ that is embedded to the
left of $C$ in $H$, let $w$ be a vertex of $C\setminus\{x,y\}$ adjacent to a
vertex of $B$. Such a vertex $w$ exists by condition 3 above. If there are
more such vertices, we choose one arbitrarily for each $B$. 
If $w$ belongs to $C_1$, then $B$ will be embedded inside $F_1$, otherwise it
will be embedded inside $F_2$.  
\end{itemize}
Let $(G',H',\H')$ be the resulting \peg. We easily see that if $(G,H,\H)$ was
planar, then $(G',H',\H')$ is planar as well. In fact, if the vertex $z$ has
degree 2 in $G$, then we may even say that $(G,H,\H)$ is planar if and only if
$(G',H',\H')$ is planar. 

The operation described above allows to reduce each $k$-fold
alternating chain with $k\ge 4$ to a smaller non-planar \peg which
contains a $(k-1)$-fold alternating chain. It also reduces obstruction
4 to obstruction 3, and obstruction 16 to a \peg that contains
obstruction 1. Therefore, when the above operation is added to the
permissible minor operations, there will only be a finite number of
minimal non-planar \pegs.  More precisely, exactly nine minimal
non-planar \pegs remain in this case.

Let us point out that the obstructions from the infinite family
$\bigcup_{k\ge 4} \ach{k}$ only play a role when cycle-compatibility is
important.  For certain types of \pegs, cycle-compatibility is not a concern.
For instance, if the graph $H$ is connected, it can be shown that $(G,H,\H)$ is
planar if and only if all the skeletons of $G$ have edge-compatible embeddings,
and therefore such a $\peg$ is planar if and only if it avoids the finitely many
exceptional obstructions.

\section{Conclusion}\label{sec:conclusion}

Note that Theorem~\ref{thm:main} together with the linear-time
algorithm for testing planarity of a \peg~\cite{abfjk-tppeg-10}
immediately implies Theorem~\ref{thm:alg}.
In any non-planar instance $I=(G,H,\H)$ only linearly many \peg-minor operations
are possible.  We test each one individually and use the linear-time
testing algorithm to check whether the result is non-planar.  In this
way we either find a smaller non-planar \peg $I'$ resulting from $I$
by one of the operations, or we have found an obstruction,
which by Theorem~\ref{thm:main} is contained in our list.  The running
time of this algorithm is at most~$O(n^3)$.

In fact, in many cases, as indicated in the paper, obstructions can be
found much more efficiently, often in linear time.  In particular, the
linear-time testing algorithm gives an indication of which property of
planar \pegs is violated for a given instance.  
Is it possible to find an obstruction in a non-planar \peg in linear
time?  In general, given a fixed \peg $(G,H,\H)$, what is the complexity of
determining whether a given \peg contains $(G,H,\H)$ as \peg-minor?  The answer
here may depend on the \peg-minor operations we allow.

It is not known whether the results on planar \pegs can be generalized to
graphs that have a partial embedding on a higher-genus surface.  In fact, even
the complexity of recognizing whether a graph partially embedded on a fixed
higher-genus surface admits a crossing-free embedding extension is still an open
problem.

%\section*{References}
\bibliographystyle{plain}
\bibliography{pepbib}

\end{document}